\documentclass[a4paper,margin=1in,fontsize=12pt]{article}
\usepackage[english]{babel} 
\usepackage{amsmath,amsfonts,amsthm} 
\usepackage[utf8]{inputenc}
\usepackage{float}
\usepackage{lipsum} 
\usepackage{blindtext}
\usepackage{graphicx} 
\usepackage{caption}
\usepackage[sc]{mathpazo} 
\usepackage[T1]{fontenc} 
\usepackage{microtype} 
\usepackage[margin=1in]{geometry}

\usepackage{multicol} 
\usepackage{booktabs} 
\usepackage{float} 
\usepackage{hyperref} 
\usepackage{lettrine} 
\usepackage{paralist} 
\usepackage{abstract} 
\usepackage{titlesec} 
\usepackage{natbib, url}
\usepackage[caption=false]{subfig}   
\usepackage{amssymb}
\usepackage{comment}
\usepackage{xcolor}
\usepackage{tikz}
\usepackage{algorithm}
\usepackage{algpseudocode}

\usepackage{setspace}
\usetikzlibrary{shapes.geometric, arrows}

\tikzstyle{startstop} = [rectangle, rounded corners, minimum width=3cm, minimum height=1cm,text centered, draw=black, fill=red!30]
\tikzstyle{process} = [rectangle, minimum width=3cm, minimum height=0.8cm, text centered, draw=black, fill=orange!30]
\tikzstyle{rounded_process} = [rectangle, rounded corners, minimum width=3cm, minimum height=0.8cm, text centered, draw=black, fill=orange!30]
\tikzstyle{rounded_process_highlight} = [rectangle, rounded corners, minimum width=3cm, minimum height=0.8cm, text centered, draw=black, thick, fill=red!30]
\tikzstyle{arrow} = [thick,->,>=stealth]

\titleformat{\section}[block]{\large\scshape\centering}{\thesection.}{1em}{} 
\titleformat{\subsection}[block]{\large}{\thesubsection.}{1em}{} 
\usepackage{fancyhdr} 
\usepackage[bottom]{footmisc}

\newtheorem{theorem}{Theorem}[section]
\newtheorem{lemma}[theorem]{Lemma}

\newtheorem{corollary}[theorem]{Corollary}

\numberwithin{equation}{section}
\newfloat{submodule}{tbp}{los}[section]
\floatname{submodule}{Submodule}

\newcommand{\given}{\,|\,}
\newcommand{\T}{\top}

\newcommand{\bs}{\mathbf{s}}
\newcommand{\bS}{\mathbf{S}}
\newcommand{\calS}{\mathcal{S}}

\newcommand{\bY}{\mathbf{Y}}
\newcommand{\by}{\mathbf{y}}
\newcommand{\bX}{\mathbf{X}}
\newcommand{\bx}{\mathbf{x}}

\newcommand{\bz}{\mathbf{z}}

\newcommand{\bV}{\mathbf{V}}
\newcommand{\bA}{\mathbf{A}}

\newcommand{\bF}{\mathbf{F}}
\newcommand{\fb}{\mathbf{f}}

\newcommand{\bu}{\mathbf{u}}
\newcommand{\bU}{\mathbf{U}}

\newcommand{\bM}{\mathbf{M}}

\newcommand{\calH}{{\cal H}}
\newcommand{\bI}{\mathbf{I}}
\newcommand{\bD}{\mathbf{D}}
\newcommand{\bL}{\mathbf{L}}

\newcommand{\bv}{\mathbf{v}}

\newcommand{\bmu}{\boldsymbol{\mu}}
\newcommand{\bSigma}{\boldsymbol{\Sigma}}
\newcommand{\beps}{\boldsymbol{\epsilon}}
\newcommand{\bpsi}{\boldsymbol{\psi}}
\newcommand{\bbeta}{\boldsymbol{\beta}}
\newcommand{\bLambda}{\boldsymbol{\Lambda}}

\newcommand{\bgamma}{\boldsymbol{\gamma}}
\newcommand{\etab}{\boldsymbol{\eta}}
\newcommand{\bPsi}{\boldsymbol{\Psi}}
\newcommand{\brho}{\boldsymbol{\rho}}

\newcommand{\blambda}{\boldsymbol{\lambda}}
\newcommand{\bzero}{\mathbf{0}}

\graphicspath{{./pics/}}

\title{\vspace{-15mm}\fontsize{24pt}{10pt}\selectfont\textbf{
Projected Bayesian Spatial Factor Models
}} 
\date{}

\author{
\large
{\textsc{Lu Zhang}}\\[2mm]
{\textsc{USC PPHS Division of Biostatistics}}\\[2mm]
\normalsize \href{mailto:lzhang63@use.edu}{lzhang63@use.edu}\\[2mm] 
}
\providecommand{\keywords}[1]{\textbf{\textit{Key words:}} #1}
\begin{document}
\maketitle 
\thispagestyle{fancy} 

\label{firstpage}

\begin{abstract}
Factor models balance flexibility, identifiability, and computational efficiency, with Bayesian spatial factor models particularly prone to identifiability challenges and scaling limitations. This work introduces Projected Bayesian Spatial Factor (PBSF) models, a new class of models designed to achieve scalability and robust identifiability for spatial factor analysis. PBSF models are defined through a novel Markov chain Monte Carlo construction, Projected MCMC (ProjMC$^2$), which leverages conditional conjugacy and projection to improve posterior stability and mixing by constraining factor sampling to a scaled Stiefel manifold. Theoretical results establish convergence of ProjMC$^2$ irrespective of initialisation. By integrating scalable univariate spatial modelling, PBSF provides a flexible and interpretable framework for low-dimensional spatial representation learning of massive spatial data. Simulation studies demonstrate substantial efficiency and robustness gains, and an application to human kidney spatial transcriptomics data highlights the practical utility of the proposed methodology for improving interpretability in spatial omics.
\end{abstract}

\newpage

\section{Introduction}\label{Intro}

Multivariate geo-indexed datasets are increasingly prevalent across disciplines, from environmental science to genomics. These datasets measure multiple variables across shared spatial domains, often exhibiting spatial autocorrelation—where nearby locations have similar values—and spatial cross-correlation—where variables demonstrate spatial interdependencies \citep{banerjee2014hierarchical, cressie2015statistics}. Such correlations typically reflect geographic proximity and common underlying drivers. Analysing these joint spatial patterns helps reveal latent dependencies and potential causal mechanisms. For instance, spatial omics involves quantifying various molecular measurements (e.g., gene expression, protein abundances, metabolites) within biological tissues, where patterns reflect the spatially structured biological microenvironment \citep{moses2022museum, bressan2023dawn, lee2025spatial}. Identifying spatially correlated molecular signatures can enhance understanding of tissue heterogeneity.  
Similarly, multivariate pollutant data in environmental science can indicate diverse sources, such as vehicular emissions or industrial discharge \citep{song2018multivariate, dai2024spatial}. Uncovering these latent spatial structures is crucial for source identification and risk assessment.

Spatial factor analysis offers a valuable framework for exploring underlying spatially correlated factors within multivariate datasets. Analogous to classical latent factor analysis, spatial factor models represent a latent multivariate spatial process as a linear combination of a small number of common spatially correlated factors. The foundational concepts and subsequent elaborations of spatial factor models have been explored in key works by \citet{wang2003generalized}, \cite{Lopes2008} and \cite{velten2022identifying}. Despite their established utility, a significant challenge in applying spatial factor models, especially those employing Gaussian Process (GP) distributions for the spatial factors, is their computational burden. The computational costs and storage requirements scale cubically and quadratically, respectively, with the number of spatial locations n, rendering them impractical for large-scale datasets \citep{stein99, banerjee2014hierarchical, cressie2015statistics}. To mitigate this scalability bottleneck, several approaches have focused on low-rank approximations of the spatial factors. These include methods based on Predictive Processes \citep{ban08, Finley_etall_2009} or Inducing points \citep{titsias2009variational}, as explored in studies by \cite{renbanerjee2013}, \cite{shang2022spatially}, and \cite{townes2023nonnegative}. Nevertheless, such approaches that employ a reduced rank representation of the desired spatial process cannot scale to modest large datasets (with tens of thousands of locations) and can result in over-smoothing the latent process from massive data sets \citep{stein2014}. 

More recently, the use of Nearest-Neighbor Gaussian Processes (NNGP) \citep{datta16}, which fall under the broader class of Vecchia approximations for scalable GP \citep[see, e.g.][]{katzfuss2021general}, have emerged as a promising alternative for modelling spatial factors, gaining traction in recent literature \citep{taylor2019spatial, zhang2022spatial}. NNGPs offer a scalable approximation to full GPs that can capture both global and local spatial pattern for large datasets \citep{zhang2024fixed}. However, while \citet{zhang2022spatial} emphasize predictive performance, their work does not address the issue of parameter identifiability, which is essential for recovering underlying spatial patterns. In contrast, \citet{taylor2019spatial} adopt constraints 
common in non-spatial factor analysis to achieve identifiability, which is unnecessarily stringent as argued in \citet{renbanerjee2013}.
Furthermore, a critical gap persists in the literature, including \citet{renbanerjee2013} and \citet{taylor2019spatial}, concerning the practical implications of weak identifiability. Specifically, the impact of identifiability issues related to intercept terms and GP hyperparameters (challenges well-documented in the general GP literature, e.g., \cite{stein99}, \cite{Zhang04},  
\cite{zhang2005towards}, \cite{DZM09}, \cite{tang2021identifiability}) on the implementation and posterior sampling efficiency of spatial factor models remains largely unaddressed.

This paper introduces a novel and flexible Bayesian spatial factor model designed for scalability with large spatial datasets, while imposing minimal yet effective restrictions to mitigate the slow convergence and poor mixing engendered by the inherent identifiability issues in factor models. A key contribution is a new model construction and sampling algorithm, termed Projected Markov Chain Monte Carlo (ProjMC\textsuperscript{2}). Specifically, within a blocked Gibbs sampling framework for a spatial factor model, a projection step is introduced to project the high-dimensional factor realizations to a subspace of the Stiefel manifold \citep{chakraborty2019statistics}, thereby enhancing identifiability and substantially improving sampling efficiency for parameters subject to weak or non-identifiability. Theoretical results establishing the existence, convergence, and properties of the posterior distribution induced by the ProjMC\textsuperscript{2} algorithm are presented in Section~\ref{sec: ProjMC2}. The induced model is termed the Projected Bayesian Spatial Factor (PBSF) model. Section~\ref{sec: model_implement} develops scalable variants of PBSF models in which spatially varying factors are modelled using NNGPs. These extensions address diverse inferential settings, including hyperparameter estimation and the accommodation of missing data, and detailed implementation algorithms are provided therein.

To illustrate the practical utility and advantages of the PBSF model,
this work features an application to a human kidney spatial transcriptomics dataset. Section~\ref{sec: real_data_analy} presents a comparative analysis against contemporary `black-box' auto-encoder-based algorithms, STAGATE \citep{dong2022deciphering} and GraphST \citep{long2023spatially}, specifically for the task of spatial domain identification. The results demonstrate that the proposed methodology achieves performance comparable to, and in some instances superior to, these leading competitors. More critically, it yields more interpretable insights into the underlying complex biological systems, an advantage of particular significance when analysing novel spatial omics data that lack pre-existing annotations or labels. Consequently, the proposed methodology provides a scalable, fully Bayesian, probability-based regression framework for large-scale spatial omics data, offering enhanced interpretability and uncertainty quantification.

The rest is organised as follows. Section~\ref{sec:BSF_model} outlines notation and background on Bayesian spatial factor models. Section~\ref{sec: simulation} presents simulation studies with sensitivity analyses of algorithmic choices, parameter identifiability, and practical performance.
Section~\ref{sec: summary} concludes with remarks and future research directions. 

\section{Bayesian Spatial factor model and sampling}\label{sec:BSF_model}
\textbf{Bayesian Spatial Factor (BSF) Model: }
Let $\by(\bs) = (y_1(\bs), \ldots, y_q(\bs))^\top \in \mathbb{R}^q$ denote the $q\times 1$ vector of dependent outcomes in location $\bs \in \mathcal{D} \subset \mathbb{R}^d$, $\bx(\bs) = (x_1(\bs), \ldots, x_p(\bs))^\top \in \mathbb{R}^p$ be the corresponding explanatory variables, and $\bbeta$ be a $p \times q$ regression coefficient matrix. A factor model can be denoted as
\begin{equation}\label{eq: spatial_factor}
    \by(\bs) = \bbeta^\top \bx(\bs) + \bLambda^\top \fb(\bs)+ \beps(\bs)\; , \; \bs \in \mathcal{D}\; ,
\end{equation}
where $\bLambda$ is a $K \times q$ loading matrix and $\fb(\bs)=(f_1(\bs), \ldots, f_K(\bs))'$ denotes a vector of $K$ components. Each component $f_k(\bs)$ is the $k$-th factor's realization at location $\bs$. $\blambda_k^{\top}$ is the $k$-th row of $\bLambda$.  The noise process $\beps(\bs) \overset{iid}{\sim} \mathrm{N}(\mathbf{0}, \bSigma)$ with covariance matrix $\bSigma$. When assigning the factors $\fb(\bs)$ with a prior with belief in spatial correlation, the hierarchical model \eqref{eq: spatial_factor} is referred to as a spatial factor model. Typical prior choices include GPs. Here, it is further assumed that each factor $f_k(\bs)$ independently follows a GP with correlation function $\rho_{\psi_k}(\cdot, \cdot)$ with hyperparameters $\psi_k$, i.e.,
\begin{equation}\label{eq:GPprior}
  f_k(\bs)\sim \mbox{GP}(0,\; \rho_{\psi_k}(\cdot, \cdot))\;, \; k = 1, \ldots, K \;.
\end{equation}
As illustrated in \citet{zhang2022spatial}, a convenient choice for the priors of parameters $\{\bbeta, \bLambda, \bSigma\}$ 
is the Matrix-Normal-Inverse-Wishart(MNIW) family, i.e.,
\begin{equation}\label{eq: LMC_priors}
    \begin{aligned}
    \bbeta \given \bSigma &\sim \mbox{MN}(\bmu_{\bbeta}, \bV_{\bbeta}, \bSigma)\; ; \;
    \bLambda \given \bSigma \sim \mbox{MN}(\bmu_{\bLambda}, \bV_{\bLambda}, \bSigma)\; ; \; \bSigma \sim \mbox{IW}(\bPsi, \nu)
    \end{aligned}\;,
\end{equation}
where $\mbox{MN}_{n,p}(\mathbf{M}, \bU, \bV)$ denotes a Matrix-Normal distribution \citep{dawid1981} with mean matrix $\bM$, the first $n \times n$ scale matrix $\bU$, and the second $p \times p$ scale matrix $\bV$, $\bmu_{\bLambda}$ is a $q \times K$ matrix and $\bV_{\bLambda}$ is a $K \times K$ positive definite matrix. 
Denote $\bpsi = \{\psi_k\}_{k = 1}^K$, 
and let $p(\bpsi)$ be an arbitrary proper prior.
The model defined by
\eqref{eq: spatial_factor}--\eqref{eq: LMC_priors}, together with the prior $p(\bpsi)$, is referred to as the Bayesian spatial factor (BSF) model. 

Without misalignment, 
model in \eqref{eq: spatial_factor} 
can be cast as $\bY = \bX\bbeta + \bF \bLambda + \beps$,
where $\bY = \by(\mathcal{S}) = [\by(\bs_1): \cdots : \by(\bs_n)]^\top$ is the $n \times q$ response matrix, $\bX = \bx(\mathcal{S}) = [\bx(\bs_1) : \cdots : \bx(\bs_n)]^\top$ is the 
design matrix with full rank ($n > p$), $\beps = \beps(\calS) = [\beps(\bs_1): \cdots : \beps(\bs_n)]^\top$, and 
$\bF$ is the $n\times K$ matrix with $j$-th column being the $n\times 1$ vector comprising $f_j(\bs_i)$'s for $i=1,2,\ldots,n$.
Define $\brho_{\psi_k}(\calS, \calS)$ to be the $n\times n$ spatial correlation matrix for $\fb_k = (f_k(\bs_1), f_k(\bs_2), \ldots, f_k(\bs_n))^{\top}$. 
Denote matrix vectorization by stacking its columns, i.e., $\mbox{vec}(\bF) = (f_1(\bs_1), \ldots, f_1(\bs_n), \ldots, f_K(\bs_1), \ldots, f_K(\bs_n))^\top$, we have $\mbox{vec}(\bF) \sim \mathrm{N}(\mathbf{0},\, \oplus_{k=1}^K \{\brho_{\psi_k}(\calS, \calS)\})$, and $\oplus_{i=1}^n$ represents the block diagonal operator stacking matrices along the diagonal. 

\textbf{Conditional Posterior of BSF Model and Posterior Sampling: } The BSF model yields conjugate conditional posteriors, enabling an efficient Gibbs sampling with block updates. When $\bF$ is fixed, under the MNIW prior \eqref{eq: LMC_priors} on $\{\bbeta, \bLambda, \bSigma\}$, the posterior for $\bgamma = [\bbeta^\top, \bLambda^\top]^\top$ and $\bSigma$ remains in the MNIW family. Specifically,
\begin{equation}\label{eq: SLMC_MNIW}
(\bgamma, \bSigma) \mid (\bF, \bY) \sim \mathrm{MNIW}(\bmu^\ast, \bV^\ast, \bPsi^\ast, \nu^\ast),
\end{equation}
where $\bV^\ast = [\bX^{\ast\top}\bX^{\ast}]^{-1}$, $\bmu^\ast = \bV^\ast[\bX^{\ast\top}\bY^{\ast}]$, $\bPsi^\ast = \bPsi + \bS^{\ast}$, $\nu^\ast = \nu + n$, and $\bS^{\ast} = (\bY^{\ast} - \bX^{\ast}\bmu^\ast)^{\top}(\bY^{\ast} - \bX^{\ast}\bmu^\ast)$. Here $\bX^{\ast}$ and $\bY^{\ast}$ arise from the argumented linear system
\begin{equation}\label{eq: augment_linear_LMC}
\begin{array}{c}
\underbrace{ \left[ \begin{array}{c} \bY\\ \bL_{\bbeta}^{-1} \bmu_{\bbeta} \\ \bL_{\bLambda}^{-1} \bmu_{\bLambda} \end{array} \right]}_{\bY^\ast}
= \underbrace{ \left[ \begin{array}{cc} \bX& \bF \\ \bL_{\bbeta}^{-1}& \mathbf{0} \\  \mathbf{0}& \bL_{\bLambda}^{-1} \end{array} \right] }_{\bX^{\ast}}  \underbrace{\left[ \begin{array}{c} \bbeta \\ \bLambda \end{array} \right]}_{\bgamma}+
\underbrace{ \left[ \begin{array}{c} \etab_1 \\ \etab_2 \\ \etab_3 \end{array} \right]}_{\etab^\ast}\; ,
\end{array}
\end{equation}
where $\bV_\Lambda = \bL_\Lambda\bL_\Lambda^\T$, and $\etab^\ast \sim \mbox{MN}(\mathbf{0}_{(n + p + K) \times q}, \bI_{n + p + K}, \bSigma)$.
In particular, when $\bSigma = \oplus_{i=1}^q \{\sigma_i^2\}$ with $\sigma_i^2 \sim \mathrm{IG}(a_i, b_i)$ prior, the conditional posterior of $(\bgamma, \bSigma) \given \bF, \bY$ is available in closed form even with misaligned data. Moreover, different mean $\bmu_{\bbeta_i}, \bmu_{\Lambda_i}$ and covariance matrix $\sigma_i^2\bV_{\bbeta_i}, \sigma_i^2\bV_{\bLambda_i}$ are allowed for each $\bgamma_i$, where $\bgamma = [\bgamma_1 : \cdots, \bgamma_q]$. Let $\bY_i$ denote the $i$-th outcome observed on the set of observed locations $\calS_i$. Reuse the note here but define $\bY^\ast_i$ and $\bX^\ast_i$ with corresponding $i$-th outcome, design matrix, $\bF$ on $\calS_i$ and prior parameters for $\bgamma_i$. Through linear algorithm we obtain that 
\begin{equation}\label{eq: SLMC_NIG_misaligned}
(\bgamma_i, \sigma_i^2) \mid (\bF(\calS_i), \bY_i) \sim \mathrm{NIG}(\bmu_i^\ast, \bV_i^\ast, a^\ast_i, b^\ast_i),
\end{equation}
where $\bV_i^\ast = [\bX_i^{\ast\top}\bX_i^{\ast}]^{-1}$, $\bmu_i^\ast = \bV_i^\ast[\bX_i^{\ast\top}\bY_i^{\ast}]$, $a_i^\ast = a_i + n_i/2$, $b_i^\ast = b + 0.5\bS_i^{\ast}$, and $\bS_i^{\ast} = (\bY_i^{\ast} - \bX_i^{\ast}\bmu_i^\ast)^{\top}(\bY_i^{\ast} - \bX_i^{\ast}\bmu_i^\ast)$. Next, given $\{\bbeta, \bLambda, \bSigma, \bpsi \}$ and let $\brho_{\psi_k}(\calS, \calS) = \bL_k\bL_k^\top$, 
\begin{equation}\label{eq: SLMC_F_cond_post}
    \mbox{vec}(\bF) \given (\bgamma, \bSigma, \bpsi, \bY) \sim \mathrm{N}((\Tilde{\bX}^\T\Tilde{\bX})^{-1}\Tilde{\bX}^\T\Tilde{\bY}, \, (\Tilde{\bX}^\T\Tilde{\bX})^{-1})\;,
\end{equation}
where $\Tilde{\bY} = [\mbox{vec}((\bY - \bX\bbeta)\bSigma^{-1/2})^\top, 0_{n*k}^\top]^\top$ and $\Tilde{\bX} = [(\bLambda * \bSigma^{-1/2}) \otimes I_{n}:  \oplus_{k = 1}^K \bL_{k}^{-\top}]^\top$. Meanwhile, since hyperparameter $\psi_k$ is independent of $\bY$ conditional on $\bF_k$, full conditional distribution of $\psi_k$ depends only on $\bF_k$. These conditional conjugacy naturally lead to a Metropolis/Slice-within-Gibbs sampler. Note that \eqref{eq: SLMC_F_cond_post} also extends to misaligned observations. Technique details are omitted here as this is elaborated in \citet[][see Eq.~(7)-(8)]{zhang2022spatial}.

\textbf{Non-identifiability issue:} Although the blocked Gibbs sampler performs well for interpolation and prediction tasks as shown in \citet{zhang2022spatial}, it is not efficient for factor analysis due to inherent non-identifiability issues. The non-identifiability problem arises firstly from the construction $\bLambda^\top \fb(\bs)$ in the target factor model \eqref{eq: spatial_factor}, where scaling and permuting rows of $\bLambda$ and the corresponding elements in $\fb(\bs)$ can yield the same likelihood. This degeneracy persists regardless of sample size. In practise, this will cause poor convergence and mixing of the Markov chain Monte Carlo (MCMC) chains for $\bLambda$ and $\fb(\bs)$, though the MCMC chain for the product $\bLambda^\top \fb(\bs)$ may converge faster. Additionally, since only samples for $\bLambda^\top \fb(\bs)$ are reliable, we need to recover the posterior samples of $\bLambda^\top \fb(\bs)$. When $q$ is large and $k$ is small, the storage of the sample for $\bLambda^\top \fb(\bs)$ can be expensive. Moreover, spatial factors $\fb(\bs)$ and the intercepts for each outcome may not be identifiable from one another \citep{stein99}, further complicating inference. When each factor is modelled via a Gaussian process with a Mat\'ern covariance function \citep{Matern86}, additional non-identifiability arises among the hyperparameters $\bpsi$ \citep{Zhang04, 
zhang2005towards, DZM09, tang2021identifiability}. Consequently, the block-update approach converges slowly and exhibits mixing difficulties for posterior inference in factor analysis settings.

\section{Projected Bayesian spatial factor (PBSF) model}\label{sec: ProjMC2}

This section introduces a Projected Markov Chain Monte Carlo (ProjMC$^2$) algorithm specifically tailored to the BSF model. The core motivation is to improve sampling efficiency for factor analysis by reducing redundant or non-identifiable directions within the latent factor space. The resulting posterior distribution is referred to as the projected Bayesian spatial factor (PBSF) model. 

\subsection{ProjMC$^2$ sampling for BSF Model}\label{subsec:BSF_PMCMC}
Let $\bF\in\mathbb{R}^{n\times K}$ be the latent factor matrix in the BSF model. Define the projection
\begin{equation}\label{eq: proj_g}
  g\colon \mathbb{R}^{n\times K} \;\to\; \Omega^g,
  \quad
  \bF \;=\; [\fb_1:\cdots:\fb_K]
  \;\mapsto\;
  \tilde{\bF}
  \;=\;
  \sqrt{n-1}\,\mathrm{QR}\!\bigl\{(I_n - \tfrac{1}{n}1_n1_n^\top)\bF\bigr\},
\end{equation}
where $\mathrm{QR}\{A\}$ produces the $Q$-factor of the thin QR decomposition of a matrix $A$, $1_n$ is the $n$-dimensional vector of all ones, $I_n$ is the $n\times n$ identity matrix, and $\Omega^g \subset \mathbb{R}^{n \times K}$ denote the image (range) of $g(\cdot)$. This map $g(\bF)$ first \emph{centres} the columns of $\bF$ (by subtracting their sample means) and then \emph{projects} the centred matrix onto a scaled Stiefel manifold \citep{chakraborty2019statistics}.  Its inverse set is given by

\begin{equation}\label{eq:inv_g}
  g^{-1}(\tilde{\bF})
  \;=\;
  \{\bF : g(\bF) = \tilde{\bF}\}
  \;=\;
  \bigl\{\tilde{\bF}\,R \;+\; 1_n\,\mu_f^\top 
    \,\mid\, 
    R \in U_K^+,\, \mu_f \in \mathbb{R}^K\bigr\},
\end{equation}
where $U_K^+$ is the space of $K\times K$ upper-triangular matrices with positive diagonal entries. 
Let $\mathcal{L}^{n\times K}$ be the Lebesgue measure on $\mathbb{R}^{n\times K}$.  We write $\varphi = \mathcal{L}^{n\times K} \circ g^{-1}$ for the induced (pushforward) measure of $\mathcal{L}^{n\times K}$ onto $\Omega^g$. 

Using $g$, an MCMC procedure is constructed based on block Gibbs updates for the BSF model. A simplified flowchart is shown in Figure~\ref{fig:PMCMC_chart}. 
\begin{figure}[h!]
\centering
\small
\begin{tikzpicture}[node distance=5.8cm]
  \tikzstyle{arrow} = [thick,->,>=stealth]
  \node (step1) [rounded_process] {\small 1.\;Sample \(\bgamma,\bSigma\) by \eqref{eq: SLMC_MNIW}};
  \node (step2) [rounded_process, right of=step1] {\small 2.\;Sample \(\bF\) by \eqref{eq: SLMC_F_cond_post} and Sample $\bpsi$};
  \node (step3) [rounded_process_highlight, right of=step2] {\small 3.\;Project \(\bF\;\mapsto\;\tilde{\bF} = g(\bF)\)};
  
  \draw [arrow] (step1) -- (step2);
  \draw [arrow] (step2) -- (step3);
  \draw [arrow] (step3.south) -- ++(0,-0.3) -| (step1.south);
\end{tikzpicture}
\caption{\centering Simplified flowchart of ProjMC$^2$ algorithm.}
\label{fig:PMCMC_chart}
\end{figure}
ProjMC$^2$ forces the sampling space of $\{\bF, \bpsi, \bgamma, \bSigma\}$ on $\Omega^g \;\times\; \calH^K \;\times\;  \mathbb{R}^{(p+K)\times q} \;\times\; \mathbb{S}_{+}^q$, where $\calH$ is the support of $\psi_k$, and $\mathbb{S}_{+}^q$ is the space of $q\times q$ symmetric positive-definite matrices. We further assume that $\calH$ is compact.  
Define $\boldsymbol{\Psi}$ as the pushforward of the product measure $\mathcal{L}^{n\times K}\times \nu_{\calH^K} \times \mathcal{L}^{(p+K)\times q}\times \nu_{\mathbb{S}_{+}^q}$ through the map $(\bF, \, \bpsi, \,\bgamma,\,\bSigma)
  \;\mapsto\;
  \bigl(g(\bF), \, \bpsi, \, \bgamma, \, \bSigma\bigr)$,
where $\nu_{\calH^K}$ is a base measure on $\calH^K$ and $\nu_{\mathbb{S}_{+}^q}$ is a base measure 
on \(\mathbb{S}_{+}^q\).  Equivalently, one may write
$\boldsymbol{\Psi}
  \;=\;
  \varphi
  \;\times\;
  \nu_{\calH^K}
  \;\times\;
  \mathcal{L}^{(p+K)\times q}
  \;\times\;
  \nu_{\mathbb{S}_{++}^q}$.
A quantitative description of the sampling process follows. Suppose there are current draws $\{\tilde{\bF}^{(l)}, \bpsi^{(l)},  \bgamma^{(l)}, \bSigma^{(l)}\}$ at the $l$-th iteration. The steps for iteration $l+1$ proceed as follows:
\begin{enumerate}
\item[(i) ] Sample $(\bgamma^{(l+1)}, \bSigma^{(l+1)})$ using the conditional distribution given in \eqref{eq: SLMC_MNIW} with $\bF = \tilde{\bF}^{(l)}$.
\item[(ii) ] Update $\tilde{\bF}^{(l+1)}$ and $\bpsi^{(l+1)}$ from the distribution with density
\[
  \int_{g^{-1}\{\tilde{\bF}^{(l+1)}\}} p\bigl(\bpsi^{(l+1)} \mid \bF, \bgamma^{(l+1)}, \bSigma^{(l+1)}, \bY \bigr)
    p\bigl(\bF \,\mid\, \bgamma^{(l+1)}, \bSigma^{(l+1)}, \bpsi^{(l)}, \bY \bigr)
  \,d\bF,
\]
where $p\bigl(\bF \,\mid\, \bgamma^{(l+1)}, \bSigma^{(l+1)}, \bpsi^{(l+1)}, \bY\bigr)$ 
is the conditional posterior in \eqref{eq: SLMC_F_cond_post}, and $p\bigl(\bpsi^{(l+1)} \mid \bF, \bgamma^{(l+1)}, \bSigma^{(l+1)}, \bY \bigr)$ the full conditional posterior of $\bpsi$ in the BSF model.
\end{enumerate}
\begin{lemma}[Transition Kernel of ProjMC$^2$ for the BSF model] 
\label{thm:Projected-MCMC-kernel}
Let $\theta_1 = \{\tilde{\bF}_1, \bpsi_1, \bgamma_1, \bSigma_1\}$, $\theta_2 = \{\tilde{\bF}_2,  \bpsi_2, \bgamma_2, \bSigma_2\}$ be two points in the sample space $\Theta := \Omega^g 
  \;\times\;
  \calH^K
  \;\times\; 
  \mathbb{R}^{(p+K)\times q} 
  \;\times\; 
  \mathbb{S}_{++}^q$, 
the transition kernel of the projected MCMC for the BSF model is
\begin{equation}
\label{eq: transit_kernel}
\small
  K(\theta_1,\theta_2)
  = \int_{g^{-1}(\tilde{\bF}_2)} 
        p\bigl(\bpsi_2 \mid \bF, \bgamma_2, \bSigma_2, \bY \bigr) p\bigl(\bF \mid \bgamma_2,\bSigma_2, \bpsi_1, \bY\bigr)
      d\bF \; p\bigl(\bgamma_2,\bSigma_2 \mid \tilde{\bF}_1,\bY\bigr),
\end{equation}
where densities $p\bigl(\bF \,\mid\, \bgamma,\bSigma, \bpsi, \bY\bigr)$ and $p\bigl(\bgamma,\bSigma \,\mid\, \bF,\bY\bigr)$ are given in \eqref{eq: SLMC_MNIW} and \eqref{eq: SLMC_F_cond_post}, respectively. The distribution $p(\bpsi \mid \bF, \bgamma, \bSigma, \bY)$ is the full conditional posterior of $\bpsi$ in the BSF model, identifiable up to a proportionality constant as $p\bigl(\bF \mid \bpsi \bigr) \times p(\bpsi)$, where $p\bigl(\bF \mid \bpsi \bigr)$ is given in \eqref{eq:GPprior} and $p\bigl(\bpsi \bigr)$ denotes the prior on the hyperparameter set $\bpsi$.
\end{lemma}

\subsection{Theoretical Properties of the ProjMC$^2$ Algorithm}
In this subsection, it is formally demonstrated that the Markov chain \((\theta_l)\) defined by the transition kernel \eqref{eq: transit_kernel} admits a unique stationary distribution, and that, for any initial point in its state space, the chain converges to this target density. These results ensure the validity and consistency of our ProjMC$^2$ approach for the BSF model.

\begin{theorem}[Convergence]\label{thm:converge}
    Let \((\theta_\ell)\) be the Markov chain on \(\Theta\) with transition kernel \(K\) given in \eqref{eq: transit_kernel}, \((\theta_\ell)\) converges in total variation to its unique stationary distribution $\pi(\cdot)$.  Hence, for any initial state $\theta \in \Theta$, 
$\lim_{l \to \infty}
    \bigl\|
      K^l(\theta,\cdot) \;-\; \pi(\cdot)
    \bigr\|_{\mathrm{TV}}
  \;=\;
  0$, where for two probability measures $\mu_1$ and $\mu_2$ on $\Theta$, the total variation distance is defined as
$\|\mu_1 - \mu_2\|_{\mathrm{TV}}
  \;=\;
  \sup_{A \subseteq \Theta}
    \bigl|\mu_1(A) \;-\; \mu_2(A)\bigr|.$
\end{theorem}

\noindent The proof proceeds via four lemmas establishing irreducibility, aperiodicity, recurrence, and the existence of an invariant finite measure for the Markov chain $(\theta_l)$. These results imply that $(\theta_l)$ is Harris positive and aperiodic, and hence Theorem~\ref{thm:converge} follows by \cite[Theorem~4.6.5]{robert1999monte}. The main technical difficulties lie in establishing recurrence and the existence of an invariant finite measure, both of which exploit the compactness of $\Omega^g \times \mathcal{H}^K$. Complete details are deferred to Appendix~\ref{appendix:proof_holder}.
The arguments used in Lemma~\ref{thm:Projected-MCMC-kernel} and Theorem~\ref{thm:converge} also apply to the case with a fixed hyperparameter set~$\bpsi$. The following corollary summarises the corresponding result.

\begin{corollary}[PBSF model with fixed hyperparameters]
\label{thm:PBSF_fixpsi}
With $\bpsi$ fixed, let $\theta_1 = \{\tilde{\bF}_1, \bgamma_1, \bSigma_1\}$ and $\theta_2 = \{\tilde{\bF}_2, \bgamma_2, \bSigma_2\}$ be two points in the sample space
\(\Theta := \Omega^g 
  \;\times\; 
  \mathbb{R}^{(p+K)\times q} 
  \;\times\; 
  \mathbb{S}_{++}^q.
\)
Then the transition kernel of the PBSF model is
\begin{equation}
\label{eq: transit_kernel_fixpsi}
  K(\theta_1,\theta_2)
  \;=\; \int_{g^{-1}(\tilde{\bF}_2)}  
        p\bigl(\bF \mid \bgamma_2,\bSigma_2, \bY\bigr)\, d\bF \;
        p\bigl(\bgamma_2,\bSigma_2 \mid \tilde{\bF}_1,\bY\bigr),
\end{equation}
where the densities $p\bigl(\bF \mid \bgamma,\bSigma,\bY\bigr)$ and $p\bigl(\bgamma,\bSigma \mid \bF,\bY\bigr)$ are given in \eqref{eq: SLMC_MNIW} and \eqref{eq: SLMC_F_cond_post}, respectively.  
Moreover, if $(\theta_\ell)$ denotes the Markov chain on $\Theta$ with transition kernel $K$ in \eqref{eq: transit_kernel_fixpsi}, then $(\theta_\ell)$ converges in total variation to its unique stationary distribution~$\pi(\cdot)$. 
\end{corollary}

\subsection{Connection and Difference between PBSF and BSF models}

We have established that the chains generated for PBSF models will converge to a valid target distribution $\pi(\cdot)$, which we refer to as the posterior distribution of the PBSF model. Based on the implementation of projection and it's connection to the BSF model in the algorithm construction, a natural question arises:

\noindent\textit{Is $\pi(\cdot)$ essentially the original BSF model's posterior with $\bF$ projected onto $\Omega^g$ via $g(\cdot)$?} 

The short answer is no. To understand why, note that ProjMC$^2$ for the BSF model can be interpreted as a Gibbs sampler. Specifically, the density \(\pi(\bgamma, \bSigma \mid \tilde{\bF} = \tilde{\bF}_0, \bpsi, \bY)\) coincides with \(p(\bgamma, \bSigma \mid \bF = \tilde{\bF}_0, \bpsi, \bY)\) for the original BSF model. 
Here, \(p(\cdot)\) denotes the conditional distributions in the original BSF model. 
Let $p(\bF, \bpsi, \bgamma, \bSigma \mid \bY)$ be the target posterior of the original BSF model. The direct projection of this posterior onto $\Theta$ is
$\int_{g^{-1}(\tilde{\bF})} p(\bF, \bpsi, \bgamma, \bSigma \mid \bY) \, d\bF$.
If we conjecture that
$\pi(\tilde{\bF}, \bpsi, \bgamma, \bSigma \mid \bY)
\;=\; \int_{g^{-1}(\tilde{\bF})} 
p(\bF, \bpsi, \bgamma, \bSigma \mid \bY) 
\, d\bF$,
then one would require
\[
\footnotesize 
\frac{p\bigl( \bF = \tilde{\bF}_0, \bpsi, \bgamma, \bSigma
\,\mid\, \bY\bigr)}{p\bigl( \bF= \tilde{\bF}_0, \bpsi 
\,\mid\, \bY\bigr) }=p\bigl(\bgamma, \bSigma \,\mid\,\bF= \tilde{\bF}_0, \bpsi, \bY\bigr) = \pi(\bgamma, \bSigma \mid \tilde{\bF}= \tilde{\bF}_0, \bpsi, \bY) = \frac{\pi(\tilde{\bF}= \tilde{\bF}_0, \bpsi,  \bgamma, \bSigma \mid \bY)}{ \pi(\tilde{\bF}= \tilde{\bF}_0, \bpsi \mid \bY)} \;,
\]
which implies
$p\bigl(\bF= \tilde{\bF}_0, \bpsi \mid \bgamma, \bSigma
, \bY\bigr) \bigl/ 
\int_{g^{-1}(\tilde{\bF}_0)} p\bigl(\bF, \bpsi \mid \bgamma, \bSigma 
, \bY\bigr)\, d\bF \bigr.$
must be independent of \(\bgamma, \bSigma\). This is nontrivial and typically does not hold, even under simple transformations. 

Although the PBSF posterior is not a direct projection of the BSF posterior, it acts as a ``squeezed'' posterior that combines the original likelihood with the original priors, but under certain integrative constraints. To see this, observe that 
$\pi(\tilde{\bF}, \bpsi \given \bgamma, \bSigma, \bY)$ 
\begin{equation}\label{eq:cond_F}
\propto\int_{g^{-1}(\tilde{\bF})} \underbrace{\mbox{MN}(\bY|\bX\beta+\bF\bLambda, I_n, \bSigma)}_{\mbox{likelihood}} \times  \underbrace{\prod_{k = 1}^K N(\fb_k \given 0, \brho_{\psi_k}) p(\psi_k)}_{\mbox{prior of } \bF, \bpsi} d\bF\;\;.
\end{equation}
On the other hand, if we condition on $\tilde{\bF}$, the distribution $ \pi( \bgamma, \bSigma \mid \tilde{\bF}, \bpsi, \mathbf{Y})$ 
\begin{equation}\label{eq:cond_other}
\propto \underbrace{\mbox{MN}(\bY \mid \bX\bbeta + \tilde{\bF}\bLambda, I_n, \bSigma)}_{\mbox{likelihood}} \times  \underbrace{\mbox{MNIW}(\bgamma, \bSigma \mid \bmu_{\bgamma}, \bV_{\bgamma}, \bPsi, \nu)}_{\mbox{prior of } \bgamma, \bSigma} \;,
\end{equation}
where $\bmu_{\bgamma} = [\bmu_{\bbeta}^\top : \bmu_{\bLambda}^\top]^\top$, $\bV_{\bgamma} = [\bV_{\bbeta}^\top : \bV_{\bLambda}^\top]^\top$. Hence, while the PBSF posterior is not the same as the original BSF posterior (with a simple projection), the samples are drawn from a distribution that incorporates both the likelihood and the original priors—albeit modified or ``squeezed''. This modification stems from the way $\tilde{\bF}$ is integrated over or restricted to the manifold defined by $g(\cdot)$, rather than being drawn directly from the support of $\bF$. 
Moreover, Brook’s Lemma \citep{Brook_1964} implies that the joint posterior distribution is determined up to a normalizing constant. Specifically, we may write 
\begin{equation}\label{eq:brook}
\small
\pi(\tilde{\bF}, \bpsi, \bgamma, \bSigma \mid \bY) \;\propto\; \frac{\pi(\bgamma, \bSigma \mid \tilde{\bF}, \bpsi, \bY)}{\pi(\bgamma=\bgamma_0,\,\bSigma=\bSigma_0 \mid \tilde{\bF}, \bpsi, \bY)} \, \pi(\tilde{\bF}, \bpsi \mid \bgamma=\bgamma_0, \bSigma=\bSigma_0, \bY)\;,  
\end{equation} 
where the reference values \((\bgamma_0,\bSigma_0)\) can be chosen arbitrarily. Here, \(\pi(\tilde{\bF}, \bpsi \mid \bgamma, \bSigma, \bY)\) and \(\pi(\bgamma, \bSigma \mid \tilde{\bF}, \bpsi, \bY)\) are proportional to \eqref{eq:cond_F} and \eqref{eq:cond_other}, respectively. 
Similar to BSF, the PBSF model allows prediction of missing and unobserved outcomes, but with an additional assumption. As it offers no improvement over BSF in this regard, further discussion is deferred to Appendix~\ref{supp: predict} for brevity.


\subsection{Comparison to Existing Manifold Sampling Methods:}
It is noteworthy that instances of forced sampling of factors constrained to the Stiefel manifold have appeared in the Bayesian literature. \citet{brubaker2012family}, \citet{Byrne_2013}, and \citet{holbrook2016bayesian} developed Riemannian Manifold Hamiltonian Monte Carlo methods that can be employed in Bayesian factor models where factors are restricted to the Stiefel manifold. However, in addition to requiring derivatives of the log-posterior density, these approaches either necessitate solving a nonlinear system at each iteration via Newton’s method, or are not generalizable to certain subsets of the manifold, rendering them more vulnerable to multimodality and non-identifiability issues. Moreover, none of these developments have been applied or extended to factors capturing spatial structures. There also exists a related literature on sampling from distributions defined directly on the Stiefel manifold \citep[e.g.,][]{Hoff_2009, Chakraborty_2019}. Yet, these methods either do not accommodate the matrix Bingham–von Mises–Fisher distribution—which is critical for incorporating spatial correlation—or rely on computationally expensive Gibbs sampling schemes, where each iteration requires constructing an orthonormal basis for the null space of the current draw.  In light of these limitations, the PBSF model is, to the best of the author’s knowledge, the first Bayesian spatial factor model that samples spatial factors on the Stiefel manifold while retaining scalable and computationally efficient sampling algorithms. Furthermore, although ProjMC\textsuperscript{2} is introduced here in tandem with its application to scalable spatial factor modelling, its underlying construction—based on conditional distributions and transition kernels—represents a general methodological contribution, potentially applicable to other classes of models defined on constrained parameter spaces.

\section{Scalable PBSF models and sampling algorithms}\label{sec: model_implement}
\subsection{Scalable modelling}
This subsection addresses scalable extensions of PBSF and its implementation. 

\textbf{Sampling of $\bF$:} The primary computational challenge in the proposed algorithm lies in efficiently sampling the high-dimensional $n\times K$ factor matrix $\bF$. To ensure scalability with respect to the number of spatial locations $n$, NNGP \citep{datta16} is employed. The NNGP provides a sparse, full-rank approximation of a full GP, effectively capturing localized and global spatial dependence while ensuring linear computational and storage complexity in $n$. Specifically, each latent factor $f_k(\bs)$ for $\bs \in \mathcal{D}$ is endowed with an NNGP prior, $\mbox{NNGP}(0, \rho_{\psi_k}(\cdot, \cdot))$, implying that $\fb_k \sim \mbox{N}(\bzero, \brho_k)$, where $\brho_k = (\bI - \bA{\rho_k})^{-1}\bD{\rho_k}(\bI - \bA_{\rho_k})^{-\top}$. Here, the matrices $\bA_{\rho_k}$ (sparse, lower-triangular) and $\bD_{\rho_k}$ (diagonal) are constructed based on conditional expectations and variances derived from the covariance function $\rho_{\psi_k}(\bs,\bs')$ \citep[see details in][]{finley2019efficient}. The number of neighbors $m$ for NNGP is typically set below $20$ for computational efficiency. A maximin ordering of spatial locations is adopted, which has been shown to be robust and efficient for Vecchia approximations and can be computed in quasilinear time in $n$\citep{katzfuss2021general, guinness2018permutation, schäfer2021sparse}. 

{In the PBSF model where $\bpsi$ is included in the parameter space, the matrices $\bD_{\rho_k}$ and $\bA_{\rho_k}$ must be recomputed at each update of $\psi_k$, with complexity $\mathcal{O}(n\cdot m^3 \cdot K) = \mathcal{O}(n)$ when $m << n$. Although these updates are scalable and can, in principle, be accelerated through parallelization, they may still represent a relatively significant computational burden—particularly when $\bpsi$ is updated via algorithms like slice sampling, which requires multiple evaluations of $\bD_{\rho_k}$ and $\bA_{\rho_k}$ within each iteration. In contrast, for the PBSF model with prefixed $\bpsi$, the matrices $\bD_{\rho_k}$ and $\bA_{\rho_k}$ can be precomputed. This eliminates the need for repeated updates, making the fixed-$\bpsi$ formulation particularly attractive in practise. Moreover, fixing $\bpsi$ is common in applications, and simulations further demonstrate that doing so can improve sampling efficiency for loadings $\bLambda$.}
 
By modelling $f_k(\bs)$ using NNGP, the cost of a direct sample of $\bF$ from the high-dimensional Gaussian distribution \eqref{eq: SLMC_F_cond_post} can be reduced from $\mathcal{O}(n^3)$ to $\mathcal{O}(n)$. Recall that the $\bL_k^{-1}$ in $\tilde{\bX}$ is now $\bL_k^{-1} = \bD_{\rho_k}^{-1/2} (\bI - \bA_{\rho_k})$, which implies that the number of nonzero elements in $\tilde{\bX}$ grows linearly with $n$. Such inherent sparsity facilitates efficient conjugate gradient algorithms \citep{zdb2019, nishimura2023prior, zhang2022applications}. More concretely, obtaining a sample from \eqref{eq: SLMC_F_cond_post} reduces to solving the linear system $\tilde{\bX}x = (\tilde{\bY} + v)$, where the elements of $v$ are drawn independently from a standard normal distribution. In our implementation, we employ the iterative LSMR method, a robust solver for sparse linear systems that avoids explicitly forming the matrix $\tilde{\bX}^\top \tilde{\bX}$ \citep{fong2011lsmr}. When generating $\tilde{\bF} = g(\bF)$, we employ the modified Gram-Schmidt to perform a thin QR decomposition. This approach is particularly efficient when $n >> K$, with total cost in $\mathcal{O}(nK^2)$ \citep{golub2012}. Detailed algorithmic procedures are provided in Appendix~\ref{SM: BSLMC_NNGP_alg}. Finally, while the discussion focuses on the NNGP approach, some of the computational strategies presented for sampling $\bF$ can be adapted or modified for use with other scalable spatial modelling approach. See \citet{zhang2022spatial} for related discussions.

\textbf{Sampling of $\bpsi$:} Since the conditional posterior of $\bpsi$ in BSF models is not available in closed form, one would need sampler such as Metropolis–Hastings (M–H), slice sampling, or Hamiltonian Monte Carlo (HMC) to update $\bpsi$. As pointed out in the construction of ProjMC$^2$, the target distribution for $\bpsi_k$ updates is proportional to $p(\bF_k \mid \psi_k) p(\psi_k)$ and the evaluation is scalable when $f_k(s)$ is modelled using NNGP. In the first simulation study in Section~\ref{sec: simulation}, a componentwise univariate slice sampler with an adaptive bracket width is implemented. The bounded parameter $\psi_k \in [a_{\psi_k}, b_{\psi_k}]$ is reparameterized into $\xi = \operatorname{log}\!\Big\{({\psi_k-a_{\psi_k}})/({b_{\psi_k}-\psi_k})\Big\}$. Sample of $\xi$ is through  Neal’s stepping–out and shrinkage scheme \citep{neal2003slice}. During a short warm–up the slice width $w$ is adapted by a diminishing–stepsize Robbins–Monro update that targets roughly $1$–$2$ step–out expansions per iteration (with a small penalty for excessive shrinkage), 
and then freeze $w$ to preserve stationarity and ergodicity \citep{roberts2007coupling}. 

\textbf{Initialisation:} To facilitate convergence, the regression parameters $\bbeta$ are initialized using ordinary least squares estimation $(\bX^\top \bX)^{-1}\bX^\top \bY$. The latent factor matrix $\bF$ and the loading matrix $\bLambda$ are initialized via principal component analysis (PCA) applied to residuals obtained from the initial regression fit, $\bY - \bX\bbeta$. Considering that PCA is employed solely for initialisation and will subsequently be updated through MCMC iterations—and given the potentially large dimensions of $n$ and $q$—randomized SVD is adopted to ensure computational complexity scales linearly with respect to $n$ and $q$ \citep{halko2011algorithm, martinsson2020randomized}. It is recommended to arrange the latent factors and loading matrices such that the factors are ordered by decreasing empirical spatial range (i.e., decreasing smoothness), enhancing convergence stability as discussed in detail in the following subsection. Finally, diagonal entries of $\bSigma$ can be initialized using residual variances from the initial regression fit, $\bY - \bX\bbeta$. In the case of misalignment, $\bbeta$ can be initialized column-wise using the observations for each outcome, while $\bLambda$ and $\bSigma$ can be initialized from the data at locations without misalignment. 

\begin{algorithm}[h]
\setstretch{1.0}  
\caption{Sampling algorithm for NNGP based PBSF model}
\label{code: alg1_brief}
\begin{algorithmic}[1]

\State \textbf{Input:} Design matrix $\bX$, outcomes $\bY$, set of spots $\chi$, prior parameters, prefixed parameters, and number of MCMC iterations $L$.

\State \textbf{initialisation, precomputation, and preallocation} \newline
\textit{\small (Includes preallocation of necessary matrices, construction of the maximin ordering of $\chi$, precomputation of matrices $\bD_{\rho_k}$ and $\bA_{\rho_k}$, and initialisation of parameters for MCMC.)}

\For{$l = 1, \dots, L$}

\State Sample $\bF^{(l)}$ by solving $\tilde{\bX}x = (\tilde{\bY} + v)$, with $v \sim N(0, \bI)$, using LSMR.

\State Update $\psi_k$ for each $k$ via M--H, slice sampling or HMC, with target density proportional to $p(\bF_k \mid \psi_k)\, p(\psi_k)$. \textit{\small (Skip for fixed $\bpsi$ )}

\State Generate projected embeddings $\tilde{\bF}^{(l)} = g(\bF^{(l)})$. \textit{\scriptsize (Update $\bX^\ast$ in \eqref{eq: augment_linear_LMC} by replacing $\bF$ with $\tilde{\bF}^{(l)}$.)}

\State Sample parameters $(\bbeta^{(l)}, \bLambda^{(l)}, \bSigma^{(l)})$ using \eqref{eq: SLMC_MNIW} (or \eqref{eq: SLMC_NIG_misaligned} with diagonal $\bSigma$ and misaligned data).



\EndFor

\State \textbf{Output:} Posterior samples $\{\tilde{\bF}^{(l)}, \bpsi^{(l)}, \bbeta^{(l)}, \bLambda^{(l)}\}_{l= 1, \ldots, L}$ of low-dimensional embeddings, regression coefficients, and loading matrix.
\newline
\textit{\small (Retain posterior samples after the warm-up period; thinning, i.e., keeping one iteration per several iterations, can be used to reduce storage.)}
\end{algorithmic}
\end{algorithm}

\textbf{Variants:} The proposed PBSF model admits several variants
. A brief summary is provided here. The most computationally efficient formulation arises when there is no misalignment and the hyperparameters $\bpsi$ are fixed. In this case, $\bSigma$ may take any form, provided it is positive definite. When misalignment is present, conditional conjugacy is exploited to ensure efficient sampling, which requires $\bSigma$ to be diagonal. This restriction is standard and particularly advantageous when the number of outcomes $q$ is large. Under misalignment, Step 7 of Algorithm~\ref{code: alg1_brief} incurs additional cost, as computations must be performed separately for each outcome using only the locations where that outcome is observed.
In all of the above variants, the hyperparameters $\bpsi$ may also be updated. Since hyperparameters in commonly used spatial kernels are notoriously difficult to sample, slice sampling or HMC may be preferable to M–H for improving chain mixing. Notably, the extension to misaligned data through \eqref{eq: SLMC_NIG_misaligned}, the implementation of slice sampling, and initialisation are new algorithmic contributions relative to \cite{zhang2022spatial}, where missing outcomes had to be imputed within the MCMC and $\bpsi$ were handled only through adaptive M-H. The formulation in \eqref{eq: SLMC_NIG_misaligned} also introduces greater flexibility in specifying priors for $\bgamma$. A simplified algorithmic summary is presented in Algorithm~\ref{code: alg1_brief}, with comprehensive details available in Appendix~\ref{SM: BSLMC_NNGP_alg}.


 
\subsection{Implementation and Practical Considerations}\label{subsec: Imple_consider}

\textbf{Convergence and Mixing in Practise:} While theoretical results guarantee that ProjMC$^2$ for the BSF model converges to its stationary distribution regardless of initialisation, posterior multimodality remains possible despite restricting factor matrix samples to the space $\Omega^g$. This concern arises naturally because $\Omega^g$ is closed under permutations and sign changes of factors. Such inherent symmetry induces identifiability issues within the likelihood function, complicating the posterior landscape and potentially leading to slower MCMC mixing and more complex convergence rate analyses.

Empirically, however, ProjMC$^2$ demonstrates greater stability than initially anticipated. Simulation studies consistently indicate that recovered factors tend to organise themselves by decreasing smoothness. Columns of the projected factor matrix $\tilde{F}$ naturally order themselves, placing smoother factors in initial columns and noisier factors in subsequent columns. The empirical investigation attributes this observed stability primarily to the QR decomposition step within the projection $g(\cdot)$. 

To intuitively illustrate this phenomenon, consider a simplified scenario with two factors, and we represent the factor matrix prior to projection as $\bF = [\fb_1, \fb_2]$. Suppose that $\fb_1 \sim \mathcal{N}(\mu_1, C_1)$ and $\fb_2 \sim \mathcal{N}(\mu_2, C_2)$. The projected factor matrix via QR decomposition, $\tilde{\bF} = [\tilde{\fb}_1, \tilde{\fb}_2]$, is given by $\tilde{\fb}_1 = \fb_1/\|\fb_1\|$, followed by $\tilde{\fb}_2 = \fb_2 - (\tilde{\fb}_1^\top \fb_2)\tilde{\fb}_1$ and normalization $\tilde{\fb}_2 = \tilde{\fb}_2/\|\tilde{\fb}_2\|$. 
When the initial factor $\fb_1$ exhibits strong correlations (smoothness), QR projection introduces minimal perturbations in subsequent factors. Consequently, posterior distributions concentrate in regions of $\Omega^g$ favoring this smooth-to-noisy factor ordering. Conversely, placing less correlated (noisy) factors first increases perturbations, dispersing posterior mass and reducing local density.

Hence, the QR projection implicitly induces a preferred ordering of factor embeddings based on the underlying correlation structure, thereby enhancing sampling stability. Although it remains theoretically challenging to fully capture all posterior modes using finite-length MCMC chains, the empirical findings suggest that the QR-induced ordering substantially enhances practical stability and sampling efficiency. 

\textbf{Postprocessing for Label Switching Issue:} Although the QR decomposition can enhance identifiability, it does not resolve the so-called label-switching issues inherent in the proposed PBSF models. In practise, MCMC samples for rows of the loading matrix $\bLambda$ often oscillate between two symmetric modes, where the signs of all elements within a row simultaneously switch. Direct posterior summarization of these draws without correcting for label switching can lead to misleading inference \citep{stephens2000dealing}. To mitigate this issue, we employ a post-processing approach to align samples from the minor mode with the dominant mode \citep[Section 22.3]{stephens2000dealing, gelman2013}. Specifically, the posterior mean of each row of the loading matrix is first computed from the MCMC chains after warm-up. Then, at each iteration, check the sign of its inner product with the current draw. If negative, flip the signs of the row and corresponding latent factor. This realignment ensures coherent posterior summaries.


\section{Simulation}\label{sec: simulation}
Two simulation studies were designed to evaluate variants of the PBSF models. For both studies, the response $\by(\bs)$ was simulated from the spatial factor model in \eqref{eq: spatial_factor} with $q = 10$, $p = 2$, $K = 2$, and a diagonal $\bSigma$ over $n = 2000$ randomly generated locations over a unit square. The explanatory variable $\bx(\bs)$ consists of an intercept and a single predictor generated from a standard normal. 
Each $f_k(\bs)$ was generated using an exponential covariance function, i.e.,
$
\rho_{\psi_k}(\bs, \bs') = \exp{(-\phi_k\|\bs - \bs'\|)}, \text{ for } \bs, \bs' \in {\mathcal{D}}\;,
$
where $\|\bs - \bs'\|$ is the Euclidean distance between $\bs$ and $\bs'$, and $\psi_k = \phi_k$ is the decay for each $k$. The true decay parameters used to generate  $f_1(\bs)$ and $f_2(\bs)$ were 6.0 and 9.0, respectively. The factors were then centred to mean zero and scaled to have norm $\sqrt{n-1}$ for data generation. The exact values for the remaining parameters are provided in Appendix~\ref{sm: values_examples}. A flat prior was assigned for $\bbeta$ and $\bLambda$ and $\mbox{IG}(2, 1.0)$ priors were assigned for the diagonal elements of $\bSigma$. The priors for the latent spatial factors are modelled through NNGP with the number of neighbors $m = 15$. Simulation studies compared the proposed ProjMC$^2$ algorithm with its MCMC counterpart without the projection step (hereafter referred to as ``Gibbs'') to quantify the effect of projection. The posterior inference for each model was based on MCMC chains with 15,000 iterations after a burn-in of 5,000 iterations. All models were run on an Apple MacBook Pro with an M2 Max chip, 12-core CPU (8 performance and 4 efficiency cores), and 96 GB of unified memory, running macOS 14.6.1 (Sonoma). Convergence diagnostics and other posterior summaries were implemented within the Julia statistical environment.

\subsection{Simulation Results}
In the first study, 25\% of the observations for each outcome were randomly withheld, resulting in 94.85\% of locations having at least one missing outcome. For this dataset, the PBSF model for misaligned data was fitted with $\bpsi$ updated, assigning a uniform prior on $[0.1,20]$ to each decay parameter. 

\textbf{Convergence and Mixing Rate Evaluation:} This study first investigated ProjMC$^2$ on convergence and mixing rate by comparing three MCMC approaches. The baseline approach is the ``Gibbs'' sampler. As previously discussed, certain non-identifiability issues can be mitigated through appropriate post-processing techniques. Among several post-processing methods tested, it was found that recentering spatial factors at zero yielded the most substantial improvement in convergence and mixing rate. Since the proposed ProjMC$^2$ algorithm also requires post-processing, comparisons were made between the ``Gibbs'' sampler with post-processing (Gibbs+Post) and ProjMC$^2$. 
For conciseness, the improvement achieved by post-processing for the baseline algorithm is provided along with the second simulation results in Appendix~\ref{sm: sim2_results}.

Table~\ref{tab:ess_comparison} summarises the effective sample sizes (ESS) for all model parameters, and Figure~\ref{fig:trace_compare} presents trace plots for the decay parameters ($\phi_1$ and $\phi_2$) and the loadings (elements of $\bLambda$). As shown in Table~\ref{tab:ess_comparison}, the intercepts ($\bbeta_0$, the first row of $\bbeta$), regression coefficients ($\bbeta_1$, the second row of $\bbeta$), and the noise covariance matrix $\bSigma$---were sampled efficiently, each yielding minimum ESS values exceeding 2,000.

\begin{table}[t]
\centering
\footnotesize
\begin{tabular}{lccccc}
\toprule
 & 
\multicolumn{2}{c}{\textbf{Gibbs + Post}} & 
\multicolumn{2}{c}{\textbf{ProjMC\textsuperscript{2}}} \\
\cmidrule(lr){2-3} \cmidrule(lr){4-5}  \textbf{ESS}
&(min/mean/med)& $<100$ 
&(min/mean/med)& $<100$ \\
\midrule
$\bbeta_0$ &  7272/11732/12359 & 0\% & 7351 / 11851 /12414	& 0\% \\
$\bbeta_1$ & 6446 /11268/ 12173 & 0\% & 6891/11484/12096 & 0\% \\
$\bLambda$  & 32/35/34  & 100\% & 178/1173/241 & 0\% \\
$\bF$& 35/872/103 & 49\% & 382/7766/7114 & 0\% \\
$\bSigma$ & 2479/9809/11766 & 0\% & 2668/9704/11131	& 0\% \\
$\bpsi$ & $(\phi_1, \phi_2):$ (36, 39) & 100\% & $(\phi_1, \phi_2):$ (370, 1060) & 0\%\\
\bottomrule
\end{tabular}
\caption{ \footnotesize Comparison of effective sample size (ESS)—reported as minimum, mean, median, and the proportion of variables with low ESS values (ESS $<100$)— across the ``Gibbs'' sampler with post-processing (Gibbs + Post), and the proposed ProjMC$^2$. Results are shown for the intercepts $\bbeta_0$, regression coefficients $\bbeta_1$, loading matrix $\bLambda$, matrix of latent factors $\bF$, the noise covariance matrix $\bSigma$, and the decay parameters $\bpsi =(\phi_1, \phi_2)$, based on MCMC chains of 20{,}000 iterations with the first 5{,}000 iterations discarded as warm-up. }
\label{tab:ess_comparison}
\end{table}

For parameters involving the loading matrix ($\bLambda$), latent spatial factors ($\bF$) and decays $(\phi_1, \phi_2)$, stable and efficient MCMC chains were only obtained through the ProjMC$^2$ method. Specifically, ProjMC$^2$ significantly increased the minimum and median ESS for spatial factors from 35 and 103 to 382 and 7,114, respectively. All loading matrix elements, initially with ESS values below 100, attained ESS values exceeding 150. And the ESS for $\phi_1$ and $\phi_2$ increases from 36 and 39 to 370 and 1060, respectively. Moreover, trace plots in Figure~\ref{fig:trace_compare} demonstrate rapid convergence of MCMC chains produced by ProjMC$^2$ for decays and loadings, reaching high-probability regions within a small number of iterations. Although loading parameters exhibit relatively slower mixing, convergence occurs reliably within a few hundred iterations, underscoring the effectiveness of the proposed approach. In addition, Figure~\ref{fig:trace_compare} reveals several noteworthy phenomena. The MCMC chains for decays and loadings exhibit similar patterns, suggesting mutual influence and correlation. 
The trace plots for decays further indicate that the PBSF model tended to slightly underestimate $\phi_1$ and overestimate $\phi_2$ in this simulation study, consistent with the preference for decreasing smoothing discussed in Section~\ref{subsec: Imple_consider}. As observed previously, this bias may propagate to $\bLambda$, resulting in slight bias in the loading estimates. 

\begin{figure}[!ht]
  \centering
  \subfloat[Decays (Gibbs+Post) \label{subfig:trace_post_phi}]{%
    \includegraphics[width=0.32\textwidth]{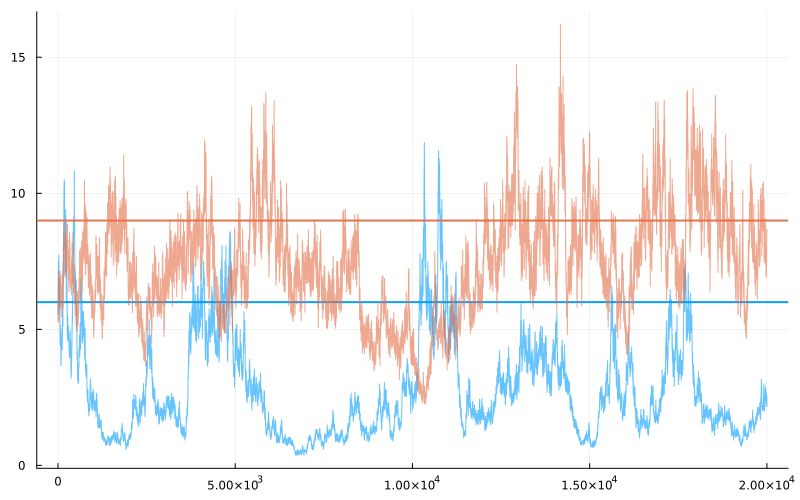}
  }
  \subfloat[Factor 1 Loadings (Gibbs+Post)\label{subfig:trace_post_Lambda1}]{%
    \includegraphics[width=0.335\textwidth]{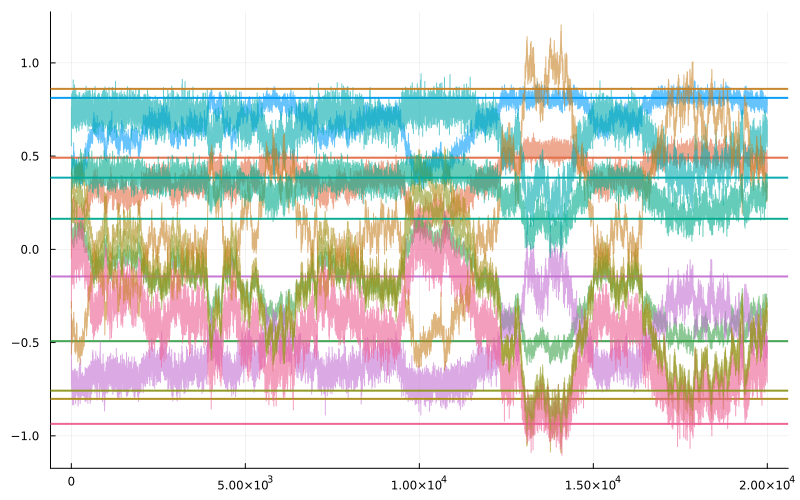}
  }
  \subfloat[Factor 2 Loadings (Gibbs+Post)\label{subfig:trace_post_Lambda2}]{%
    \includegraphics[width=0.335\textwidth]{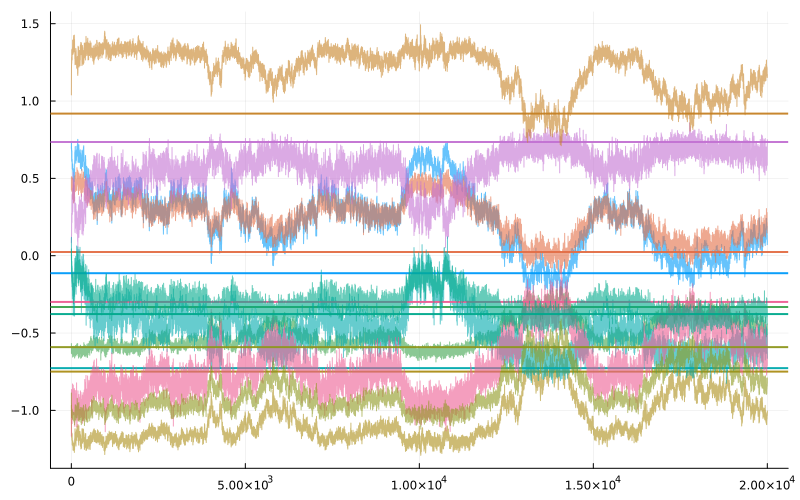}
  }\\
  \subfloat[Decays (ProjMC\textsuperscript{2}) \label{subfig:trace_projmc_phi}]{%
    \includegraphics[width=0.32\textwidth]{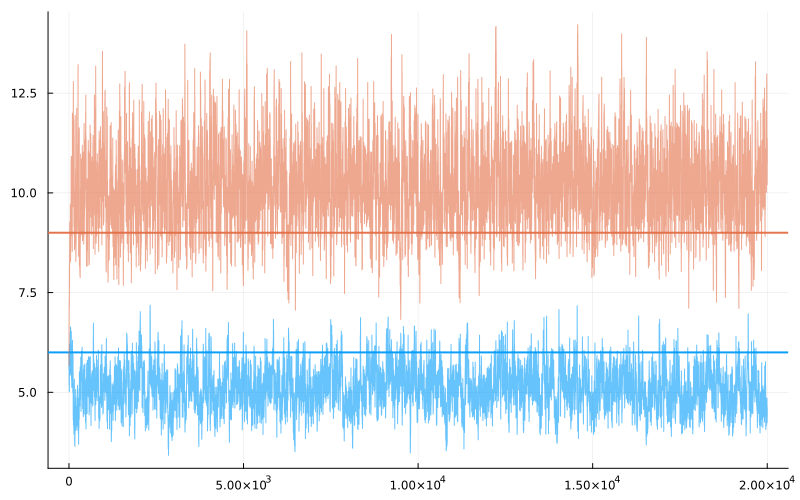}
  }
  \subfloat[Factor 1 Loadings (ProjMC\textsuperscript{2}) 
  \label{subfig:trace_projmc_Lambda1}]{%
    \includegraphics[width=0.33\textwidth]{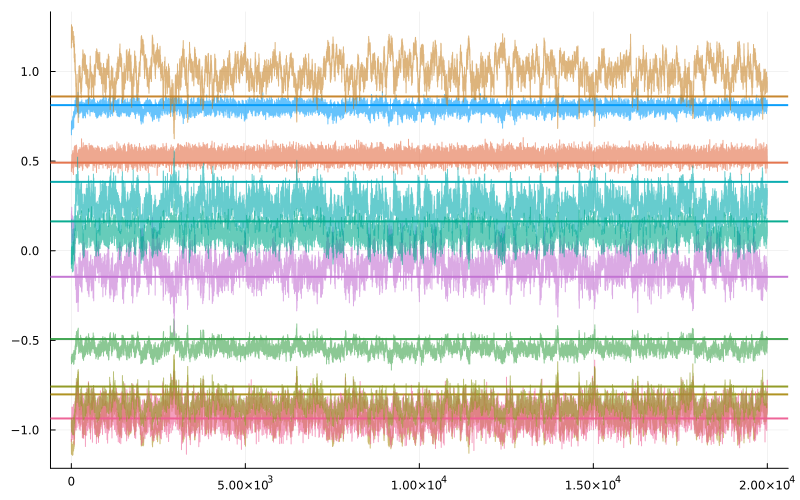}
  }
  \subfloat[Factor 2 Loadings (ProjMC\textsuperscript{2}) \label{subfig:trace_projmc_Lambda2}]{%
    \includegraphics[width=0.33\textwidth]{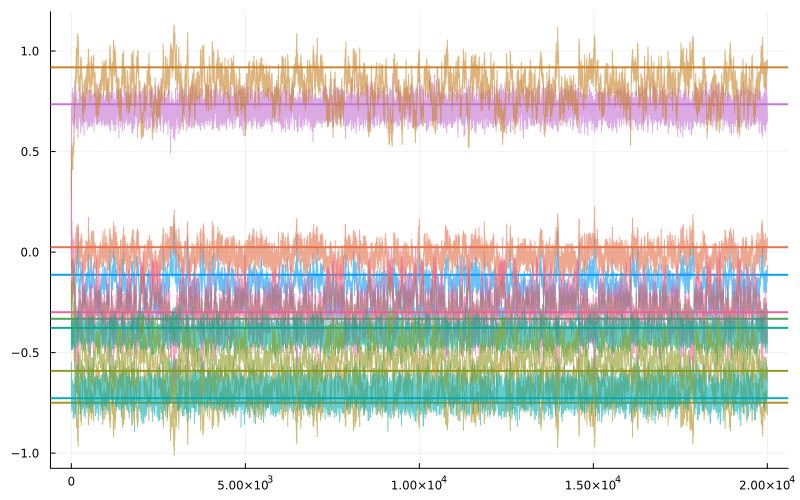}
  }
  \caption{\small Trace plots of MCMC chains for weakly identifiable parameters: the decays (first column) and loading matrix $\bLambda$ (second and third columns). Rows correspond to results from the blocked Gibbs sampler with post-processing (top row), and ProjMC\textsuperscript{2} (bottom row). Horizontal lines indicate the parameter values used to generate data.}
  \label{fig:trace_compare}
\end{figure}

\textbf{Inference Accuracy: }
To compare the inference accuracy of the high-dimensional spatial factors across different algorithms, both the true factor and all posterior samples were projected onto the scaled sphere $\sqrt{n-1} \cdot\mathcal{S}^{n-1}$. For each factor, the point estimate was obtained using the Fr\'echet mean (or mean direction) of its posterior samples \cite{mardia2009directional}. Figure~\ref{fig:f_compare} provides a visual comparison between the true latent spatial factors and their point estimates obtained from each method. Visual inspection of Figure~\ref{fig:f_compare} indicates that both Gibbs+Post and ProjMC$^2$ successfully capture the dominant spatial patterns present in both factors ($\fb_1$ and $\fb_2$). However, a closer examination reveals differences in the recovery of more subtle features. The factors estimated via ProjMC$^2$ appear to accentuate patterns that are distinct to each factor. Conversely, Gibbs+Post sometimes yields estimates where subtle patterns exhibit similarity across different factors. This tendency in ProjMC$^2$ is likely attributable to the inherent constraints imposed by sampling on the Stiefel manifold, which enforces stricter orthogonality.

\begin{figure}[!ht]
  \centering
  \captionsetup[subfloat]{width=0.3\textwidth, justification = centering}
  \subfloat[\centering True $f_1$ \label{subfig:MCMC_true_f1}]{%
    \includegraphics[width=0.27\textwidth, height=0.22\textwidth]{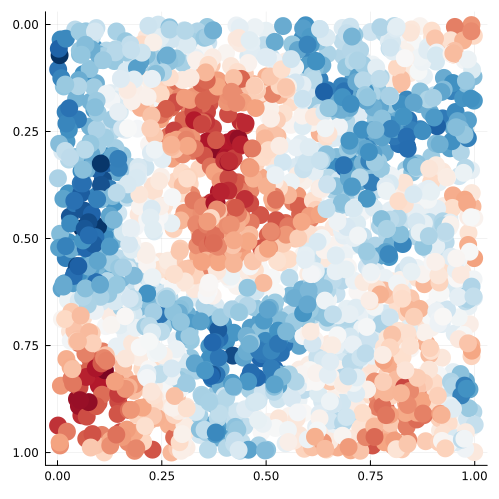}
  }
  \subfloat[\centering Est $f_1$ (Gibbs + Post) \label{subfig:MCMC_est_f1}]{%
    \includegraphics[width=0.27\textwidth, height=0.22\textwidth]{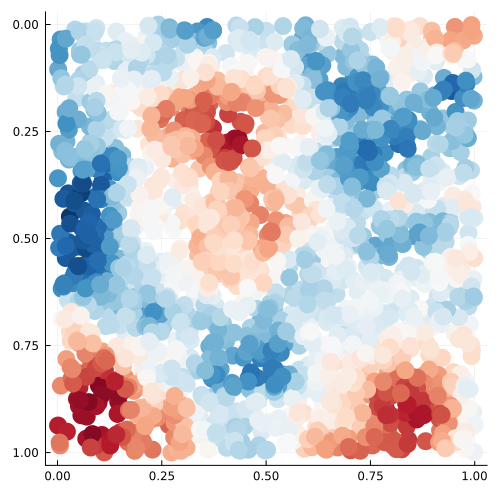}
  }
  \subfloat[\centering Est $f_1$ (ProjMC\textsuperscript{2})\label{subfig:PMCMC_est_f1}]{%
    \includegraphics[width=0.27\textwidth, height=0.22\textwidth]{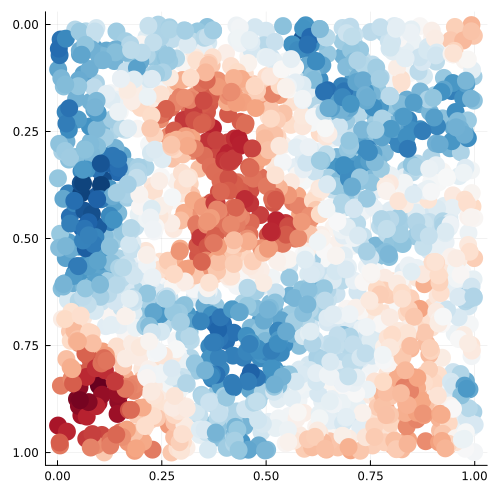}
  }\\
  \subfloat[\centering True $f_2$ \label{subfig:MCMC_true_f2}]{%
    \includegraphics[width=0.27\textwidth, height=0.22\textwidth]{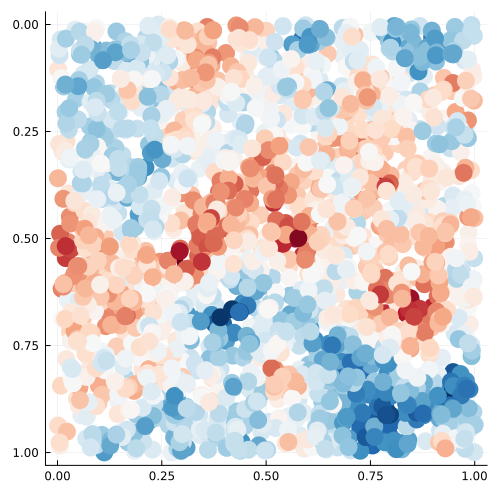}
  }
  \subfloat[\centering Est $f_2$ (Gibbs + Post)\label{subfig:MCMC_est_f2}]{%
    \includegraphics[width=0.27\textwidth, height=0.22\textwidth]{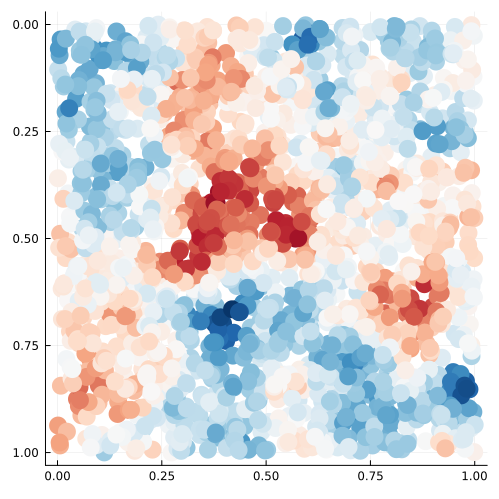}
  }
  \subfloat[\centering Est $f_2$ (ProjMC\textsuperscript{2})\label{subfig:PMCMC_est_f2}]{%
    \includegraphics[width=0.27\textwidth, height=0.21\textwidth]{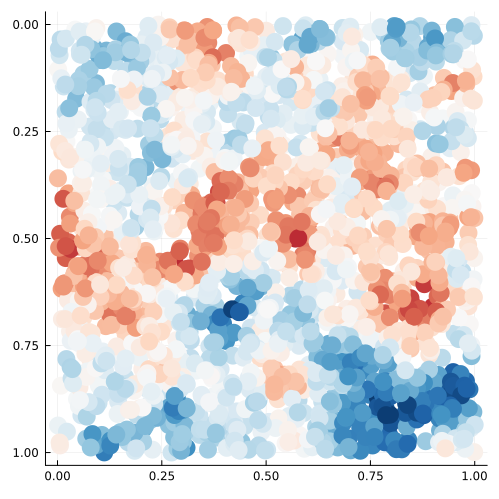}
  }
  \caption{ \footnotesize Scatter plots of the true and estimated posterior means for the two latent spatial factors, $\fb_1$ (top row) and $\fb_2$ (bottom row), from the first simulation study. Dot locations indicate spatial positions, and colours represent latent factor values. Columns correspond to the true factors, Gibbs sampler with post-processing, and the proposed ProjMC\textsuperscript{2} method. All factor estimates were centred at zero and rescaled to have norm $\sqrt{n-1}$, and results for the same factor share a common colour scale for visual comparison. }
  \label{fig:f_compare}
\end{figure}

Quantitative assessment of inference accuracy was based on two metrics, summarised in Table~\ref{tab:latent_diag_compare}. First, the fidelity of the point estimate was measured by the Euclidean distance between the estimated Fr\'echet mean and the true factor. Second, the concentration or stability of the posterior distribution for each factor was quantified by the spherical variance, calculated as $n-1 - \bar{R}^2$, where $\bar{R}$ denotes the Euclidean norm of the sample mean vector, computed using the posterior samples projected onto the sphere of radius $\sqrt{n-1}$. Overall, ProjMC$^2$ yields more accurate point estimates in this simulation study, as indicated by lower Euclidean distances and spherical variances. 

\begin{table}[htbp]
\centering
\footnotesize
\begin{tabular}{lcccc}
\toprule
\textbf{Latent Factor} & 
\multicolumn{2}{c}{\textbf{Gibbs + Post}} & 
\multicolumn{2}{c}{\textbf{ProjMC\textsuperscript{2}}} \\
\cmidrule(lr){2-3} \cmidrule(lr){4-5}
& Euclidean Dist. & Sphere Var. & Euclidean Dist. & Sphere Var. \\
\midrule
$f_1$ & 27.24 & 359.2 & 13.82 & 172.7 \\
$f_2$ & 21.48 & 377.6 & 18.46 & 336.7 \\
\bottomrule
\end{tabular}
\caption{\footnotesize Comparison of posterior summaries for the two latent spatial factors ($\fb_1$ and $\fb_2$) across the Gibbs sampler with post-processing (Gibbs + Post) and ProjMC\textsuperscript{2}. Each method is evaluated using two diagnostics: Euclidean distance between the true and estimated factor, and the spherical variance of the posterior samples.}
\label{tab:latent_diag_compare}
\end{table}

Finally, the posterior inference for the identifiable model parameters were assessed, specifically the regression coefficient matrix $\bbeta$ (encompassing intercepts $\bbeta_0$ and slope coefficients $\bbeta_1$) and the diagonal elements of the noise covariance matrix $\bSigma$. Detailed posterior summaries, including posterior means and 95\% credible intervals obtained from both the Gibbs+Post sampler and the ProjMC\textsuperscript{2} algorithm, are presented alongside the true parameter values in Table~\ref{tab: sim_infer_sum}. Examination of Table~\ref{tab: sim_infer_sum} reveals that the posterior inferences for the regression coefficients $\bbeta$ and the noise variance parameters (diagonal elements of $\bSigma$) are almost indistinguishable between the two methods. Both approaches yield posterior means that closely approximate the true values, and the corresponding 95\% credible intervals demonstrate almost the same width and coverage. 

\begin{table}[htbp]
\centering
\scriptsize
\setlength{\tabcolsep}{3pt}  
\begin{tabular}{cccccccccccc}
\toprule
 &  & 
\multicolumn{2}{c}{\textbf{Gibbs + Post}} & 
\multicolumn{2}{c}{\textbf{ProjMC\textsuperscript{2}}} &
 &  & 
\multicolumn{2}{c}{\textbf{Gibbs + Post}} & 
\multicolumn{2}{c}{\textbf{ProjMC\textsuperscript{2}}}\\
\cmidrule(lr){3-4} \cmidrule(lr){5-6} \cmidrule(lr){9-10} \cmidrule(lr){11-12}
& & mean & 95\%CI &  mean & 95\%CI & & & mean & 95\%CI &  mean & 95\%CI\\
\midrule
\hline
$\bbeta_{[1, 1]}$ & 1.0 & 0.91 & (0.84, 0.99) & 0.91 & (0.84, 0.99) &$\bbeta_{[1, 6]}$ & -1.5 & -1.6 & (-1.76, -1.43) & -1.6 & (-1.76, -1.43) \\
$\bbeta_{[1, 2]}$ & -1.0 & -1.08 & (-1.19, -0.98) & -1.08 & (-1.19, -0.98) &$\bbeta_{[1, 7]}$ & 0.5 & 0.47 & (0.28, 0.66) & 0.47 & (0.28, 0.66) \\
$\bbeta_{[1, 3]}$& 1.0 & 1.0 & (0.94, 1.07) & 1.0 & (0.94, 1.07) &$\bbeta_{[1, 8]}$ & 0.3 & 0.27 & (0.2, 0.35) & 0.27 & (0.2, 0.35) \\
$\bbeta_{[1, 4]}$ & -0.5 & -0.47 & (-0.61, -0.33) & -0.47 & (-0.61, -0.32) &$\bbeta_{[1, 9]}$ & -2.0 & -2.07 & (-2.2, -1.94) & -2.07 & (-2.2, -1.94) \\
$\bbeta_{[1, 5]}$ & 2.0 & 2.0 & (1.93, 2.06) & 1.99 & (1.93, 2.06) &$\bbeta_{[1, 10]}$ & 1.5 & 1.49 & (1.41, 1.57) & 1.49 & (1.42, 1.56) \\
\hline
$\bbeta_{[2, 1]}$ & -3.0 & -2.82 & (-2.96, -2.69) & -2.82 & (-2.96, -2.68) &$\bbeta_{[2, 6]}$ & 3.0 & 3.12 & (2.83, 3.41) & 3.12 & (2.84, 3.42) \\
$\bbeta_{[2, 2]}$ & 2.0 & 2.08 & (1.9, 2.26) & 2.08 & (1.9, 2.26) &$\bbeta_{[2, 7]}$ & 4.0 & 4.1 & (3.76, 4.43) & 4.1 & (3.77, 4.43) \\
$\bbeta_{[2, 3]}$ & 2.0 & 2.04 & (1.93, 2.16) & 2.04 & (1.93, 2.16) &$\bbeta_{[2, 8]}$ & -2.5 & -2.5 & (-2.64, -2.36) & -2.5 & (-2.64, -2.36) \\
$\bbeta_{[2, 4]}$ & -1.0 & -1.12 & (-1.37, -0.86) & -1.12 & (-1.37, -0.86) &$\bbeta_{[2, 9]}$ & 5.0 & 5.15 & (4.93, 5.38) & 5.15 & (4.93, 5.37) \\
$\bbeta_{[2, 5]}$ & -4.0 & -4.0 & (-4.12, -3.88) & -4.0 & (-4.12, -3.88) &$\bbeta_{[2, 10]}$ & -3.0 & -2.95 & (-3.09, -2.82) & -2.95 & (-3.09, -2.82) \\
\hline
$\bSigma_{[1, 1]}$ & 0.5 & 0.52 & (0.47, 0.56) & 0.51 & (0.47, 0.56) &$\bSigma_{[6, 6]}$ & 2.5 & 2.59 & (2.4, 2.79) & 2.59 & (2.4, 2.79) \\
$\bSigma_{[2, 2]}$ & 1.0 & 1.04 & (0.97, 1.12) & 1.04 & (0.97, 1.12) &$\bSigma_{[7, 7]}$ & 3.5 & 3.54 & (3.29, 3.81) & 3.54 & (3.29, 3.8) \\
$\bSigma_{[3, 3]}$ & 0.4 & 0.41 & (0.38, 0.45) & 0.41 & (0.38, 0.45) &$\bSigma_{[8, 8]}$ & 0.45 & 0.47 & (0.43, 0.52) & 0.47 & (0.43, 0.52) \\
$\bSigma_{[4, 4]}$ & 2.0 & 2.07 & (1.91, 2.23) & 2.06 & (1.91, 2.22) &$\bSigma_{[9, 9]}$ & 1.5 & 1.59 & (1.47, 1.71) & 1.58 & (1.47, 1.7) \\
$\bSigma_{[5, 5]}$ & 0.3 & 0.31 & (0.28, 0.34) & 0.31 & (0.28, 0.34) &$\bSigma_{[10, 10]}$ & 0.5 & 0.5 & (0.46, 0.54) & 0.5 & (0.46, 0.54) \\
\bottomrule
\end{tabular}
\caption{\footnotesize Posterior inference for identifiable model parameters. Comparison of posterior means and 95\% credible intervals for regression coefficients ($\bbeta$) and noise variances (diagonal elements of $\bSigma$) obtained using the Gibbs sampler with post-processing (Gibbs+Post) and the proposed ProjMC\textsuperscript{2} algorithm, referenced against the true parameter values.}
\label{tab: sim_infer_sum}
\end{table}

\textbf{Simulation II Results:} In the second study, no missingness was assumed and the most sampling-efficient algorithm was examined, namely the PBSF model without misalignment and with a fixed value of $\bpsi$. Within the specified priors, the decay parameters for $f_1(\bs)$ and $f_2(\bs)$ were deliberately set to 4.0 and 6.0, respectively, in order to simulate a scenario in which the decay rates are misspecified or underestimated—a situation commonly encountered in practise. To avoid redundancy, repeated evaluations for the second simulation study (e.g., ESS tables and posterior inference) are reported in Appendix~\ref{sm: sim2_results}. A brief summary of the main observations is provided here. First, post-processing substantially increased the sampling efficiency of the intercepts $\bbeta_0$ for the ``Gibbs'' sampler. Second, fixing the hyperparameters led to a 50\%--100\% increase in ESS for weakly identifiable parameters. Overall, the remaining results were broadly consistent with those obtained in the first simulation study. In terms of computational cost, switching from Simulation I to Simulation II reduced the running time from approximately 36 minutes to under 6 minutes. Although misspecifying the hyperparameters introduced slightly larger bias in factor learning, the algorithm continued to perform well for spatially aware low-dimensional representation learning. These results suggest that fixing hyperparameters can provide a practical and computationally efficient strategy for probability-based pattern recognition.

\subsection{Sensitivity Analysis}\label{subsec: sens}

Initial explorations revealed that ProjMC\textsuperscript{2} consistently orders estimated spatial factors by decreasing smoothness. This ordering appears driven by the algorithm's QR decomposition, as it vanishes when this step is excluded. To systematically assess the sensitivity to initial conditions and hyperparameter prior, and to further clarify the QR decomposition's role in this smoothness-based ordering, additional experiments for simulation II were conducted. These experiments employed initial values designed to challenge the algorithm's inherent ordering. Specifically, regression coefficients ($\bbeta$) and noise variances (diagonal of $\bSigma$) were initialized at their true values. However, the loading matrix $\bLambda$ was initialized using its true value but with its two rows permuted, encouraging an initial factor ordering contrary to the expected smoothness hierarchy ($\fb_1$ smoother than $\fb_2$). Furthermore, we varied the prefixed spatial decay parameters ($\phi_1, \phi_2$) across three scenarios: 1) \textbf{Test 1:} $\phi_1=6.0$, $\phi_2=9.0$. These values match those used in the data generation process, reflecting moderate prior smoothness beliefs consistent with the true factors. 2) \textbf{Test 2:} $\phi_1=9.0$, $\phi_2=3.0$. This setting imposes stronger prior smoothness on the second factor relative to the first, contradicting the true smoothness relationship. 3) \textbf{Test 3:} $\phi_1=18.0$, $\phi_2=18.0$. These larger decay values correspond to weaker, identical smoothness priors for both factors, representing a less informative scenario.

Visual inspection of the estimated factors (Fig~\ref{fig:f_compare_sensitive}) confirms the robustness of the algorithm's ordering. Across all scenarios, ProjMC\textsuperscript{2} consistently recovered factors by decreasing smoothness ($\fb_1$ smoother than $\fb_2$), overriding permuted initialisations of $\bLambda$ and the misaligned priors in Test 2. This strongly indicates the QR decomposition's dominance in enforcing smoothness-based ordering.

\begin{figure}[!ht]
  \centering
  \subfloat[\centering \scriptsize True $\fb_1$ (ProjMC\textsuperscript{2})\label{subfig:PMCMC_true_f1_2}]{%
    \includegraphics[width=0.24\textwidth, height=0.21\textwidth]{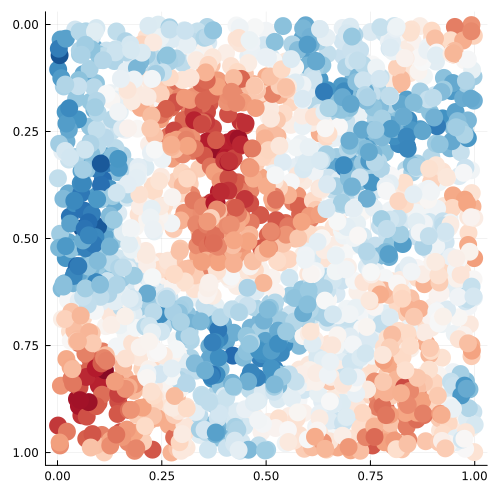}
  }
  \subfloat[\centering \scriptsize Est $\fb_1$ test 1\label{subfig:PMCMC_est_f1_t1}]{%
    \includegraphics[width=0.24\textwidth, height=0.21\textwidth]{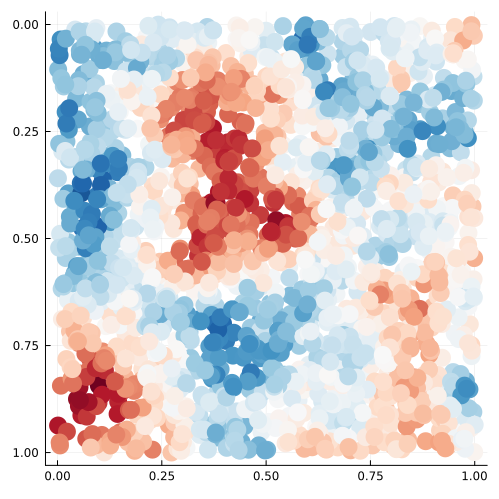}
  }
  \subfloat[\centering \scriptsize Est $\fb_1$ test 2\label{subfig:PMCMC_est_f1_t2}]{%
    \includegraphics[width=0.24\textwidth, height=0.21\textwidth]{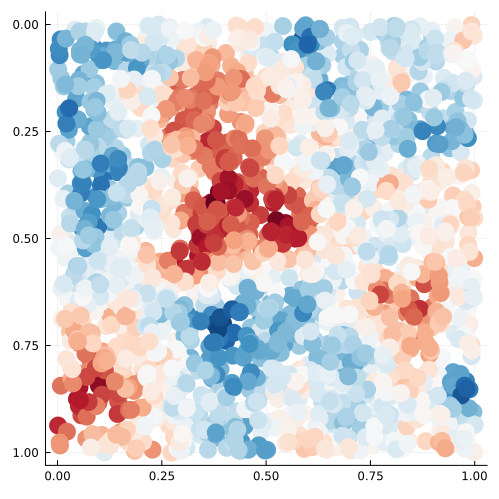}
  }
  \subfloat[\centering \scriptsize Est $\fb_1$ test 3 \label{subfig:PMCMC_est_f1_t3}]{%
    \includegraphics[width=0.24\textwidth, height=0.21\textwidth]{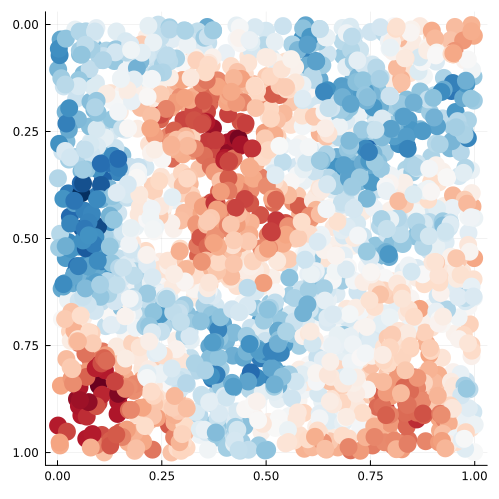}
  }
  \\
  \subfloat[\centering \scriptsize True $\fb_2$ \label{subfig:PMCMC_true_f2_2}]{%
    \includegraphics[width=0.24\textwidth, height=0.21\textwidth]{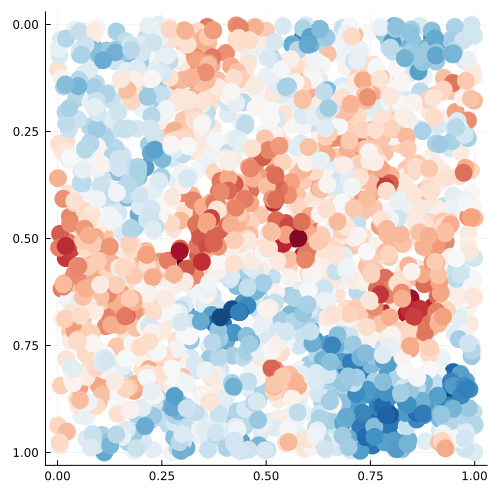}
  }
  \subfloat[\centering \scriptsize Est $\fb_2$ test 1 \label{subfig:PMCMC_est_f2_t1}]{%
    \includegraphics[width=0.24\textwidth, height=0.21\textwidth]{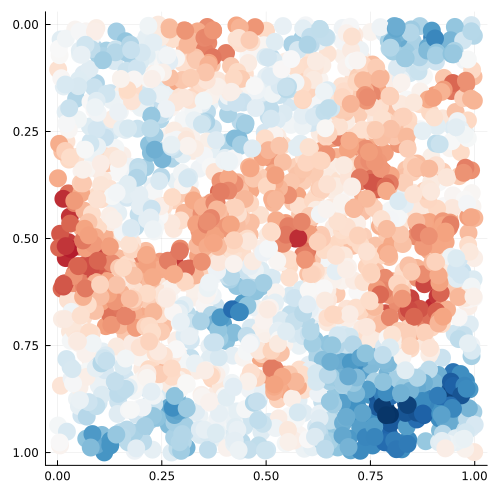}
  }
  \subfloat[\centering \scriptsize Est $\fb_2$ test 2\label{subfig:PMCMC_est_f2_t2}]{%
    \includegraphics[width=0.24\textwidth, height=0.21\textwidth]{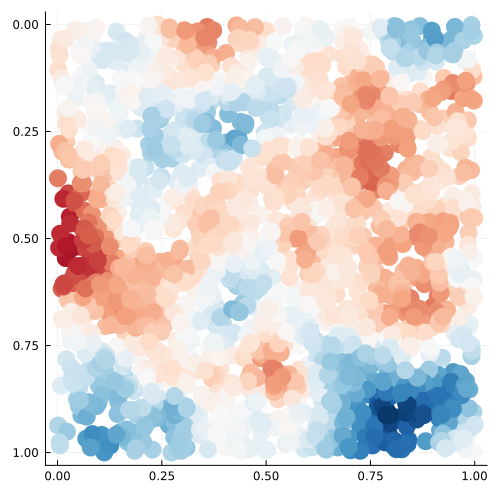}
  }
  \subfloat[\centering \scriptsize Est $\fb_2$ test 3\label{subfig:PMCMC_est_f2_t3}]{%
    \includegraphics[width=0.24\textwidth, height=0.21\textwidth]{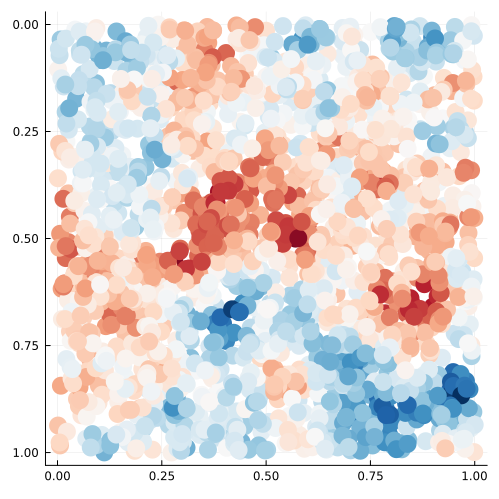}
  }

  \caption{\footnotesize Scatter plots comparing the true and estimated posterior mean of the two latent spatial factors, $\fb_1$ (top row) and $\fb_2$ (bottom row). Dot locations indicate spatial positions, and colours represent latent factor values. The first column shows the true factors projected onto the Stiefel manifold. The second to fourth columns show the estimated factors for three sensitive tests.}
  \label{fig:f_compare_sensitive}
\end{figure}

Quantitative MCMC efficiency (ESS) and latent factors inference accuracy (Euclidean distance to true factors, Spherical Variance of posterior samples) are summarised in Table~\ref{tab:ess_comparison_sens}~and~\ref{tab:latent_diag_compare_sens}, respectively. Analysis of these metrics highlights an anticipated trade-off mediated by the prior specifications. Stronger or more informative smoothness priors (Tests 1 and 2) generally lead to more concentrated posterior distributions, as evidenced by higher minimum and median ESS values (Table~\ref{tab:ess_comparison_sens}) and lower spherical variances (Table~\ref{tab:latent_diag_compare_sens}) for the challenging parameters in $\bLambda$ and $\bF$.

\begin{table}[htbp]
\centering
\footnotesize
\begin{tabular}{lccccccc}
\toprule
 & 
\multicolumn{2}{c}{\textbf{Test 1}} & 
\multicolumn{2}{c}{\textbf{Test 2}} & 
\multicolumn{2}{c}{\textbf{Test 3}} \\
\cmidrule(lr){2-3} \cmidrule(lr){4-5} \cmidrule(lr){6-7} \textbf{ESS}
&(min/mean/med)& $<100$ 
&(min/mean/med)& $<100$ 
&(min/mean/med)& $<100$ \\
\midrule
$\bbeta_0$ & 7214/11597/12272	 & 0\% & 6087/11233/12159 & 0\% & 3893/9211/9659	 & 0\% \\
$\bbeta_1$ & 6444/11110/11869 & 0\% & 5367/10817/11750 & 0\% & 3374/8669/9048 & 0\% \\
$\bLambda$ & 155/1654/234		 & 0\%  & 493/4553/2648	  & 0\% & 142/347/177 & 0\% \\
$\bF$ & 375/8899/9874 & 0\% & 3183/13399/14261	 & 0\% & 374/7528/6386 & 0\% \\
$\bSigma$ & 6247/10932/12249 & 0\%  & 5961/11316/12725	 & 0\% & 4262/9626/11125 & 0\% \\
\bottomrule
\end{tabular}
\caption{ \footnotesize Comparison of effective sample size (ESS)—reported as minimum, mean, median, and the proportion of variables with low ESS values (ESS $<100$)— across three sensitivity tests. Results are shown for the intercepts $\bbeta_0$, regression coefficients $\bbeta_0$, loading matrix $\bLambda$, matrix of latent factors $\bF$, and the noise covariance matrix $\bSigma$, based on MCMC chains of 20{,}000 iterations with the first 5{,}000 iterations discarded as warm-up. }
\label{tab:ess_comparison_sens}
\end{table}

\begin{table}[htbp]
\centering
\footnotesize
\begin{tabular}{lcccccc}
\toprule
 & 
\multicolumn{2}{c}{\textbf{Test 1}} & 
\multicolumn{2}{c} {\textbf{Test 2}} & 
\multicolumn{2}{c} {\textbf{Test 3}}  \\
\cmidrule(lr){2-3} \cmidrule(lr){4-5} \cmidrule(lr){6-7}
& Eucl. Dist. & Sphere Var. & Eucl. Dist. & Sphere Var. & Eucl. Dist. & Sphere Var. \\
\midrule
$\fb_1$ & 15.75 & 125.72 & 28.29 & 101.23 & 14.17 & 258.53\\
$\fb_2$ & 20.93 & 303.74 & 35.28 & 216.96 & 16.04 & 290.72 \\
\bottomrule
\end{tabular}
\caption{ \footnotesize Comparison of posterior summaries for the two latent spatial factors ($\fb_1$ and $\fb_2$) across three sensitivity tests. Each test is evaluated using two diagnostics: Euclidean distance between the true and estimated factor, and the spherical variance of the posterior samples.}
\label{tab:latent_diag_compare_sens}
\end{table}

Specifically, Test 2, which imposed priors misaligned with the true factor smoothness, yielded the most stable MCMC chains for $\bLambda$ and $\bF$ (highest ESS overall) and the lowest spherical variances.
 However, this stability was achieved at the expense of accuracy; Test 2 exhibited the largest Euclidean distances between the estimated and true factors. The visual results in Figure~\ref{fig:f_compare_sensitive} corroborate this, showing noticeable over-smoothing in the estimated $f_2$ for Test 2 compared to the true pattern and the estimates from the other tests. Conversely, Test 3, employing weaker priors, resulted in generally lower ESS values but achieved Euclidean distances comparable to or better than Test 1. This suggests that relaxed prior constraints can improve factor recovery accuracy by mitigating prior-induced bias, despite potentially lower sampling efficiency.

\section{Real Data Analysis}\label{sec: real_data_analy}
The application and utility of the proposed PBSF model are illustrated using a spatial transcriptomics (ST) dataset. As a probabilistic regression framework, the model facilitates comprehensive posterior inference for all parameters, enabling an interpretable analysis of spatial transcriptomics data. This is particularly valuable for exploring spatial patterns in novel datasets lacking annotations or benchmarks. This analysis used a public 10x Genomics Xenium healthy human kidney dataset \citep{10xgenomics_kidney_preview}, profiling the expression of the Xenium Human Multi-Tissue and Cancer Panel (377 genes) across 97,560 cells. An initial filter retained cells with detectable expression for at least 20 genes, resulting in a final dataset comprising 69,490 cells for analysis. The Hematoxylin and Eosin stained tissue image is in Figure~\ref{fig:jasa_fivepanel}(a). Gene expression counts were normalized using the \texttt{SCTransform} function in the \texttt{Seurat} R package.

The PBSF model was fitted with $K=6$ latent factors (embeddings). Priors for regression coefficients ($\bbeta$), factor loadings ($\bLambda$), and error covariance ($\bSigma$) were specified identically to those for the simulation studies. Spatial decay parameters for the six factors were selected from a discrete grid, allowing effective spatial ranges to vary between 2000 and 600 micrometers. A single MCMC chain was run for 1000 iterations, discarding the initial 500 iterations as burn-in. ESS diagnostics (500 post-burn-in samples) indicated satisfactory convergence for spatial factors $\bF$ (only 209 of 416,940 elements had ESS $<100$). 
Posterior mean of the latent factors, scaled by the norm of the corresponding row in the estimated loading matrix ($\hat{\bLambda}$), were extracted as low-dimensional spatial embeddings for downstream analyses. The posterior mean of the product $\bF \bLambda$ was computed to represent the estimated spatially varying component of gene expression, effectively denoised and centred.

For comparison, alternative dimension reduction and clustering techniques were applied to the same filtered and normalized dataset. These included PCA using the top $18$ principal components (selected via elbow plot) for subsequent clustering, and two contemporary graph-based methods for spatial transcriptomics: STAGATE \citep{dong2022deciphering} and GraphST \citep{long2023spatially}, identified as representative in recent benchmarks \citep{kang2025benchmarking}.
Default parameter settings were used to obtain low-dimensional embeddings from STAGATE and GraphST.

To identify distinct spatial domains, Gaussian mixture models (\texttt{mclust} R package) were applied to low-dimensional embeddings from all methods, identifying $8$ distinct clusters. The PBSF result is shown in Figure~\ref{fig:jasa_fivepanel}~(b). For closer comparison of clustering performance, focus was placed on a subregion exhibiting rich spatial heterogeneity (indicated by the black circle in Figure~\ref{fig:jasa_fivepanel}~(b)). Clustering results within this subregion for PCA, STAGATE, GraphST, and PBSF are displayed in Figure~\ref{fig:jasa_fivepanel}~(c-f). Qualitatively, PCA yielded the most fragmented clustering, albeit capturing some large-scale spatial trends. STAGATE produced clusters with well-defined boundaries and relatively simple structures. The spatial domains identified by GraphST and PBSF appear broadly similar in structure; however, PBSF produces smoother, more continuous boundaries, reducing noise at subtle interfaces while still preserving fine localized structures.

However, definitive evaluation is challenging due to the absence of ground-truth annotations. To elucidate the biological basis of the PBSF-identified clusters, the dominant gene expression patterns contributing to the estimated spatial random effects (posterior mean of $\bF\bLambda$) were investigated. For each of the 8 clusters, genes were ranked by their mean estimated spatial effect, and the top five contributors were selected. Figure~\ref{fig:jasa_fivepanel}~(g) presents the distributions (mean and 90\% credible intervals) of these estimated gene-specific spatial effects for the leading genes per cluster. This approach isolates spatially structured gene expression signals from non-spatial variation and noise, yielding more interpretable molecular signatures for each spatial domain and informing the selection of representative marker genes visualized in Figure~\ref{fig:jasa_fivepanel}~(h–m).

\begin{figure}[htbp] 
\centering
\vspace{-0.5cm}
\subfloat[\centering \scriptsize H\&E ]{%
    \includegraphics[width=0.4\textwidth]{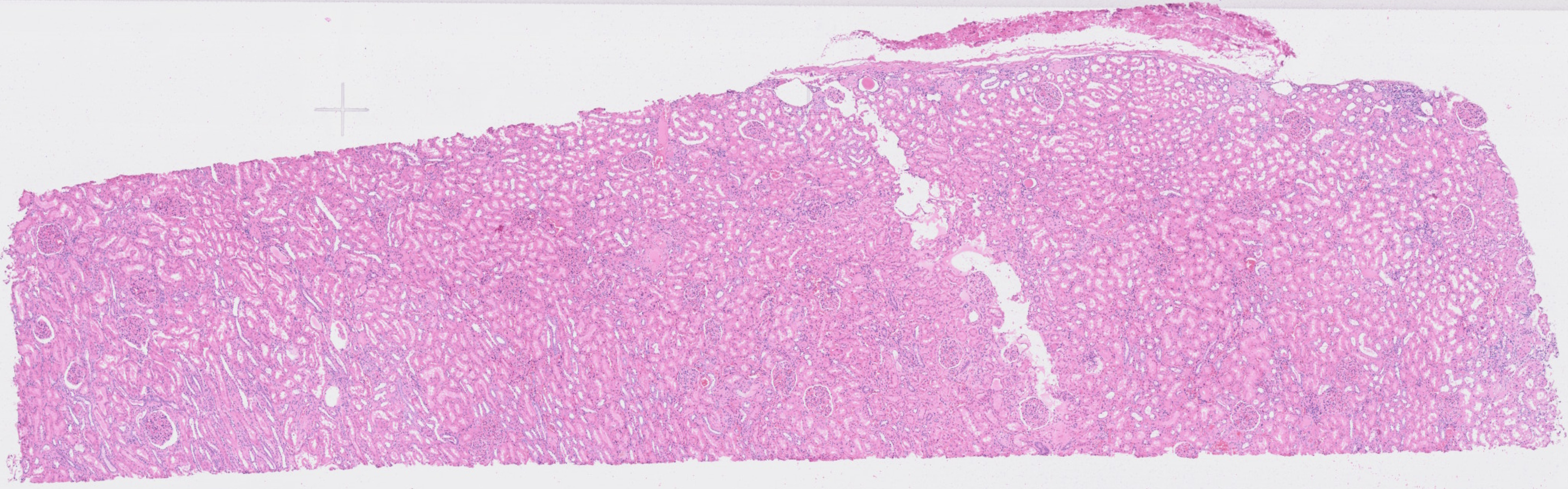}
  }
  \hfill
\subfloat[\centering \scriptsize PBSF clustering result ]{%
    \includegraphics[width=0.55\textwidth]{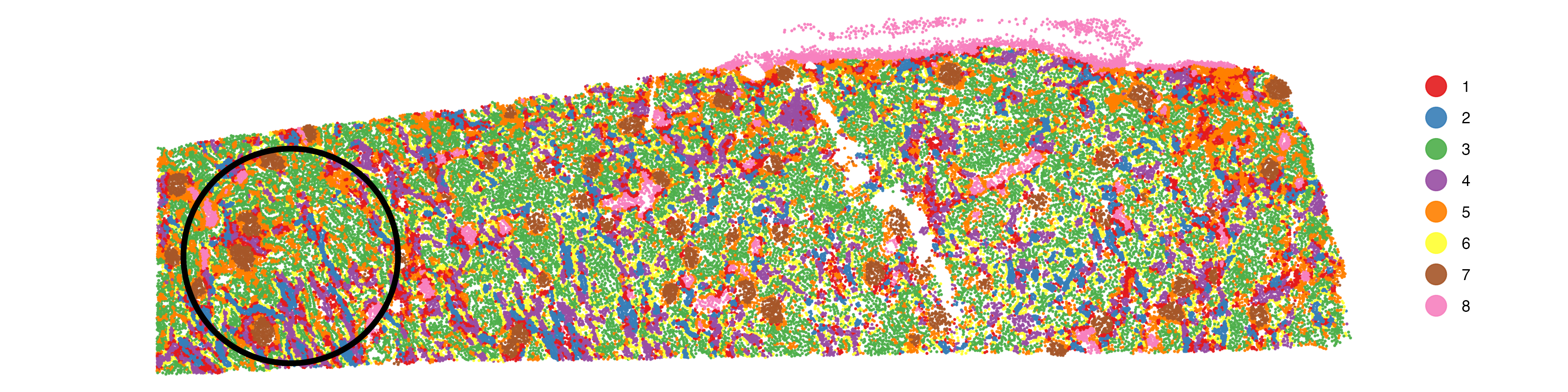}
  } \\
\subfloat[\centering \scriptsize PCA ]{%
    \includegraphics[width=0.22\textwidth]{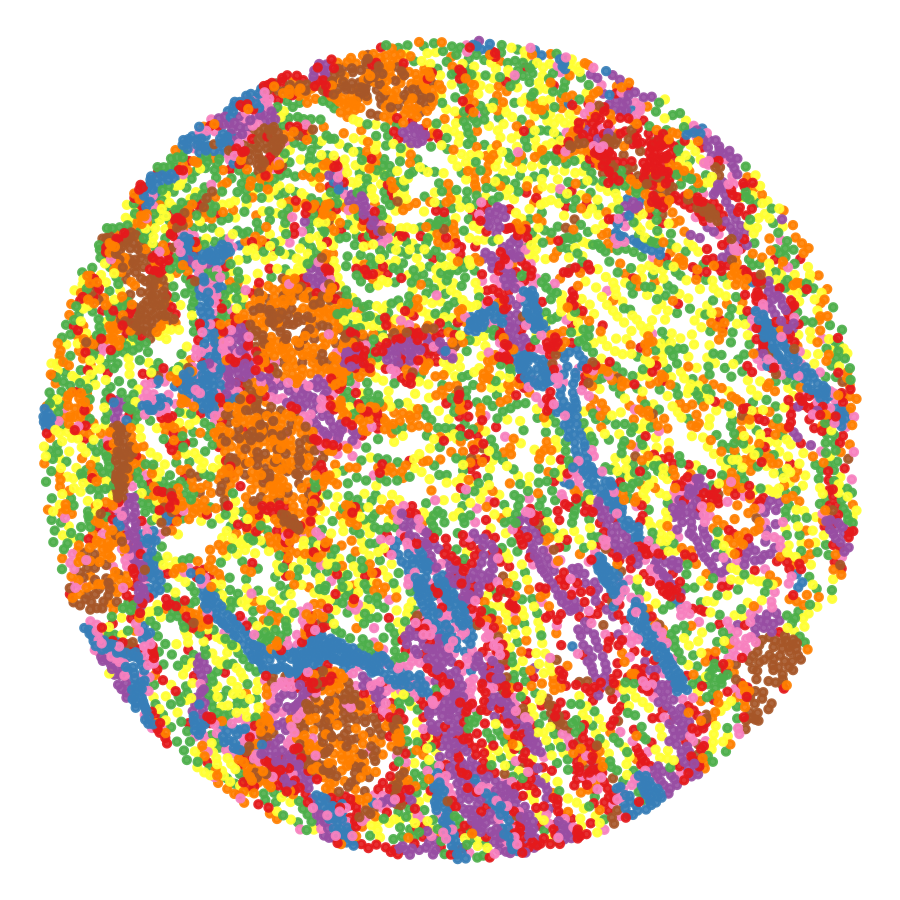}
  }
  \hfill
\subfloat[\centering \scriptsize STAGATE ]{%
    \includegraphics[width=0.22\textwidth]{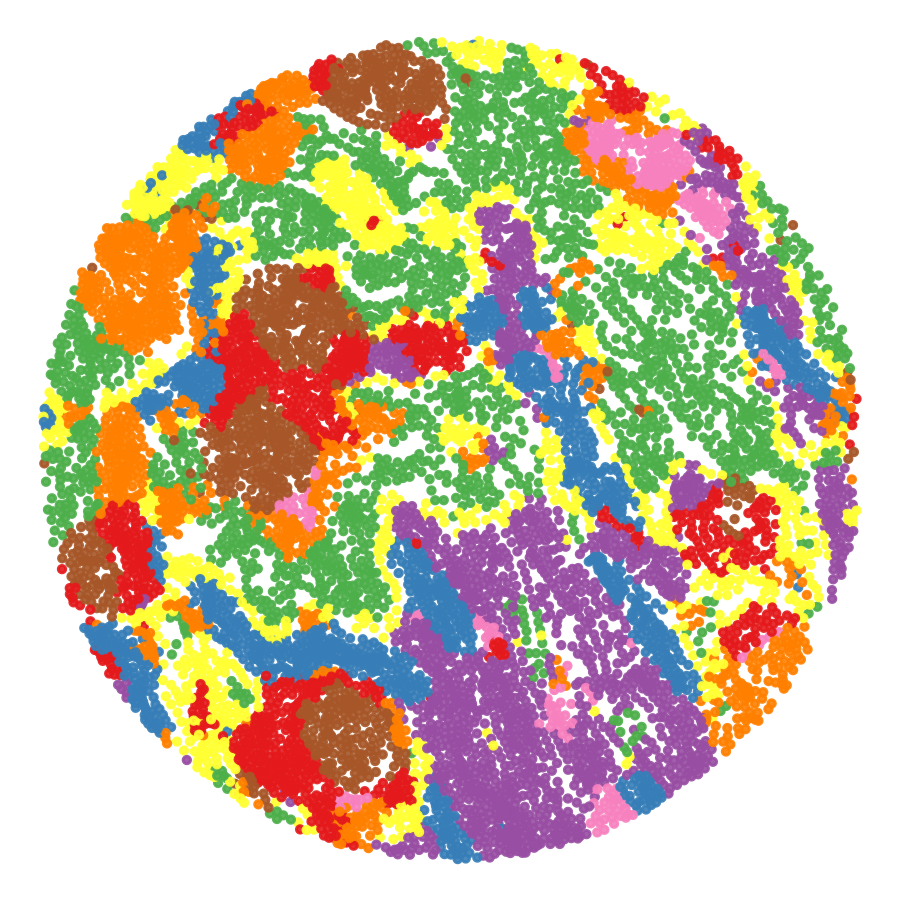}
  }
  \hfill
\subfloat[\centering \scriptsize GraphST ]{%
    \includegraphics[width=0.22\textwidth]{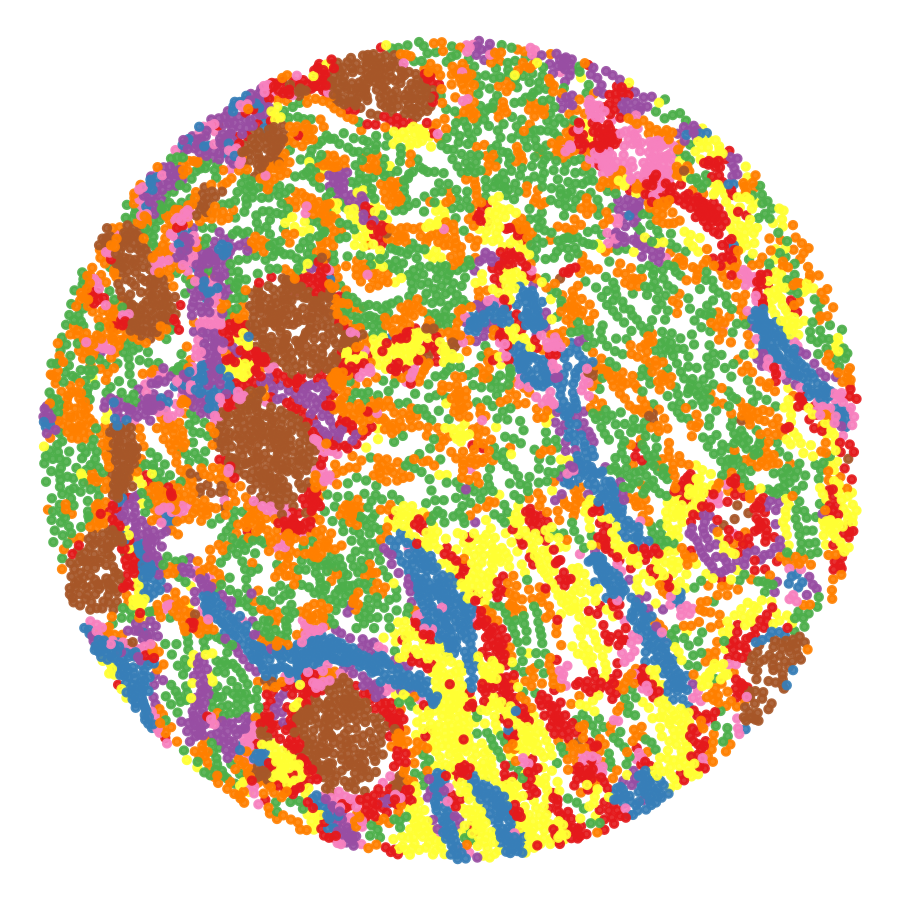}
  }
  \hfill
\subfloat[\centering \scriptsize PBSF ]{%
    \includegraphics[width=0.293\textwidth]{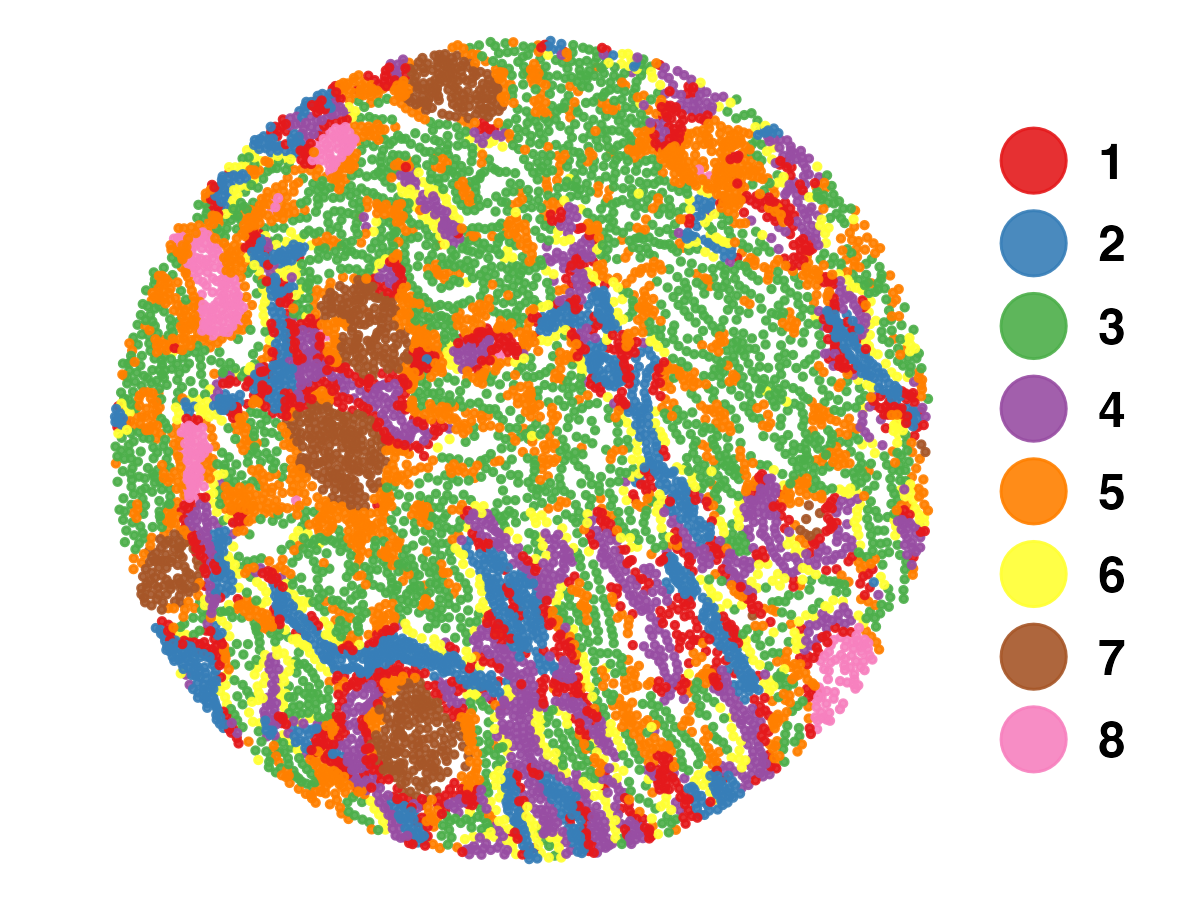}
  } \\
  \begin{minipage}[c]{0.46\textwidth}
    \centering
    \vspace*{\fill}
    \subfloat[\centering Spatially leading genes by cluster]{%
    \includegraphics[width=\textwidth]{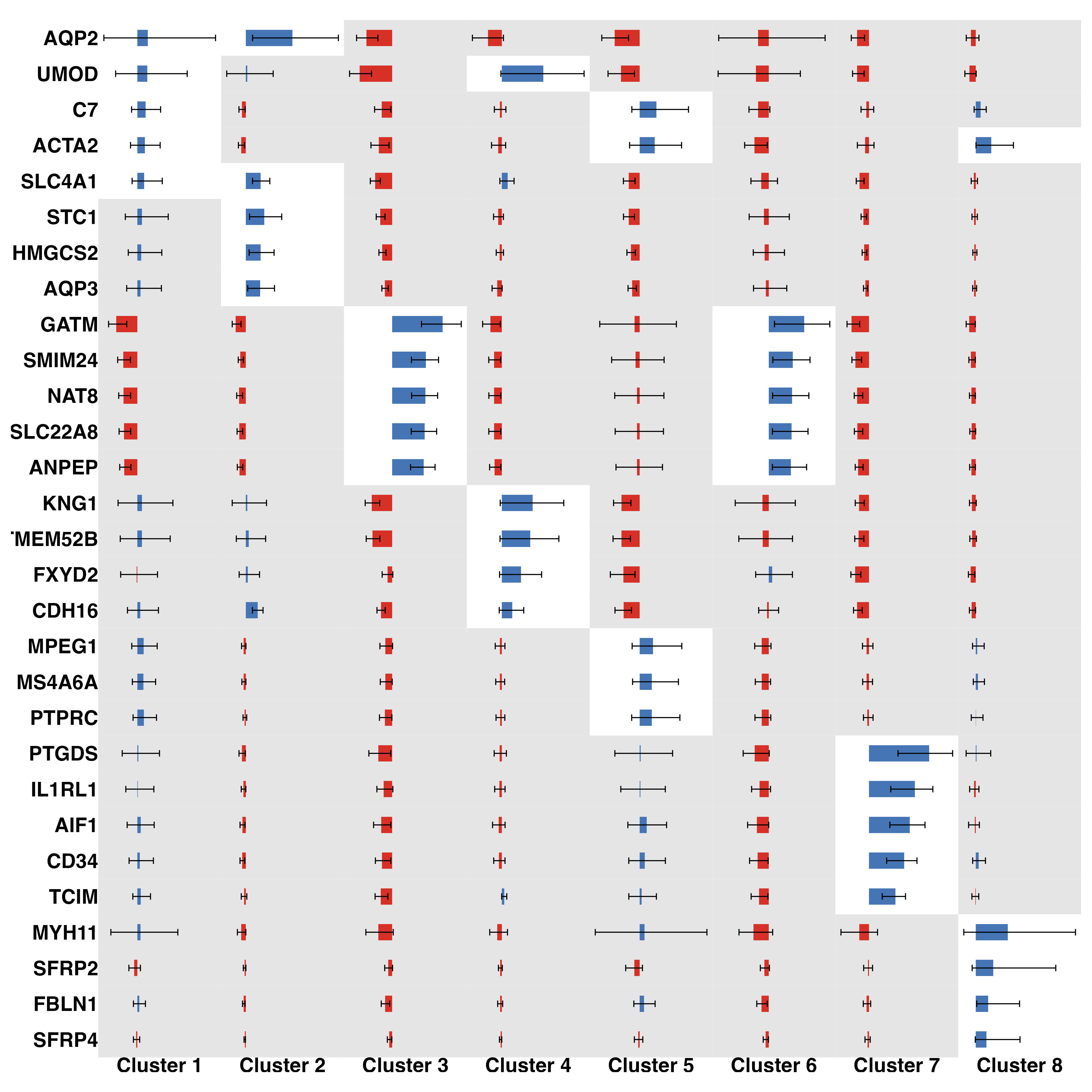}}  
    \vspace*{\fill}
  \end{minipage}
  \hfill
\begin{minipage}[c]{0.48\textwidth}
\vspace*{\fill}
    \subfloat[\centering \scriptsize AQP2]{%
    \includegraphics[width=0.33\textwidth]{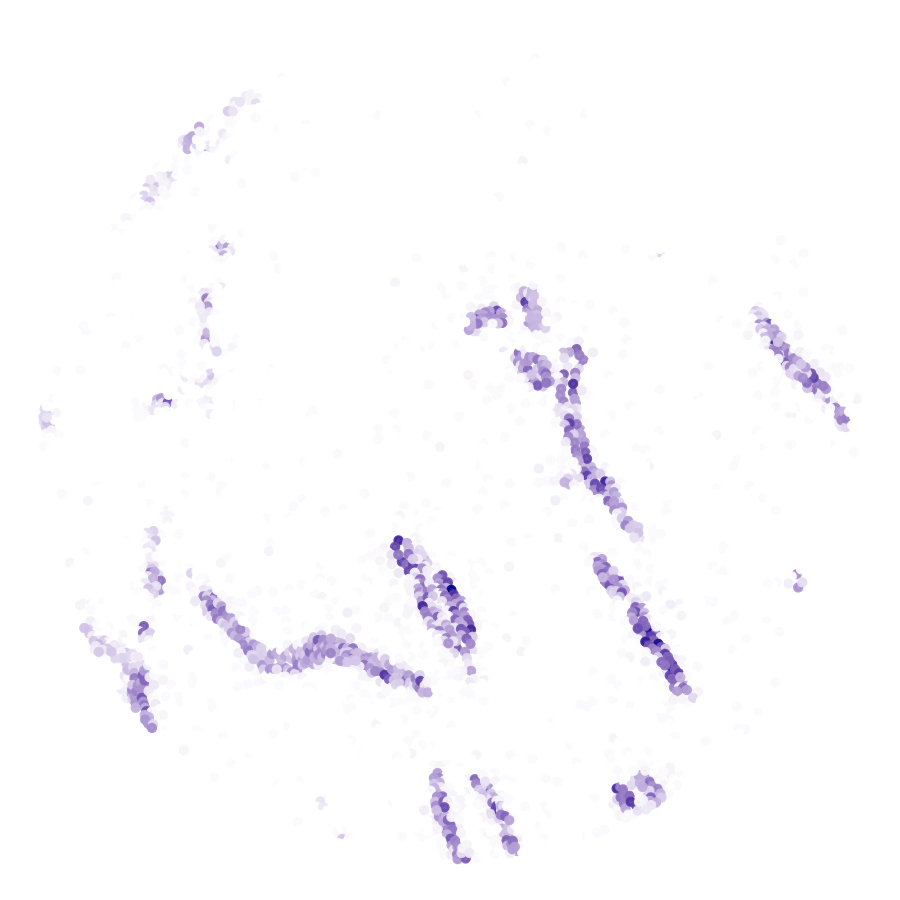}}
    \hfill
     \subfloat[\centering \scriptsize GATM]{%
    \includegraphics[width=0.33\textwidth]{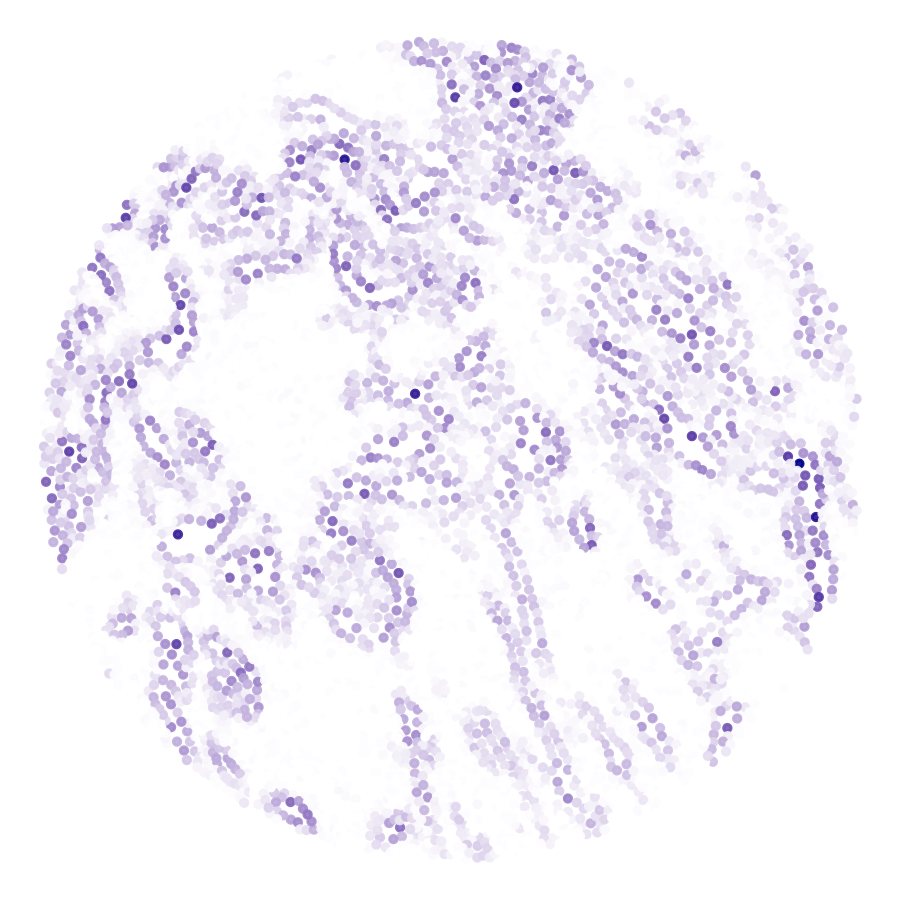}}
    \hfill
    \subfloat[\centering \scriptsize UMOD]{%
    \includegraphics[width=0.33\textwidth]{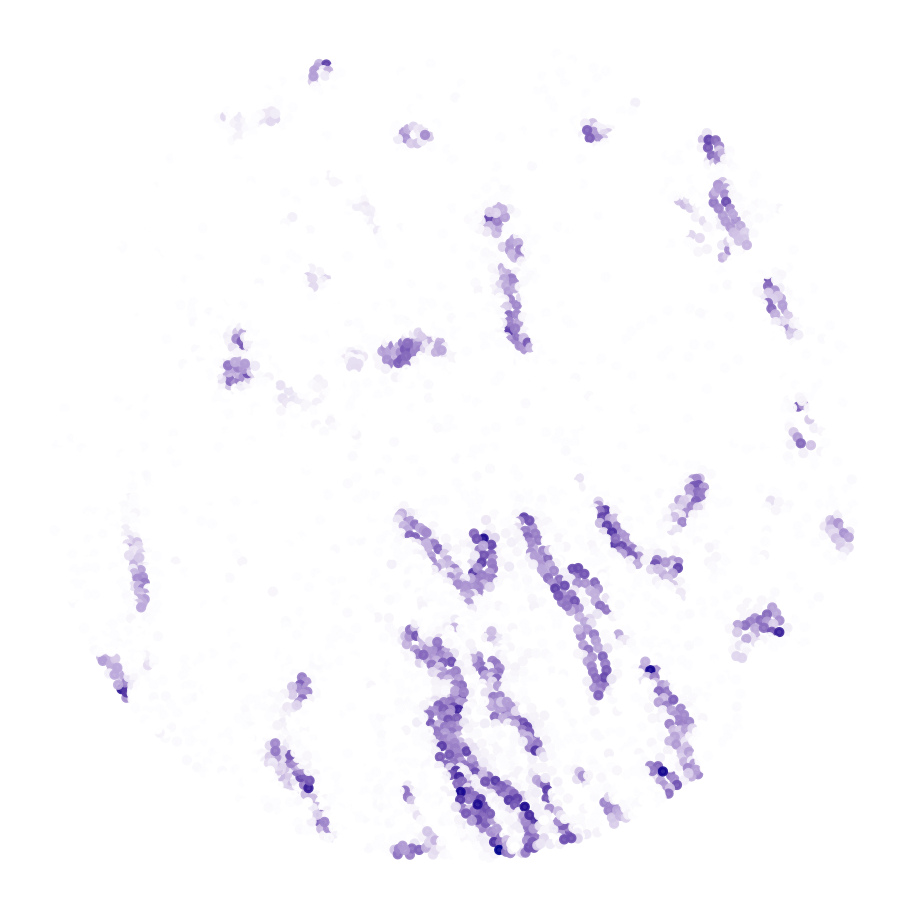}} \\
    \subfloat[\centering \scriptsize C7]{%
    \includegraphics[width=0.33\textwidth]{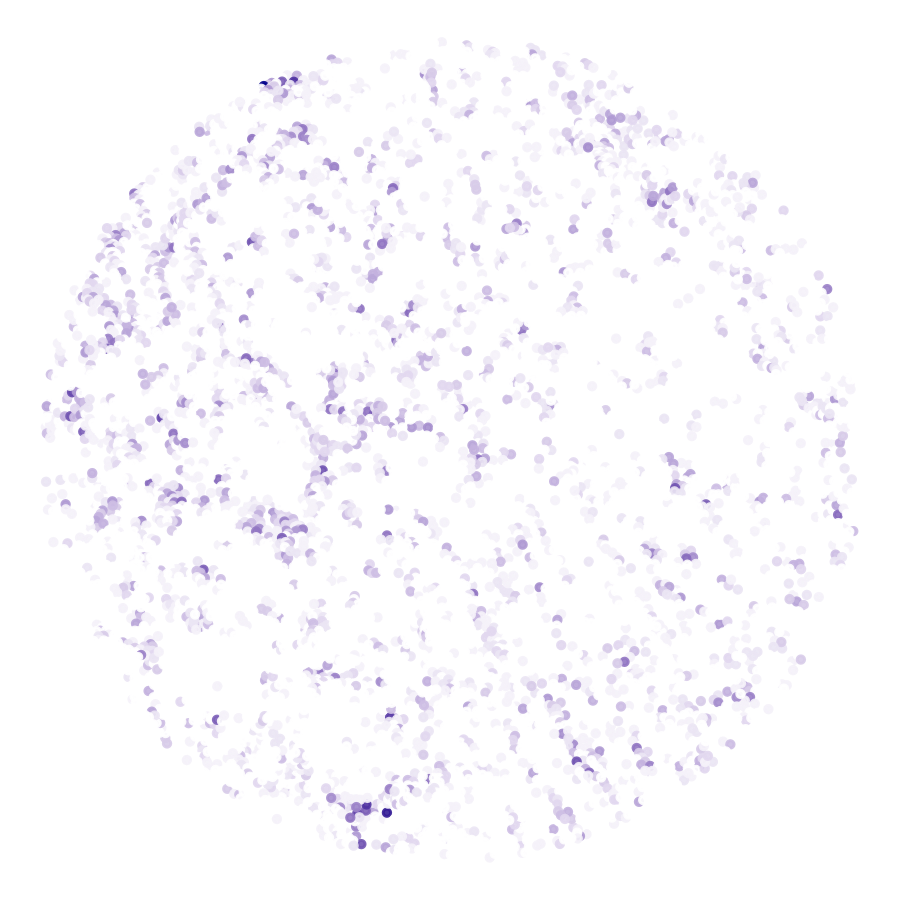}} 
    \hfill
    \subfloat[\centering \scriptsize PTGDS]{%
    \includegraphics[width=0.33\textwidth]{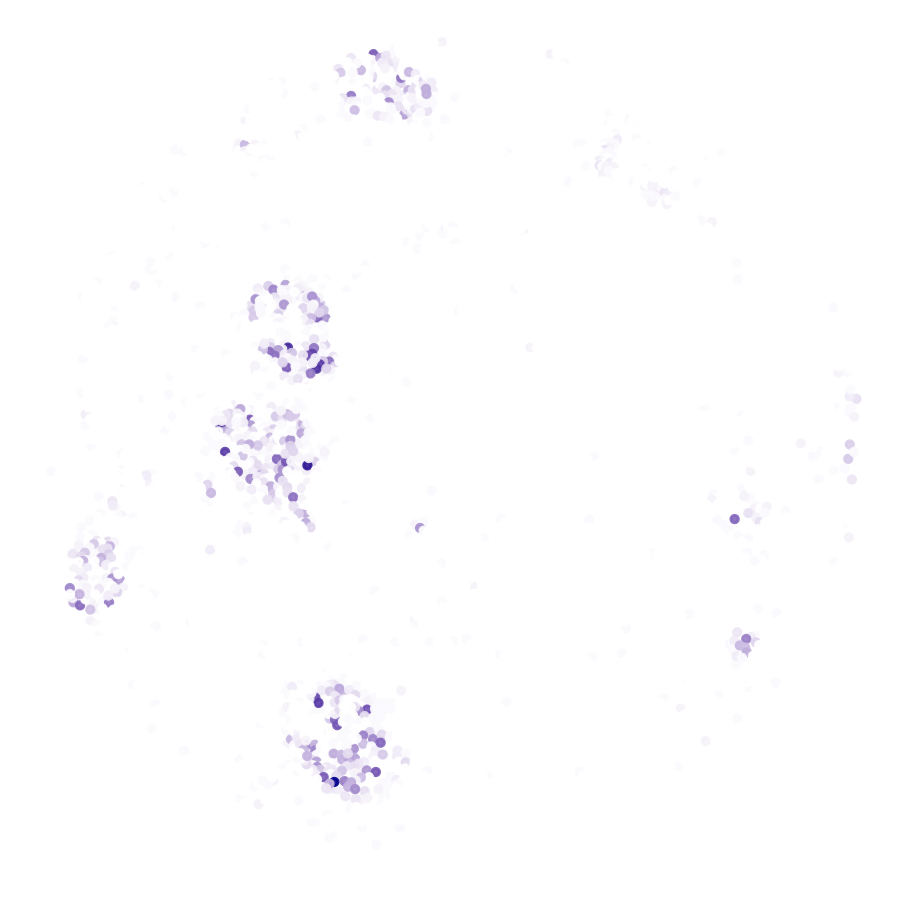}} 
    \hfill
    \subfloat[\centering \scriptsize MYH11]{%
    \includegraphics[width=0.33\textwidth]{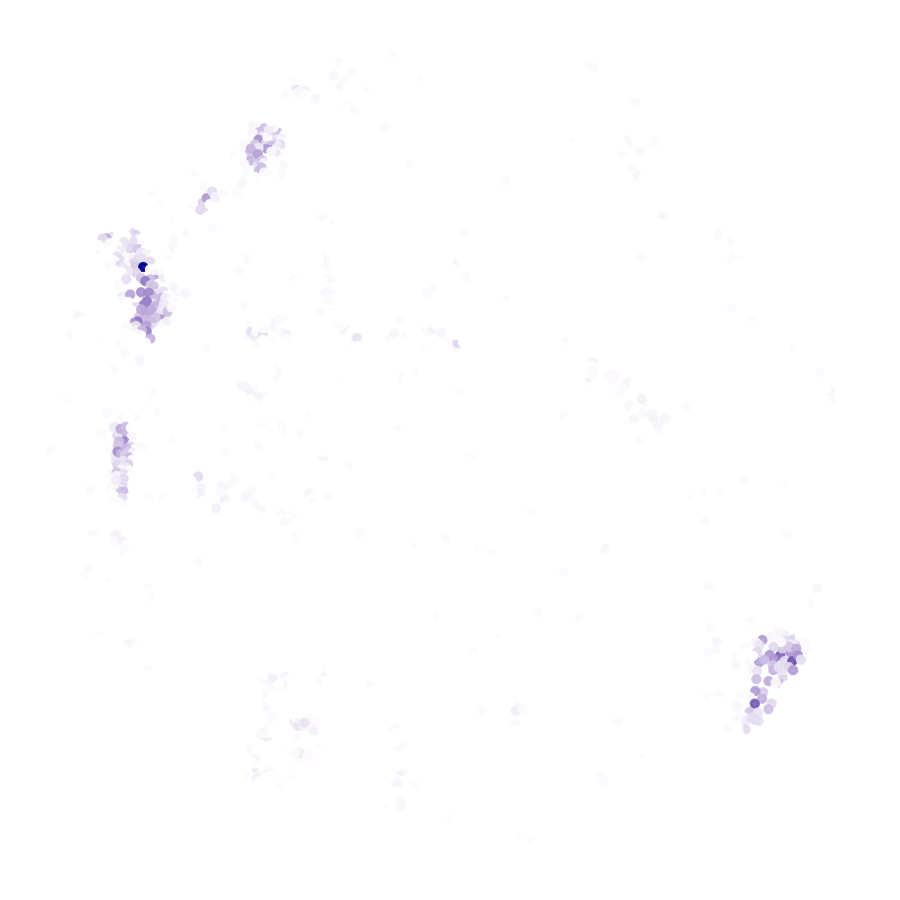}}  \\
\vspace*{\fill}
\end{minipage}\\
\subfloat[\centering \scriptsize \scriptsize $\fb_1$]{%
    \includegraphics[width=0.16\textwidth]{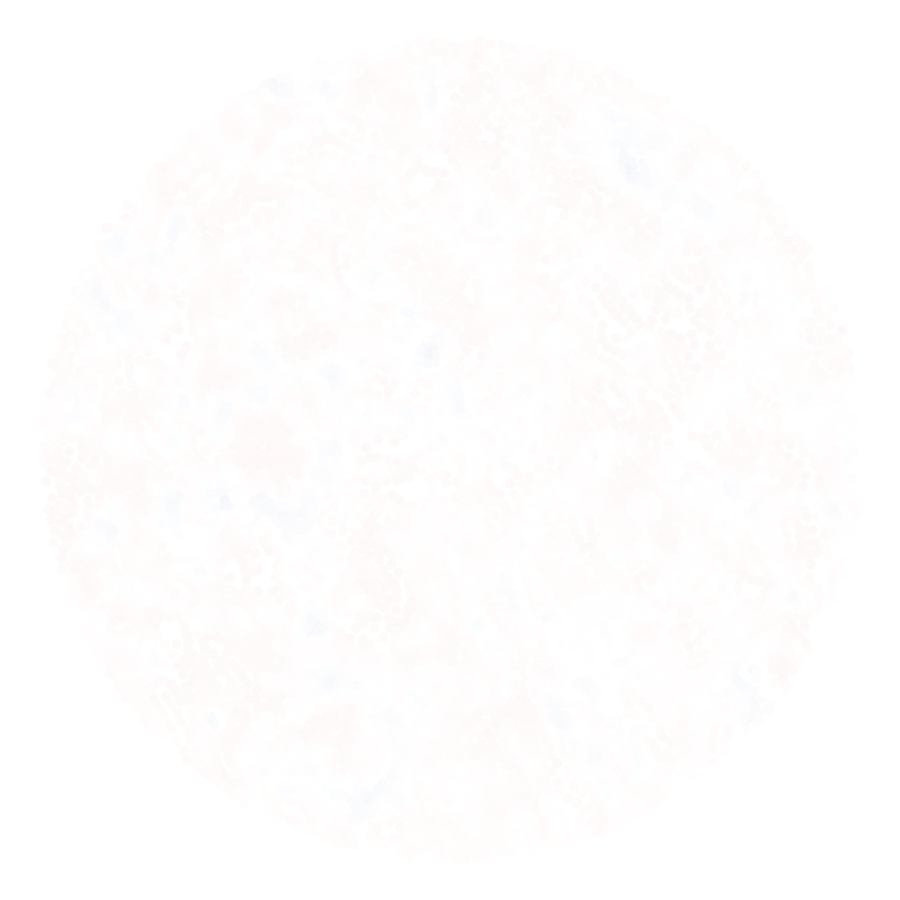}}
\hfill
\subfloat[\centering \scriptsize $\fb_2$]{%
    \includegraphics[width=0.16\textwidth]{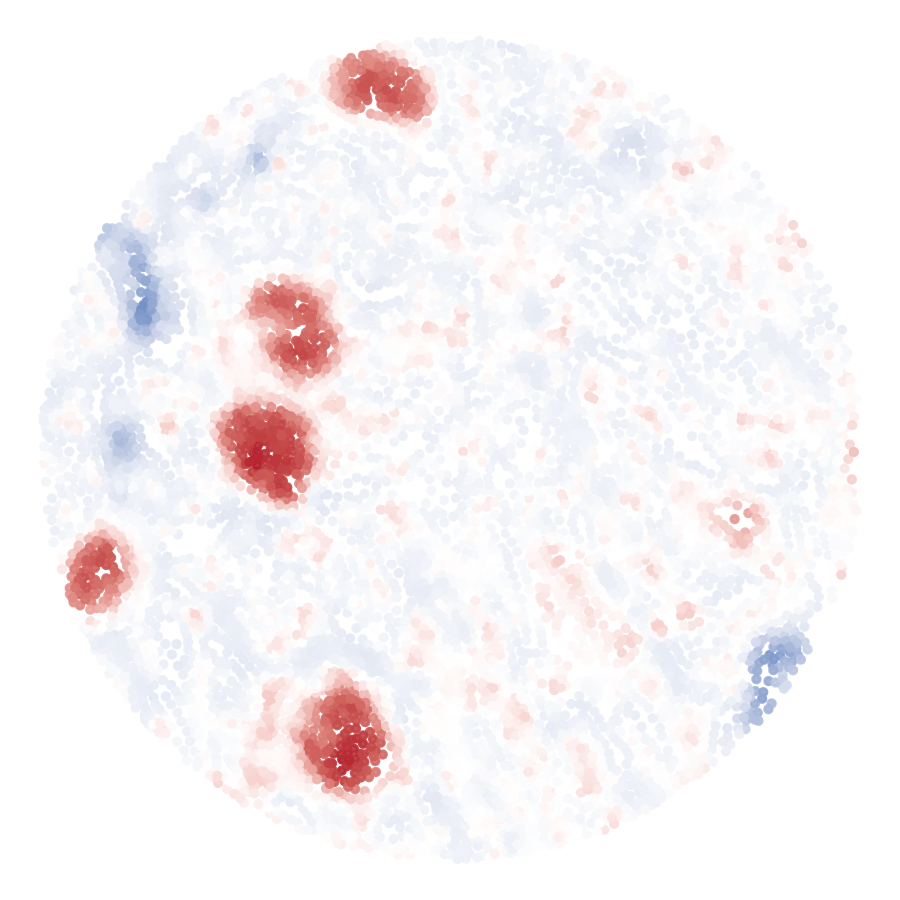}}
\hfill
\subfloat[\centering \scriptsize $\fb_3$]{%
    \includegraphics[width=0.16\textwidth]{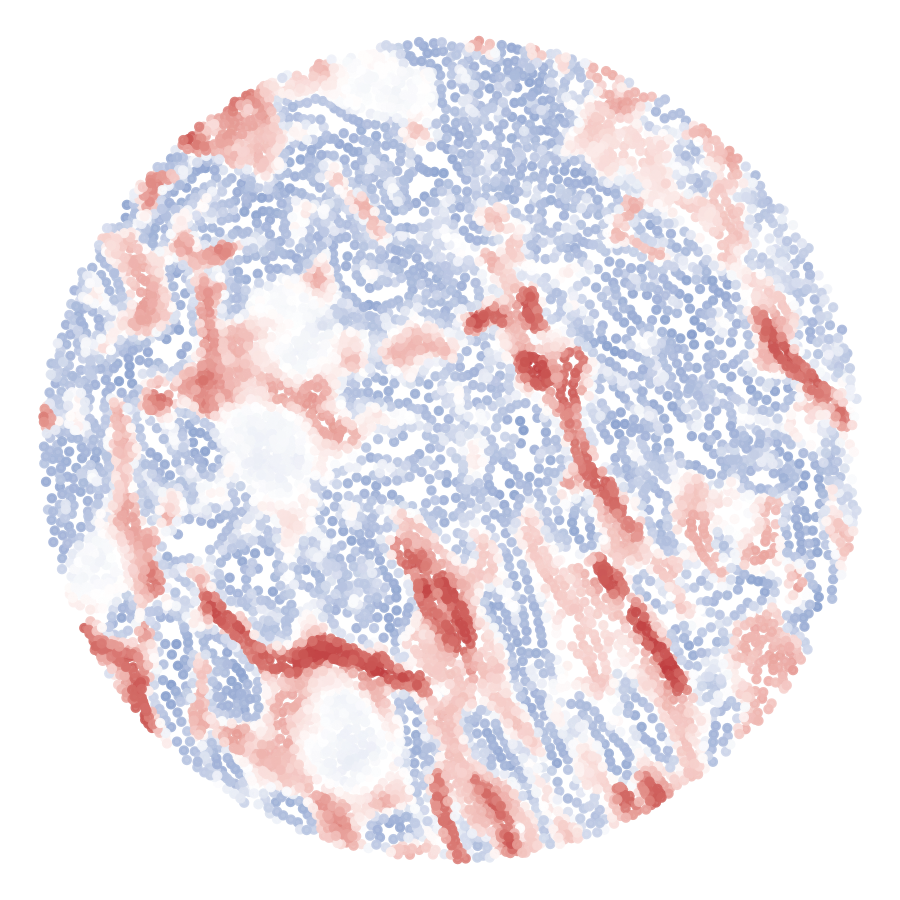}}
\hfill
\subfloat[\centering \scriptsize $\fb_4$]{%
    \includegraphics[width=0.16\textwidth]{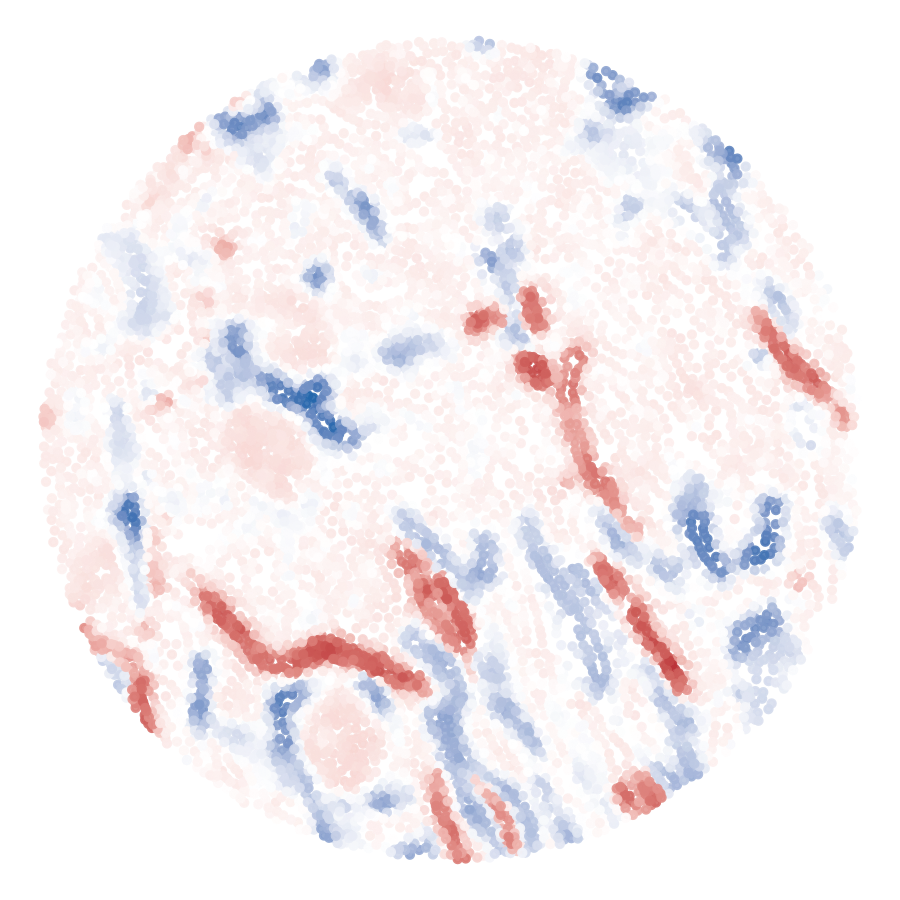}}
\hfill
\subfloat[\centering \scriptsize $\fb_5$]{%
    \includegraphics[width=0.16\textwidth]{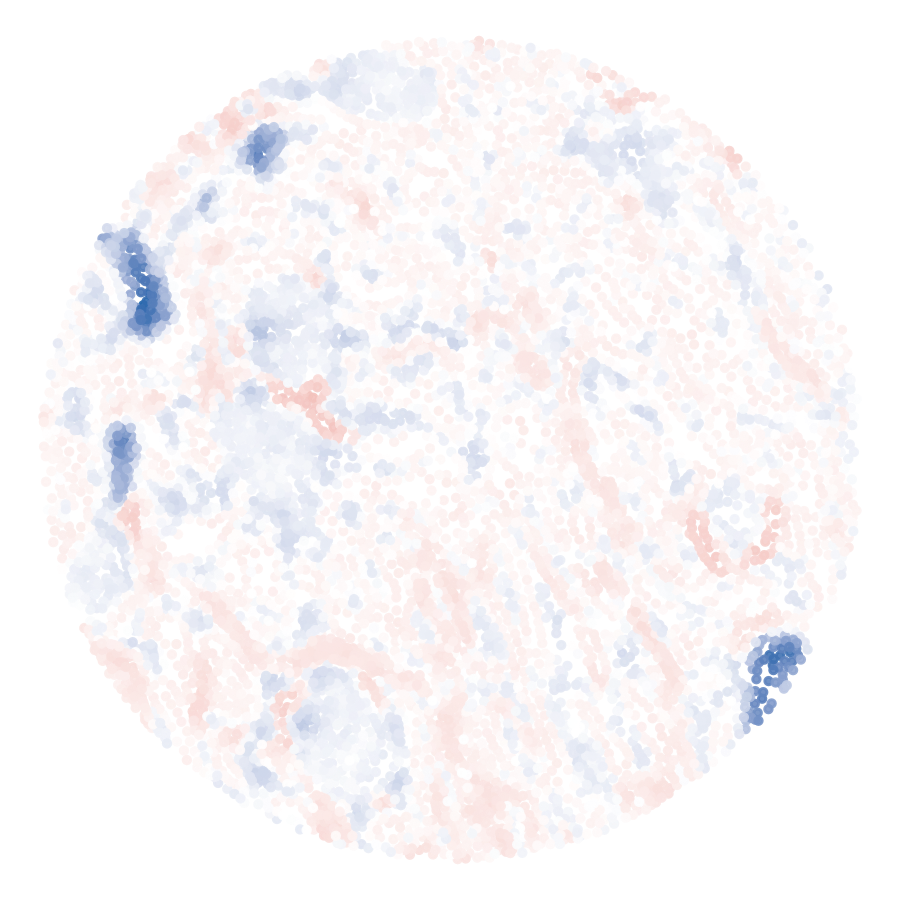}}
\hfill
\subfloat[\centering \scriptsize $\fb_6$]{%
    \includegraphics[width=0.16\textwidth]{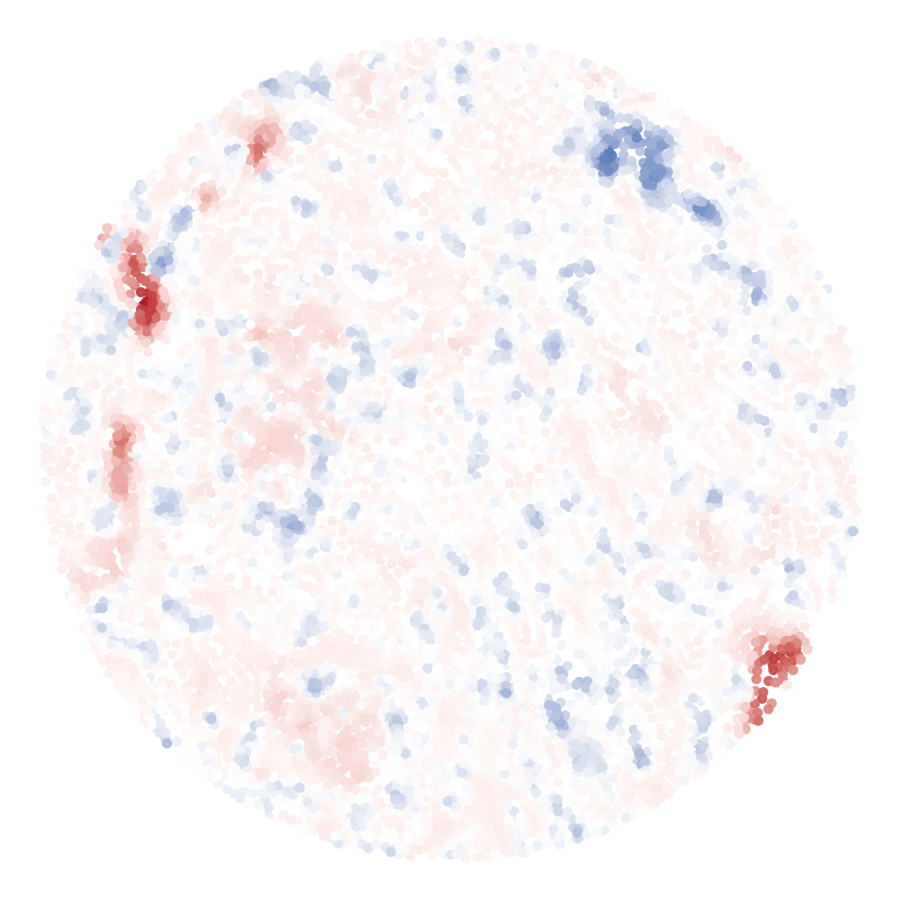}}
  \caption{\footnotesize Spatial Transcriptomics of a Human Kidney Section: Clustering, Marker Genes, and Latent Factor Patterns (a) H\&E-stained image of a Xenium-captured kidney section. (b) Unsupervised clustering on the whole tissue using posterior means of PBSF spatial factors; region of interest highlighted by a black circle.
  (c–f) Clustering within the region of interest via PCA (c), STAGATE (d), GraphST (e), and PBSF (f), using \texttt{mclust} with 8 clusters. (g) Top cluster-specific spatially variable genes identified by PBSF, based on posterior means of $\bF\bLambda$; Each bar plot shows the distribution (mean and 90\% CI) of spatial effects (beyond intercept and noise) across all cells within the cluster. (h–m) Spatial expression patterns of representative marker genes (AQP2, GATM, UMOD, C7, PTGDS, MYH11) selected based on their prominence in panel (g).(n–s) Spatial pattern of six latent spatial factors ($\fb_1$–$\fb_6$) estimated by PBSF in the selected region.}
  \label{fig:jasa_fivepanel}
\end{figure}

Based on the cluster-specific gene signatures summarised in Figure~\ref{fig:jasa_fivepanel}~(g) and the spatial expression patterns in Figure~\ref{fig:jasa_fivepanel}~(h-m), potential biological identities for several clusters can be inferred.  Cluster 2 is characterized by high AQP2 expression (Fig~\ref{fig:jasa_fivepanel}(h)), suggesting collecting ducts. Cluster 3 shows elevated GATM expression (Fig~\ref{fig:jasa_fivepanel}(i)), indicative of proximal tubule regions. Cluster 4 exhibits strong UMOD expression (Fig~\ref{fig:jasa_fivepanel}(j)), characteristic of the thick ascending limb. Cluster 7 displays high PTGDS expression (Fig~\ref{fig:jasa_fivepanel}(l)) and morphology consistent with glomeruli. Other clusters (1, 5, 6, 8) appear less distinct, either lacking a single dominant marker gene or exhibiting profiles similar to adjacent clusters (Fig~\ref{fig:jasa_fivepanel}(g)). These often correspond to boundary regions between major anatomical structures. Finally, Figure~\ref{fig:jasa_fivepanel}~(n-s) visualizes the spatial patterns of the six estimated latent factors $(\fb_1, \ldots, \fb_6)$ in the subregion, each capturing distinct, smooth spatial gene expression structures delineating clustered tissue domains. 

This application highlights several key advantages of the proposed PBSF model. First, the explicit modelling of spatial dependence through latent factors allows PBSF to effectively capture and represent smooth spatial patterns inherent in the tissue architecture (Fig~\ref{fig:jasa_fivepanel}~(n-s)). This spatial smoothing contributes to denoising the expression data, facilitating the identification of biologically relevant domains. Second, the resulting low-dimensional embeddings and estimated spatial effects ($\bF\bLambda$) provide interpretable summaries of complex ST data. Compared to PCA, the PBSF embeddings lead to more spatially coherent clusters (Fig~\ref{fig:jasa_fivepanel}~(c)~vs~\ref{fig:jasa_fivepanel}~(f)). While graph-based methods like STAGATE and GraphST also produce structured clusters, PBSF offers the additional benefits of a generative model, including posterior inference and the direct estimation of spatially varying gene contributions (Fig~\ref{fig:jasa_fivepanel}~(g)). The ability to identify marker genes based on modelled spatial effects, rather than just overall expression levels, provides a powerful tool for interpreting the biological significance of the discovered spatial domains. Thirdly, its probabilistic nature provides a principled framework for uncertainty quantification. The model can provide full posterior distributions not only for the spatial effects of specific genes within identified clusters but also for factor loadings and other model parameters, allowing for the potential incorporation of prior information and offering richer insights than point estimates alone. Computationally, on a MacBook Pro (M2 Max chip, no NVIDIA GPU), STAGATE required 1 hour 35 minutes and GraphST 3 hours 36 minutes. The PBSF MCMC (1000 iterations) completed in 5 hours 8 minutes. However, convergence diagnostics suggest the MCMC chain could potentially have been terminated earlier with satisfactory ESS. Although MCMC-based inference is more computationally intensive, the runtime remains manageable for datasets of this scale. The substantial benefits in model interpretability, principled uncertainty quantification, and direct spatial dependency modelling justify the computational investment.

\section{Conclusion and Future work}\label{sec: summary}
This paper proposed the PBSF models along with the ProjMC\textsuperscript{2} sampling algorithm. Unlike traditional approaches relying on explicit posterior distributions, ProjMC\textsuperscript{2} defines the target distribution implicitly via conditional distributions and a projection step, ensuring convergence through an existence proof. Simulations confirm enhanced convergence, mixing, and identifiability of latent factors and loadings. The inherent QR decomposition naturally orders latent factors without constraining likelihood.
PBSF models readily extend to multi-slice ST data, combining outcomes and design matrices from multiple slices. Incorporating slice-specific indicators allows effective estimation and correction of inter-slice variability, improving inference robustness.
Future extensions include integrating sparsity-inducing priors (e.g., spike-and-slab, Horseshoe) for gene selection, and adapting the framework for spatio-temporal and environmental exposure analyses. While the domain of low-dimensional embedding and pattern recognition is increasingly dominated by auto-encoder-based 'black-box' algorithms, a compelling case remains for principled, regression-based modelling. The pursuit of understanding and elucidating underlying mechanisms is a cornerstone of scientific inquiry, enabling the translation of knowledge into actionable insights. This work underscores the continued relevance and utility of interpretable statistical models in this endeavor.

\section{Acknowledgments}
This work was supported by the National Institutes of Health (NIH) under Grant P20HL176204, P30ES007048, and R01ES031590. Special thanks to Kelly N Street, Jesse A Goodrich, and Jonathan Nelson for their fruitful discussions on applications.


\appendix

\bibliographystyle{ba}  
\bibliography{lubib}

@book{cressie2015statistics,
	title={Statistics for spatio-temporal data},
	author={Cressie, Noel and Wikle, Christopher K},
	year={2015},
	publisher={John Wiley \& Sons, Hoboken, NJ}
}

@book{banerjee2014hierarchical,
	title={Hierarchical modeling and analysis for spatial data},
	author={Banerjee, Sudipto and Carlin, Bradley P and Gelfand, Alan E},
	year={2014},
	publisher={CRC Press, Boca Raton, FL}
}

@Manual{R,
    	title = {R: A Language and Environment for Statistical Computing},
    	author = {{R Core Team}},
    	organization = {R Foundation for Statistical Computing},
    	address = {Vienna, Austria},
    	year = {2017},
    	url = {https://www.R-project.org/},
    }

@article{finley2019efficient,
	title={Efficient algorithms for Bayesian Nearest Neighbor Gaussian Processes},
	author={Finley, Andrew O and Datta, Abhirup and Cook, Bruce C and Morton, Douglas C and Andersen, Hans E and Banerjee, Sudipto},
	journal={Journal of Computational and Graphical Statistics},
	volume={28},
	number={2},
	pages={401-414},
	year={2019}
}

@book{golub2012,
	title={Matrix Computations, 4th Edition},
	author={Golub, Gene H. and Van Loan, Charles F.},
	year={2012},
	publisher={Johns Hopkins University Press}
}

@article{ban08,
author={Banerjee, S. and Gelfand, A. E. and Finley, A. O. and Sang, H.},
title={{Gaussian Predictive Process Models for Large Spatial Datasets}},
journal={Journal of the Royal Statistical Society, Series B},
volume={70},
year={2008},
pages={825-848}
}

@article{stein2014,
title = "Limitations on Low Rank Approximations for Covariance Matrices of Spatial Data ",
journal = "Spatial Statistics ",
volume = "8",
number = "0",
pages = "1-19",
year = "2014",
note = "",
issn = "",
author = "Michael L. Stein"
}

@article{Finley_etall_2009,
  title={{Improving the performance of predictive process modeling for large datasets}},
  author={Finley, A. O. and Sang, H. and Banerjee, S. and Gelfand, A. E.},
  journal={Computational statistics and data analysis},
 volume={53},
  number={8},
  pages={2873--2884},
  year={2009},
  MRNUMBER = {}
}

@book{gelman2013,
Author = {Andrew Gelman and John B. Carlin and Hal S. Stern and David B. Dunson and Aki Vehtari and Donald B. Rubin},
Title = {Bayesian Data Analysis, 3rd Edition},
Series = {Chapman \& Hall/CRC Texts in Statistical Science},
Publisher = {Chapman \& Hall/CRC},
Year = {2013}
}

@INBOOK{stein99,
AUTHOR = {Stein, M. L.},
TITLE = {Interpolation of Spatial Data: Some Theory for Kriging},
   edition = "First",
   publisher = "Springer",
   year = {1999},
   }

@ARTICLE{datta16,
    author = {Datta, A. and Banerjee, S. and Finley, A. O. and Gelfand, A. E.},
    title = {Hierarchical nearest-neighbor Gaussian process models for large geostatistical datasets},
    journal = {Journal of the American Statistical Association},
    volume = {111},
    year = {2016},
    pages={800--812}
}

@article{guinness2018permutation,
  title={Permutation and grouping methods for sharpening {G}aussian process approximations},
  author={Guinness, Joseph},
  journal={Technometrics},
  volume={60},
  number={4},
  pages={415--429},
  year={2018},
  publisher={Taylor \& Francis}
}

@article{katzfuss2021general,
  title={{A general framework for Vecchia approximations of Gaussian processes}},
  author={Katzfuss, Matthias and Guinness, Joseph},
  journal={Statistical Science},
  volume={36},
  number={1},
  pages={124--141},
  year={2021},
  publisher={JSTOR}
}

@article{renbanerjee2013,
   author = {Ren, Q. and Banerjee, S.},
   title = {Hierarchical factor models for large spatially misaligned datasets: A low-rank predictive process approach.},
   journal = {Biometrics},
   volume = {69},
   number = {},
   pages = {19-30},
   year = {2013},
   type = {Journal Article}
}

@ARTICLE{Lopes2008,
  author = {Hedibert Freitas Lopes and Esther Salazar and Dani Gamerman},
  title = {Spatial Dynamic Factor Analysis},
  journal = {Bayesian Analysis},
  year = {2008},
  volume = {3(4)},
  pages = {759 - 792}
}

@article{zdb2019,
  author = {Lu Zhang and Abhirup Datta and Sudipto Banerjee},
  title = {Practical Bayesian Modeling and Inference for Massive Spatial Datasets On Modest Computing Environments},
  journal = {Statistical Analysis and Data Mining: The ASA Data Science Journal},
  volume = {12},
  number = {3},
  pages = {197-209},
  year = {2019},
  doi = {10.1002/sam.11413}
  }

@article{fong2011lsmr,
  title={{LSMR: An iterative algorithm for sparse least-squares problems}},
  author={Fong, David Chin-Lung and Saunders, Michael},
  journal={SIAM Journal on Scientific Computing},
  volume={33},
  number={5},
  pages={2950--2971},
  year={2011},
  publisher={SIAM}
}

@article{nishimura2023prior,
  title={Prior-preconditioned conjugate gradient method for accelerated {G}ibbs sampling in ``large n, large p'' {B}ayesian sparse regression},
  author={Nishimura, Akihiko and Suchard, Marc A},
  journal={Journal of the American Statistical Association},
  volume={118},
  number={544},
  pages={2468--2481},
  year={2023},
  publisher={Taylor \& Francis}
}

@article{taylor2019spatial,
  title={Spatial factor models for high-dimensional and large spatial data: An application in forest variable mapping},
  author={Taylor-Rodriguez, Daniel and Finley, Andrew O and Datta, Abhirup and Babcock, Chad and Andersen, Hans Erik and Cook, Bruce D and Morton, Douglas C and Banerjee, Sudipto},
  journal={Statistica Sinica},
  volume={29},
  number={3},
  pages={1155--1180},
  year={2019},
  publisher={Institute of Statistical Science}
}

@article{dawid1981,
    author = {Dawid, A. P.},
    title = "{Some matrix-variate distribution theory: Notational considerations and a Bayesian application}",
    journal = {Biometrika},
    volume = {68},
    number = {1},
    pages = {265-274},
    year = {1981},
    month = {04},
    issn = {0006-3444},
    doi = {10.1093/biomet/68.1.265}
}

@article{zhang2022spatial,
  title={{Spatial factor modeling: A Bayesian matrix-normal approach for misaligned data}},
  author={Zhang, Lu and Banerjee, Sudipto},
  journal={Biometrics},
  volume={78},
  number={2},
  pages={560--573},
  year={2022},
  publisher={Wiley Online Library}
}

@book{Matern86,
  title={Spatial {V}ariation},
  author={Mat\'ern, B.},
  publisher={Springer-Verlag},
  year={1986},
}

@article{Zhang04,
  title={Inconsistent estimation and asymptotically equal interpolations in model-based geostatistics},
  author={Zhang, Hao},
  journal={Journal of the American Statistical Association},
  volume={99},
  number={465},
  pages={250--261},
  year={2004},
}

@article{zhang2005towards,
  title={Towards reconciling two asymptotic frameworks in spatial statistics},
  author={Zhang, Hao and Zimmerman, Dale L},
  journal={Biometrika},
  volume={92},
  number={4},
  pages={921--936},
  year={2005},
  publisher={Oxford University Press}
}

@article{DZM09,
  title={Fixed-domain asymptotic properties of tapered maximum likelihood estimators},
  author={Du, J. and Zhang, H. and Mandrekar, V. S.},
  journal={The Annals of Statistics},
  volume={37},
  number={6A},
  pages={3330--3361},
  year={2009},
}

@article{tang2021identifiability,
  title={On identifiability and consistency of the nugget in {G}aussian spatial process models},
  author={Tang, Wenpin and Zhang, Lu and Banerjee, Sudipto},
  journal={Journal of the Royal Statistical Society Series B: Statistical Methodology},
  volume={83},
  number={5},
  pages={1044--1070},
  year={2021},
  publisher={Oxford University Press}
}

@article{chakraborty2019statistics,
  title={{Statistics on the Stiefel manifold: Theory and applications}},
  author={Chakraborty, Rudrasis and Vemuri, Baba C},
  year={2019}
}

@book{robert1999monte,
  title={Monte Carlo statistical methods},
  author={Robert, Christian P and Casella, George and Casella, George},
  volume={2},
  year={1999},
  publisher={Springer}
}

@book{meyn2012markov,
  title={Markov chains and stochastic stability},
  author={Meyn, Sean P and Tweedie, Richard L},
  year={2012},
  publisher={Springer Science \& Business Media}
}

@article{zhang2022applications,
  title={{Applications of Conjugate Gradient in Bayesian Computation}},
  author={Zhang, Lu},
  journal={Wiley StatsRef: Statistics Reference Online},
  pages={1--7},
  year={2022},
  publisher={Wiley Online Library}
}

@article{halko2011algorithm,
  title={An algorithm for the principal component analysis of large data sets},
  author={Halko, Nathan and Martinsson, Per-Gunnar and Shkolnisky, Yoel and Tygert, Mark},
  journal={SIAM Journal on Scientific computing},
  volume={33},
  number={5},
  pages={2580--2594},
  year={2011},
  publisher={SIAM}
}

@article{martinsson2020randomized,
  title={{Randomized numerical linear algebra: Foundations and algorithms}},
  author={Martinsson, Per-Gunnar and Tropp, Joel A},
  journal={Acta Numerica},
  volume={29},
  pages={403--572},
  year={2020},
  publisher={Cambridge University Press}
}

@article{bentley1975multidimensional,
  title={Multidimensional binary search trees used for associative searching},
  author={Bentley, Jon Louis},
  journal={Communications of the ACM},
  volume={18},
  number={9},
  pages={509--517},
  year={1975},
  publisher={ACM New York, NY, USA}
}

@article{dong2022deciphering,
  title={Deciphering spatial domains from spatially resolved transcriptomics with an adaptive graph attention auto-encoder},
  author={Dong, Kangning and Zhang, Shihua},
  journal={Nature communications},
  volume={13},
  number={1},
  pages={1739},
  year={2022},
  publisher={Nature Publishing Group UK London}
}

@book{mardia2009directional,
  title={Directional statistics},
  author={Mardia, Kanti V and Jupp, Peter E},
  year={2009},
  publisher={John Wiley \& Sons}
}

@article{schäfer2021sparse,
  title={{Sparse Cholesky Factorization by Kullback--Leibler Minimization}},
  author={Schäfer, Florian and Katzfuss, Matthias and Owhadi, Houman},
  journal={SIAM Journal on scientific computing},
  volume={43},
  number={3},
  pages={A2019--A2046},
  year={2021},
  publisher={SIAM}
}

@article{stephens2000dealing,
  title={Dealing with label switching in mixture models},
  author={Stephens, Matthew},
  journal={Journal of the Royal Statistical Society: Series B (Statistical Methodology)},
  volume={62},
  number={4},
  pages={795--809},
  year={2000},
  publisher={Wiley Online Library}
}

@misc{10xgenomics_kidney_preview,
  author = {{10x Genomics}},
  title = {Human Kidney Preview Data - Xenium Human Multi-Tissue and Cancer Panel 1 Standard},
  year = {2024},
  howpublished = {\url{https://www.10xgenomics.com/datasets/human-kidney-preview-data-xenium-human-multi-tissue-and-cancer-panel-1-standard}},
  note = {Accessed: 2025-04-28}
}

@article{long2023spatially,
  title={{Spatially informed clustering, integration, and deconvolution of spatial transcriptomics with GraphST}},
  author={Long, Yahui and Ang, Kok Siong and Li, Mengwei and Chong, Kian Long Kelvin and Sethi, Raman and Zhong, Chengwei and Xu, Hang and Ong, Zhiwei and Sachaphibulkij, Karishma and Chen, Ao and others},
  journal={Nature Communications},
  volume={14},
  number={1},
  pages={1155},
  year={2023},
  publisher={Nature Publishing Group UK London}
}

@article{kang2025benchmarking,
  title={Benchmarking computational methods for detecting spatial domains and domain-specific spatially variable genes from spatial transcriptomics data},
  author={Kang, Liping and Zhang, Qinglong and Qian, Fan and Liang, Junyao and Wu, Xiaohui},
  journal={Nucleic Acids Research},
  volume={53},
  number={7},
  pages={gkaf303},
  year={2025},
  publisher={Oxford University Press}
}

@article{wang2003generalized,
  title={Generalized common spatial factor model},
  author={Wang, Fujun and Wall, Melanie M},
  journal={Biostatistics},
  volume={4},
  number={4},
  pages={569--582},
  year={2003},
  publisher={Oxford University Press}
}

@article{townes2023nonnegative,
  title={Nonnegative spatial factorization applied to spatial genomics},
  author={Townes, F William and Engelhardt, Barbara E},
  journal={Nature methods},
  volume={20},
  number={2},
  pages={229--238},
  year={2023},
  publisher={Nature Publishing Group US New York}
}

@article{velten2022identifying,
  title={Identifying temporal and spatial patterns of variation from multimodal data using MEFISTO},
  author={Velten, Britta and Braunger, Jana M and Argelaguet, Ricard and Arnol, Damien and Wirbel, Jakob and Bredikhin, Danila and Zeller, Georg and Stegle, Oliver},
  journal={Nature methods},
  volume={19},
  number={2},
  pages={179--186},
  year={2022},
  publisher={Nature Publishing Group US New York}
}

@article{shang2022spatially,
  title={Spatially aware dimension reduction for spatial transcriptomics},
  author={Shang, Lulu and Zhou, Xiang},
  journal={Nature communications},
  volume={13},
  number={1},
  pages={7203},
  year={2022},
  publisher={Nature Publishing Group UK London}
}

@inproceedings{titsias2009variational,
  title={Variational learning of inducing variables in sparse Gaussian processes},
  author={Titsias, Michalis},
  booktitle={Artificial intelligence and statistics},
  pages={567--574},
  year={2009},
  organization={PMLR}
}

@article{zhang2024fixed,
  title={{Fixed-domain asymptotics under Vecchia’s approximation of spatial process likelihoods}},
  author={Zhang, Lu and Tang, Wenpin and Banerjee, Sudipto},
  journal={Statistica Sinica},
  volume={34},
  number={4},
  pages={1863},
  year={2024}
}

@article{moses2022museum,
  title={Museum of spatial transcriptomics},
  author={Moses, Lambda and Pachter, Lior},
  journal={Nature methods},
  volume={19},
  number={5},
  pages={534--546},
  year={2022},
  publisher={Nature Publishing Group US New York}
}

@article{bressan2023dawn,
  title={The dawn of spatial omics},
  author={Bressan, Dario and Battistoni, Giorgia and Hannon, Gregory J},
  journal={Science},
  volume={381},
  number={6657},
  pages={eabq4964},
  year={2023},
  publisher={American Association for the Advancement of Science}
}

@article{lee2025spatial,
  title={{Spatial Omics in Clinical Research: A Comprehensive Review of Technologies and Guidelines for Applications}},
  author={Lee, Yoonji and Lee, Mingyu and Shin, Yoojin and Kim, Kyuri and Kim, Taejung},
  journal={International Journal of Molecular Sciences},
  volume={26},
  number={9},
  pages={3949},
  year={2025},
  publisher={MDPI}
}

@article{dai2024spatial,
  title={{Spatial source apportionment of airborne coarse particulate matter using PMF-Bayesian receptor model}},
  author={Dai, Tianjiao and Dai, Qili and Yin, Jingchen and Chen, Jiajia and Liu, Baoshuang and Bi, Xiaohui and Wu, Jianhui and Zhang, Yufen and Feng, Yinchang},
  journal={Science of The Total Environment},
  volume={917},
  pages={170235},
  year={2024},
  publisher={Elsevier}
}

@article{song2018multivariate,
  title={Multivariate linear regression model for source apportionment and health risk assessment of heavy metals from different environmental media},
  author={Song, Yinxian and Li, Huimin and Li, Jizhou and Mao, Changping and Ji, Junfeng and Yuan, Xuyin and Li, Tianyuan and Ayoko, Godwin A and Frost, Ray L and Feng, Yuexing},
  journal={Ecotoxicology and Environmental Safety},
  volume={165},
  pages={555--563},
  year={2018},
  publisher={Elsevier}
}

@article{roberts2007coupling,
  title={{Coupling and ergodicity of adaptive Markov chain Monte Carlo algorithms}},
  author={Roberts, Gareth O and Rosenthal, Jeffrey S},
  journal={Journal of applied probability},
  volume={44},
  number={2},
  pages={458--475},
  year={2007},
  publisher={Cambridge University Press}
}

@article{neal2003slice,
  title={Slice sampling},
  author={Neal, Radford M},
  journal={The annals of statistics},
  volume={31},
  number={3},
  pages={705--767},
  year={2003},
  publisher={Institute of Mathematical Statistics}
}

@article{Brook_1964, 
title={On the Distinction Between the Conditional Probability and the Joint Probability Approaches in the Specification of Nearest-Neighbour Systems}, 
volume={51}, 
ISSN={0006-3444}, 
url={http://dx.doi.org/10.2307/2334154}, 
DOI={10.2307/2334154}, 
number={3/4}, 
journal={Biometrika}, 
publisher={JSTOR}, 
author={Brook, D.}, 
year={1964}, 
month=dec, 
pages={481} }

@article{Hoff_2009, 
title={{Simulation of the Matrix Bingham–von Mises–Fisher Distribution, With Applications to Multivariate and Relational Data}}, 
volume={18}, 
ISSN={1537-2715}, 
url={http://dx.doi.org/10.1198/jcgs.2009.07177}, 
DOI={10.1198/jcgs.2009.07177}, 
number={2}, 
journal={Journal of Computational and Graphical Statistics}, 
publisher={Informa UK Limited}, 
author={Hoff, Peter D.}, 
year={2009}, 
month=jan, 
pages={438–456} }

@inproceedings{brubaker2012family,
  title={{A family of MCMC methods on implicitly defined manifolds}},
  author={Brubaker, Marcus and Salzmann, Mathieu and Urtasun, Raquel},
  booktitle={Artificial intelligence and statistics},
  pages={161--172},
  year={2012},
  organization={PMLR}
}

@article{Byrne_2013, 
title={{Geodesic Monte Carlo on Embedded Manifolds}}, 
volume={40}, 
ISSN={1467-9469}, 
url={http://dx.doi.org/10.1111/sjos.12036}, 
DOI={10.1111/sjos.12036}, 
number={4}, 
journal={Scandinavian Journal of Statistics}, 
publisher={Wiley}, 
author={Byrne, Simon and Girolami, Mark}, 
year={2013}, 
month=sep, 
pages={825–845} }

@article{holbrook2016bayesian,
  title={Bayesian inference on matrix manifolds for linear dimensionality reduction},
  author={Holbrook, Andrew and Vandenberg-Rodes, Alexander and Shahbaba, Babak},
  journal={arXiv preprint arXiv:1606.04478},
  year={2016}
}

@article{Chakraborty_2019, 
title={{Statistics on the Stiefel manifold: Theory and applications}}, 
volume={47}, 
ISSN={0090-5364}, 
url={http://dx.doi.org/10.1214/18-aos1692}, 
DOI={10.1214/18-aos1692}, 
number={1}, 
journal={The Annals of Statistics}, 
publisher={Institute of Mathematical Statistics}, 
author={Chakraborty, Rudrasis and Vemuri, Baba C.}, 
year={2019}, 
month=feb }

\section{Proofs for Lemma~\ref{thm:Projected-MCMC-kernel}}\label{supp:kernal_proof}

\begin{proof}
In view of step (i) and (ii), we obtain that for any Borel sets $A_1 \in \mathcal{B}
  \bigl( \Omega^g \bigr)$, $A_2 \in \mathcal{B}
  \bigl( \calH^K \bigr)$, $A_3 \in \mathcal{B}
  \bigl( \mathbb{R}^{(p+K)\times q} \bigr)$, $A_4 \in\mathcal{B}
  \bigl( \mathbb{S}_{+}^{q} \bigr)$, when $A = A_1 \times A_2 \times A_3 \times A_4$, the conditional probability of transiting from $\theta_1$ to $A$ is
  \begin{align*}
K(\theta_1,\;A)
  \;=&\;
    \int_{A_3\times A_4}\Bigl[
      \int_{g^{-1}(A_1)} \Bigl\{\int_{A_2} p\bigl(\bpsi \,\mid\, \bF, \bgamma,\bSigma, \bY\bigr)  d\bpsi \Bigr\}
       p\bigl(\bF \,\mid\, \bgamma,\bSigma, \bpsi_1, \bY\bigr)\,
      d\bF
    \Bigr] \\
    &\;p\bigl(\bgamma,\bSigma \,\mid\, \tilde{\bF}_1,\bY\bigr)
  \;\,d\bgamma\,d\bSigma\;.
  \end{align*}
Further, by taking $A$ to be the whole space, $\int_{g^{-1}\{\Omega^g\}} \Bigl\{\int_{\calH^K} p\bigl(\bpsi \,\mid\, \bF, \bgamma,\bSigma, \bY\bigr)  d\bpsi \Bigr\} 
    p\bigl(\bF \,\mid\, \bgamma, \bSigma, \bpsi_1, \bY\bigr)
  \,d\bF  = \int_{\mathbb{R}^{n \times K}}p\bigl(\bF \,\mid\, \bgamma, \bSigma, \bpsi_1 \bY\bigr) d\bF = 1$. Hence, $K(\theta_1, \cdot)$ is a probability measure for any $\theta_1$. Meanwhile, for any $A \in \mathcal{B}(\Theta)$,
$K(\cdot, A)$ is continuous, and thus measurable. Through the definition 4.2.1 of transition kernel in \cite{robert1999monte}, \eqref{eq: transit_kernel} is a valid transition kernel.
\end{proof}

\section{Proofs for Theorem~\ref{thm:converge}}\label{appendix:proof_holder}

\begin{proof}
By Lemma~\ref{thm:Harris_positive}, the Markov chain \((\theta_l)\) is Harris positive. By Lemma~\ref{lem:aperiodic}, \((\theta_l)\) is aperiodic. Then, it follows immediately from \cite[Theorem~4.6.5]{robert1999monte} that Theorem~\ref{thm:converge} holds. 
\end{proof}

\begin{lemma}[Strong \(\boldsymbol{\Psi}\)-irreducibility]\label{lem:strong_psi_irred}
Let \((\theta_\ell)\) be the Markov chain on \(\Theta\) with transition kernel \(K\) given in \eqref{eq: transit_kernel}, and let \(\boldsymbol{\Psi}\) be the 
measure on \(\Theta\) defined in Section~\ref{subsec:BSF_PMCMC}. 
For every measurable set \(A\subseteq \Theta\) with \(\boldsymbol{\Psi}(A)>0\), we have
\[
  K(\theta, A) \;>\; 0
  \quad
  \text{for all } \theta \in \Theta.
\]
Hence, \((\theta_\ell)\) is strongly \(\boldsymbol{\Psi}\)-irreducible.
\end{lemma}


\begin{proof}
By construction, for each \(\theta = \{\tilde{\bF}, \bpsi, \bgamma,\bSigma\}\in \Theta\), the kernel \(K(\theta,\cdot)\) assigns mass to any Borel set \(A\) according to
\begin{align*}
  K(\theta,A)
  \;=&\;
  \int_{A}
    \biggl\{
      \int_{g^{-1}(\tilde{\bF}')}
        p(\bpsi' \mid \bF, \bgamma',\bSigma', \bY) p(\bF \mid \bgamma',\bSigma', \bpsi, \bY)\, d\bF
    \biggr\}\,\\
    &\,p(\bgamma',\bSigma' \mid \tilde{\bF},\bY)
  \;d\tilde{\bF}'\, d\bpsi'\,d\bgamma'\,d\bSigma'.
\end{align*}
Since the densities \(p(\bpsi' \mid \bF, \bgamma',\bSigma', \bY)\), \(p(\bF \mid \bgamma,\bSigma, \bpsi, \bY)\), and \(p(\bgamma,\bSigma \mid \bF,\bY)\) are strictly positive over their supports, for any set \(A\) with \(\boldsymbol{\Psi}(A)>0\), it follows that
\[
  K(\theta,A)
  \;=\;
  \int_{A}
    \text{(positive integrand)}
  \;>\; 0,
  \quad
  \forall\, \theta\in\Theta.
\]
By definition (see, e.g., \cite[Definition 4.3.1]{robert1999monte}), if \(K(\theta,A)>0\) for all \(\theta\) whenever \(\boldsymbol{\Psi}(A)>0\), the chain is strongly \(\boldsymbol{\Psi}\)-irreducible. 
\end{proof}

\begin{lemma}[Aperiodicity]\label{lem:aperiodic}
Under the same conditions as in Lemma~\ref{lem:strong_psi_irred}, the Markov chain \((\theta_\ell)\) is aperiodic.
\end{lemma}

\begin{proof}
Based on \citet[][Thm 5.4.4]{meyn2012markov}, an irreducible Markov chain on a general state space \(\Theta\) is said to be 
\emph{$d$-cycle with \(d>1\)} if there exists \(d\) disjoint Borel sets \(\{D_0, D_1, \ldots, D_{d-1}\}\) such that 
\begin{enumerate}
    \item[(i) ] $K(\theta, D_{(i+1)\bmod d}) \;=\; 1
   \quad
   \text{for all}~\theta \in D_i, 
   \quad i=0,\dots,d-1$.
\item[(ii) ] The set $N = [\cup_{i = 0}^{d - 1} D_i]^c$ is $\boldsymbol{\Psi}$-null. And $\boldsymbol{\Psi}(D_i) > 0$ for $i = 0, \ldots, d-1$
\end{enumerate}

Suppose, for contradiction, $d>1$. Pick any \(\theta \in D_i\). By definition,
\[
  K\bigl(\theta,\,D_{(i+1)\bmod d}\bigr) \;=\; 1,
  \quad
  \text{meaning}
  \quad
  K\bigl(\theta,\,\Theta \setminus D_{(i+1)\bmod d}\bigr) \;=\; 0.
\]
However, under our assumptions (strict positivity of densities and hence full support), 
for \emph{every} \(\theta\in \Theta\) and \emph{every} nonempty Borel set \(B\subseteq\Theta\), 
we have \(K(\theta,B)>0\).  In particular, if we choose 
\(B = \Theta \setminus D_{(i+1)\bmod d}\), then 
\(\theta\in D_i\) implies \(K(\theta, B) > 0\).  
This contradicts \(K(\theta,B)=0\).  By definition, this implies that \((\theta_\ell)\) has $d = 1$ and is thus \emph{aperiodic} (see, e.g. the definition in \citep[Chp 5.]{meyn2012markov}.
\end{proof}

According to Definition 4.4.8 in \cite{robert1999monte}, the chain $(\theta_l)$ is \textit{Harris recurrent} if there exists a measure $\boldsymbol{\Psi}$ such that 
    \begin{enumerate}
        \item[(i) ] $(\theta_l)$ is $\boldsymbol{\Psi}$-irreducible\;, and
        \item[(ii) ] for every set $A$ with $\boldsymbol{\Psi}(A) > 0$, $A$ is \textit{Harris recurrent}.
    \end{enumerate}
    A set $A$ is \textit{Harris recurrent} if $P_\theta(\eta_A = \infty) = 1$ for all $\theta \in A$, where $\eta_A$ is the number of passages of $(\theta_l)$ in $A$ and $P_\theta(\eta_A = \infty)$ is the probability of visiting $A$ an infinite number of times starting from the initial point $\theta_1 = \theta$. 
\begin{lemma}[Recurrence]\label{thm:recurr}
The chain $(\theta_l)$ is \textit{Harris recurrent}.
\end{lemma}
\begin{proof}
    Suppose $A$ is a set with $\boldsymbol{\Psi}(A) > 0$. Through the definition of the transition kernel in Lemma~\ref{thm:Projected-MCMC-kernel}, we observe that $K(\theta_1, A)$ for $\theta_1 = \{\tilde{\bF}_1, \bpsi_1, \bgamma_1, \bSigma_1\}$ only depends on $\tilde{\bF}_1$ and $\bpsi_1$. Since $\tilde{\bF}_1$ and $\bpsi$ belong to compact sets $\Omega^g$ and $\calH^K$, respectively, the infimum $\inf_{\theta} \int_{A} K(\theta, y) dy = \rho > 0$. The probability of the chain $(\theta_l)$ not reaching set $A$ in $h$ iterations is 
    \begin{align*}
        &\quad P(\theta_2 \notin A, \ldots, \theta_h \notin A) = \int_{A^C} \cdots \int_{A^C} K(\theta_1, \theta_2) \cdots K(\theta_{h-1}, \theta_{h})d\theta_h \cdots d\theta_2\\
        &= \underbrace{\int_{A^C} \cdots \int_{A^C}}_{h-2} K(\theta_1, \theta_2) \cdots K(\theta_{h-2}, \theta_{h-1}) \left\{ \int_{A^C} K(\theta_{h-1}, \theta_{h}) d\theta_{h} \right\}d\theta_{h-1} \cdots d\theta_{2} \\
        &\leq  \underbrace{\int_{A^C} \cdots \int_{A^C}}_{h-2} K(\theta_1, \theta_2) \cdots K(\theta_{h-2}, \theta_{h-1}) \underbrace{\left\{ 1 - \inf_{\theta_{h-1}} \int_A K(\theta_{h-1}, \theta_{h}) d\theta_{h} \right\}}_{1-\rho}d\theta_{h-1} \cdots d\theta_{2} \\
        &\leq (1-\rho)^{h-1}\;.
    \end{align*}
    Define the stopping time at $A$ as $\tau_A = \inf\{l > 1; \theta_l \in A\}$, we have that 
    $$p_{\theta}(\tau_A < \infty) = 1 - \lim_{h \rightarrow \infty} \{p_{\theta}(\tau_A > h)\} = 1$$ for any initial point $\theta$. 
    By Proposition 4.4.9 in \cite{robert1999monte}, we can show that $P_\theta(\eta_A = \infty) = 1$ when $P_\theta(\tau_A < \infty)$ for every $\theta \in A$. Therefore, $A$ is Harris recurrent. With the irreducibility given in Lemma~\ref{lem:strong_psi_irred}, we complete the proof.
\end{proof}

\begin{lemma}\label{thm:Harris_positive}[Invariant Finite Measure]
There exists an invariant finite measure $\pi(\cdot)$ for the chain $(\theta_l)$, and, hence, $(\theta_l)$ is Harris positive.
\end{lemma}

\begin{proof}
By Lemma~\ref{thm:recurr}, $(\theta_l) = (\tilde{\bF}_l, \bpsi_l, \bgamma_l, \bSigma_l)$ is a recurrent chain, and, therefore, $(\tilde{\bF}_l, \bpsi_l)$ is also a recurrent chain. By Thm 4.5.4 in \citet{robert1999monte}, there exists an invariant $\sigma$-finite measure $\pi_{\tilde{\bF}, \bpsi}(\cdot)$ (unique up to a multiplicative factor) for chain $(\tilde{\bF}_l, \bpsi_l)$ on space $\Omega^g \times \calH^K$. Further, since $\Omega^g \times \calH^K$ is compact, any $\sigma$-finite measure on $\Omega^g \times \calH^K$ must be finite. In particular, \(\pi_{\tilde{\bF}, \bpsi}(\Omega^g \times \calH^K)<\infty\). Now, for any $B \in \mathcal{B}(\Theta)$, given the fact that $K(\theta, B) = K\bigl( \{\tilde{\bF}, \bpsi\} , B\bigr)$, we can define a measure $\pi(\cdot)$ on $\Theta$ through
\begin{equation}\label{eq:pi_construct}
\pi(B) = \int_{\Omega^g \times \calH^K} K\bigl(\{\tilde{\bF}, \bpsi\}, B\bigr) \pi_{\tilde{\bF}, \bpsi}(d\tilde{\bF}d\bpsi)
\end{equation}
As shown in Lemma~\ref{thm:Projected-MCMC-kernel}, $K\bigl(\{\tilde{\bF}, \bpsi\}, \Theta\bigr) = 1$, and, hence, 
\[\pi(\Theta) = \int_{\Omega^g} K\bigl(\{\tilde{\bF}, \bpsi\}, \Theta\bigr) \pi_{\tilde{\bF}, \bpsi}(d\tilde{\bF}d\bpsi) = \pi_{\tilde{\bF}, \bpsi}(\Omega^g \times \calH^K) < \infty\;.\]
Moreover, since $\pi_{\tilde{\bF}, \bpsi}(\cdot)$ is the unique invariant measure for $(\tilde{\bF}_l, \bpsi_l)$, 
\begin{align*}
\pi_{\tilde{\bF}, \bpsi}(A) &=\int_{\Omega^g \times \bpsi} K\Bigl(\{\tilde{\bF}, \bpsi\},\{ \theta': (\tilde{\bF}', \bpsi')\in A\}\Bigr)\, \pi_{\tilde{\bF}, \bpsi}(d\tilde{\bF}d\bpsi) \\
&= \pi\bigl(\{(\tilde{\bF}, \bpsi, \bgamma, \bSigma)\in\Theta: (\tilde{\bF}, \bpsi)\in A\}\bigr)
\end{align*}
for any \(A \in \mathcal{B}(\Omega^g \times \calH^K)\). Therefore, \eqref{eq:pi_construct} can be written as 
$$
\pi(B) = \int_{\Theta} K(\theta, B) \pi(d\theta)
$$
Hence, $\pi(\cdot)$ is the invariant finite measure for the chain $(\theta_l)$, and, therefore, $(\theta_l)$ is Harris positive.
\end{proof}

\section{Prediction through PBSF models}\label{supp: predict}
PBSF models allow prediction of missing outcomes and outcomes at unobserved locations under an additional assumption. Let \(\mathcal{U}\) denote a set of unobserved locations. Extending the latent factor matrix \(\bF\) to include rows corresponding to \(\mathcal{U}\), the ProjMC\(^2\) algorithm can be directly adapted to generate posterior samples of latent factors at unobserved sites. Denote by \(\bY_u\) the missing or unobserved outcomes to be predicted, and assume that \(\bY_u\) is conditionally independent of the observed outcomes \(\bY\) given \((\bF, \bgamma, \bSigma)\). Under this assumption, the posterior predictive distribution of \(\bY_u\) is obtained by integrating model~\eqref{eq: spatial_factor} with respect to the posterior distribution of \((\bF, \bgamma, \bSigma)\). Accordingly, posterior predictive samples of \(\bY_u\) can be generated from~\eqref{eq: spatial_factor} for each draw \((\bgamma^{(l)}, \bSigma^{(l)}, \bF^{(l)})\). Nevertheless, as emphasized in the motivation for the PBSF model, prediction is more naturally carried out under the original BSF formulation, which affords greater flexibility and can yield improved predictive performance. For this reason, detailed implementation and illustration of prediction are omitted here for brevity.

\newpage

\section{Algorithm of NNGP based PBSF models}\label{SM: BSLMC_NNGP_alg}
\begin{algorithm}[H]
\caption{ProjMC$^2$ for NNGP-based PBSF models with diagonal $\bSigma$ (prefixed or updated $\bpsi$) }
\label{alg:projmc2_driver}
\begin{algorithmic}[1]

\State \textbf{Input:} Design matrix $\bX$, outcomes $\bY$, set of spots with at least one observation $\calS$, prior parameters $\bmu_{\bbeta},\bV_{\bbeta}, \bmu_{\bLambda},\bV_{\bLambda}$, IG hyperparameters $a,\{b_i\}_{i=1}^q$, neighbor size $m$, MCMC length $L$, warm-up $L_{\text{warm}}$, and either prefixed $\{\psi_k\}_{k=1}^K$ or bounds $\{[a_{\psi_k},b_{\psi_k}]\}_{k=1}^K$ for slice updates.

\State \textbf{initialisation}
\Statex \quad $(\bbeta^{(0)},\bLambda^{(0)},\bSigma^{(0)}) \gets \textsc{InitializeParams}(\bX,\bY, K)$ (and $\bpsi^{(0)}$ if update $\bpsi$)

\State \textbf{Precomputation \& Preallocation}
\Statex \quad  Calculate Cholesky decompositions $\bV_{\bLambda} = \bL_{\bLambda}\bL_{\bLambda}^\top$ and $\bV_{\bbeta} = \bL_{\bbeta}\bL_{\bbeta}^\top$ \Comment{{\scriptsize $\mathcal{O}(p^3 + K^3)$}}

\Statex \quad $(\{\bA_{\rho_k}\}_{k = 1}^K,\{\bD_{\rho_k}\}_{k = 1}^K) \gets \textsc{PrecomputeNNGP}(\calS,\bpsi,m)$  (if update $\bpsi$ use $\bpsi^{(0)}$)\Comment{{\scriptsize $\mathcal{O}(n\log n)+\mathcal{O}(Knm^3)$}}

\Statex \quad \textsc{PreallocateBuffers}$(L,n,p,q,K)$ \Comment{{\scriptsize storage $\mathcal{O}(L\cdot\{nK+(p+K)q\})$}}

\For{$l=1$ \textbf{to} $L$}
    \State $(\tilde{\bX},\tilde{\bY}) \gets \textsc{BuildTildeXY}(\bSigma^{(l-1)},\bbeta^{(l-1)},\bLambda^{(l-1)},\{\bA_{\rho_k}\},\{\bD_{\rho_k}\})$ \Comment{{\scriptsize $\mathcal{O}((p{+}K)nq)$}}
    \State $\bF^{(l)} \gets \textsc{UpdateF\_LSMR}(\tilde{\bX},\tilde{\bY})$ \Comment{{\scriptsize LSMR solve cost}}
    \If{\textit{update $\bpsi$}}
        \State $\bpsi^{(l)} \gets \textsc{UpdatePsi\_Slice}(\bpsi^{(l-1)},\{[a_{\psi_k},b_{\psi_k}]\},L_{\text{warm}},...)$
        \State Refresh $(\bA_{\rho_k},\bD_{\rho_k})$ for any updated $\psi_k$
    \EndIf
    \State $\tilde{\bF}^{(l)} \gets \textsc{ProjectEmbeddings}(\bF^{(l)})$ \Comment{{\scriptsize $\mathcal{O}(nK^2)$}}
    \State $(\bbeta^{(l)},\bLambda^{(l)},\bSigma^{(l)}) \gets \textsc{UpdateMNIW}(\tilde{\bF}^{(l)},\bX,\bY,\bmu_{\bbeta},\bV_{\bbeta},\bmu_{\bLambda},\bV_{\bLambda},a,\{b_i\})$
\EndFor

\State \textbf{Output:} Posterior draws $\{\tilde{\bF}^{(l)},\bpsi^{(l)},\bbeta^{(l)},\bLambda^{(l)}\}_{l=1}^L$ (retain post-warm-up; optional thinning).

\end{algorithmic}
\end{algorithm}

\begin{submodule}[H]
\small
\setstretch{1.0}
\caption{\textsc{UpdatePsi\_Slice} (adaptive logit–slice for bounded $\psi_k$)}
\label{alg:update_psi_slice}
\begin{algorithmic}[1]
\Require Current $\{\psi_k\}_{k=1}^K$, bounds $\{[a_{\psi_k},b_{\psi_k}]\}_{k=1}^K$, warm-up length $L_{\text{warm}}$, adaptation targets $(\tau_e,\tau_s)$, weight $\lambda$, and learning-rate schedule $\{\eta_t\}$ {\scriptsize \textbf{Adaptation statistics (per iteration $t$ and parameter $\psi_k$):} 
$e_t$ = total \emph{step-out expansions} (left$+$right) taken while bracketing the slice; 
$\tau_e$ = target mean expansions (e.g., $1$–$2$).
$s_t$ = \emph{shrinkage rejections} (failed uniform draws inside $[L,R]$ before acceptance); 
$\tau_s$ = target mean shrink steps (e.g., $1$–$3$). 
$\lambda$ (e.g., $0.25$) tempers the shrink term; $\eta_t$ diminishes during warm-up.}
\For{$k=1$ \textbf{to} $K$}
    \State \textbf{Transform} to $\mathbb{R}$: $\xi \gets \log\!\big((\psi_k-a_{\psi_k})/(b_{\psi_k}-\psi_k)\big)$
    \State \textbf{Define log-target} in $\xi$:
    \Statex \quad $\ell_k(\xi) \gets \log p(\bF_k \mid \psi_k(\xi)) + \log p(\psi_k(\xi)) + \log\!\left|\frac{d\psi_k}{dy}\right|$,
    \Statex \quad $\psi_k(\xi)=a_{\psi_k}+\dfrac{b_{\psi_k}-a_{\psi_k}}{1+e^{-\xi}}$, \; $\log\!\left|\frac{d\psi_k}{dy}\right|=\log(b_k{-}a_k)-\log(1{+}e^\xi)-\log(1{+}e^{-\xi})$.
    \Statex \quad {\small (\textit{Note: $\bF_k \mid \psi_k(\xi)$ is modelled using NNGP. For efficiency and storage management, it is preferable to preallocate }}
    \Statex \quad {\small \textit{$(\bA'_{\rho_k}, \bD'_{\rho_k})$ for computing $\log p(\bF_k \mid \psi_k(\xi))$ and update them as needed.)}}
    \State \textbf{Slice level}: $h \gets \ell_k(\xi) + \log U$, \; $U \sim \mathrm{Unif}(0,1)$
    \State \textbf{Stepping-out}: initialize $[L,R]$ of width $w_k$ around $\xi$; set $e_t\gets 0$
    \While{$\ell_k(L) \ge h$} \State $L \gets L - w_k$; \;$e_t \gets e_t + 1$ \EndWhile
    \While{$\ell_k(R) \ge h$} \State $R \gets R + w_k$; \;$e_t \gets e_t + 1$ \EndWhile
    \State \textbf{Shrinkage}: set $s_t\gets 0$
    \While{true}
        \State $\xi' \sim \mathrm{Unif}(L,R)$
        \If{$\ell_k(\xi') \ge h$} 
            \State \textbf{break} 
        \ElsIf{$\xi' < \xi$}
            \State $L \gets \xi'$ 
        \Else
            \State $R \gets \xi'$
        \EndIf
        \State $s_t \gets s_t + 1$
    \EndWhile
    \State \textbf{Map back}: $\psi_k \gets a_k + \dfrac{b_k-a_k}{1+e^{-\xi'}}$
    \State \textbf{Warm-up adaptation of width $w_k$}:
    \If{$t \le L_{\text{warm}}$}
        \State $\log w_k \leftarrow \log w_k + \eta_t\Big[(e_t-\tau_e)+\lambda\, (s_t-\tau_s)\Big]$ \Comment{{\scriptsize increases $w_k$ if expansions are frequent; decreases if shrinkage is excessive}}
    \Else
        \State \textbf{Freeze} $w_k$ (no further adaptation to preserve stationarity/ergodicity)
    \EndIf
\EndFor
\State \textbf{Return} updated $\{\psi_k\}_{k = 1}^K$ (and $\{w_k\}_{k = 1}^K$)
\end{algorithmic}
\end{submodule}

\begin{submodule}[H]
\setstretch{1.0}
\caption{\textsc{PrecomputeNNGP}$(\calS,\bpsi,m)$}
\label{alg:precompute_nngp}
\begin{algorithmic}[1]
\State Construct the maximin ordering of $\calS$ \Comment{{\scriptsize $\mathcal{O}(n\log n)$}}

\State Build the nearest neighbor for $\calS$ with the new ordering using K-D tree \citep{bentley1975multidimensional} \Comment{{\scriptsize $\mathcal{O}(n\log n)$}}


\For{$k=1$ \textbf{to} $K$ }
    \State Construct $\{\bA_{\rho_k}\}_{k = 1}^K$ and $\{\bD_{\rho_k}\}_{k = 1}^K$ as described, for example, in \citet{finley2019efficient}
    \Comment{{\scriptsize $\mathcal{O}(nm^3)$ per $k$}}
\EndFor
\State \textbf{Return} $\{\bA_{\rho_k}\}_{k = 1}^K,\{\bD_{\rho_k}\}_{k = 1}^K$
\end{algorithmic}
\end{submodule}

\begin{submodule}[h]
\setstretch{1.0}
\caption{\textsc{PreallocateBuffers}$(L,n,p,q,K)$}
\label{alg:preallocate}
\begin{algorithmic}[1]
\State Allocate storage for MCMC chains $\tilde{\bF}^{(l)}$, $\bbeta^{(l)}$, $\bLambda^{(l)}$ (and $\bpsi^{(l)}$ if update $\bpsi$) \Comment{{\scriptsize storage: $\mathcal{O}(L\cdot\{nK+(p+K)q\})$}}
\State Allocate work arrays for $\bbeta$, $\bLambda$ and $\bSigma$ update: $\bX^\ast$, $\bY^\ast$, $\bmu^\ast$, $\bL^\ast$ the Cholesky decomposition of $\bV^\ast$, vector $\bu$ with length $(p+K) q$, matrix for storing residual $Y_{Xm} = \bY^\ast - \bX^\ast\mu^\ast$, and vector for storing updated $b^\ast$ for all outcomes. 
\Comment{{\scriptsize storage: $\mathcal{O}(nq)$}}
\State Allocate for $\tilde{\bF}$ update: $\tilde{\bX}$, $\tilde{\bY}$, $\bF$, QR decomposition of $\bF$, diagonal matrix $D_{\bSigma}$, $f_m$ that store column mean of $\bF$, and vector $\bv$ with length $n(K+q)$. {(If misalignment is present, allocate according to the observed locations for each outcome)} \Comment{{\scriptsize storage: $\mathcal{O}(nq)$}}
\end{algorithmic}
\end{submodule}

\begin{submodule}[H]
\setstretch{1.0}
\caption{\textsc{InitializeParams}$(\bX,\bY,K)$}
\label{alg:init_params}
\begin{algorithmic}[1]
\State $\bbeta^{(0)} \gets (\bX^\top\bX)^{-1}\bX^\top\bY$ {(If misalignment is present, initialize each column using the observed locations for the corresponding outcome)} \Comment{{\scriptsize $\mathcal{O}(np^2{+}npq)$}}
\State Conduct randomized SVD of $R = \bY - \bX\bbeta$ with $K$ principle components, store the loading matrix as $\bLambda^{(0)}$ and update remaining residual $R$. {(If misalignment is present, use only locations with complete observations)}  \Comment{{\scriptsize $\mathcal{O}(nqK)$}}
\State For $i = 1, \ldots, q$, initialize $\sigma_i^{2(0)}$ with variance of the the remaining residual. \Comment{{\scriptsize $\mathcal{O}(nq)$}}
\State \textbf{Return} $(\bbeta^{(0)},\bLambda^{(0)},\bSigma^{(0)})$
\end{algorithmic}
\end{submodule}

\begin{submodule}[H]
\setstretch{1.0}
\caption{\textsc{BuildTildeXY}$(\bSigma,\bbeta,\bLambda,\{\bA_{\rho_k}\},\{\bD_{\rho_k}\})$}
\label{alg:build_tilde}
\begin{algorithmic}[1]
\If{$l=1$}
    \State Initialize $\tilde{\bY}$ and $\Tilde{\bX}$ in \eqref{eq: SLMC_F_cond_post} with $\bL_k^{-1} = \bD_{\rho_k}^{-\frac{1}{2}}(\bI - \bA_{\rho_k})$ \Comment{{\scriptsize $\mathcal{O}(nq(p{+}K))+\mathcal{O}(nmK) \approx \mathcal{O}(nq)$}}
\EndIf
\State Update $D_{\bSigma}$ with diagonal elements from $\bSigma^{(l-1)}$  
\State Update $\Tilde{\bX}$ and $\Tilde{\bY}$ in \eqref{eq: SLMC_F_cond_post} with $\bSigma = D_{\bSigma}$, $\bbeta = \bbeta^{(l-1)}$ and $\bLambda = \bLambda^{(l-1)}$. \Comment{{\scriptsize $\mathcal{O}((p{+}K)nq)$}}
\State \textbf{Return} $(\tilde{\bX},\tilde{\bY})$
\end{algorithmic}
\end{submodule}

\begin{submodule}[H]
\setstretch{1.0}
\caption{\textsc{UpdateF\_LSMR}$(\tilde{\bX},\tilde{\bY})$}
\label{alg:updateF_lsmr}
\begin{algorithmic}[1]
\State Sample $\bv \sim \mathrm{N}(\mathbf{0},\bI_{Kn})$ \Comment{{\scriptsize $\mathcal{O}(nK)$}}
\State Solve $\tilde{\bX}\,\mathrm{vec}(\bF) = \tilde{\bY}+\bv$ by LSMR; reshape to $\bF \in \mathbb{R}^{n\times K}$
\State \textbf{Return} $\bF$
\end{algorithmic}
\end{submodule}

\begin{submodule}[H]
\setstretch{1.0}
\caption{\textsc{ProjectEmbeddings}$(\bF)$}
\label{alg:project_embeddings}
\begin{algorithmic}[1]
\State Record column mean of $\bF$ in $f_m$ and update $\bF = \bF - \mathbf{1}_n  f_m^\top$ \Comment{{\scriptsize $\mathcal{O}(nq)$}}
\State Compute the thin Q matrix of $\bF$ using Modified Gram-Schmidt (MGS) and store it as $\tilde{\bF}$  \Comment{{\scriptsize $\mathcal{O}(nK^2)$}}
\State Update $\tilde{\bF} = \sqrt{n} \cdot Q$ \Comment{{\scriptsize $\mathcal{O}(nK)$}}
\State \textbf{Return} $\tilde{\bF}$
\end{algorithmic}
\end{submodule}

\begin{submodule}[H]
\setstretch{1.0}
\caption{\textsc{UpdateMNIW}$(\tilde{\bF},\bX,\bY,\bmu_{\bbeta},\bV_{\bbeta},\bmu_{\bLambda},\bV_{\bLambda},a,\{b_i\})$}
\label{alg:update_mniw}
\begin{algorithmic}[1]
\State Construct $\bX^{\ast}$ and $\bY^{\ast}$ in \eqref{eq: augment_linear_LMC} with $\bF$ replaced by $\tilde{\bF}^{(l)}$ 
\State Calculate the Cholesky decomposition $\bL^\ast$ of $\bV^{\ast-1} = \bX^{\ast\top} \bX^{\ast} = \bL^\ast \bL^{\ast\top}$.  \Comment{{\scriptsize $\mathcal{O}((n+p+K)(p+K)^2)$ }}

\If{no misalignment}
\State Calculate $\bmu^\ast = \bV^{\ast} [\bX^{\ast\top} \bY^{\ast}]$ and update $Y_{Xm} = \bY^\ast - \bX^\ast \bmu^\ast$  \Comment{{\scriptsize$\mathcal{O}((n+p+K)(p+K)q)$}}
\State Sample $(\bgamma^{(l)}, \bSigma^{(l)})$ from \eqref{eq: SLMC_MNIW} {\small \textit{This step applies to both unrestricted and diagonal forms of $\bSigma$. For clarity, we present the technical details under the diagonal $\bSigma$ setting.}}
\Statex \quad (i) Sample elements of $\bSigma^{(l)}$ from $\mathrm{IG}\bigl(a^\ast, b_i^\ast\bigr)$ with $a^\ast = a + \frac{n}{2}$, $b_i^\ast = b_i + \frac{1}{2} Y_{Xmi}'Y_{Xmi}$. Here $Y_{Xmi}$ denotes the $i$-th column of $Y_{Xm}$. 
\Comment{{\scriptsize$\mathcal{O}((n+p+K)q)$}}
\Statex \quad (ii) Sample $\bgamma^{(l)} = [\bbeta^{(l)\top}, \bLambda^{(l)^\top}]^\top$ from $\mbox{MN}(\bmu^\ast, \bV^\ast, \bSigma^{(l)})$ 
\Statex \quad \quad (ii.a) Sample $\bu \sim \mbox{MN}(\mathbf{0}, \bI_{p+K}, \{\bSigma^{(l)}\}^{1/2} )$ \Comment{{\scriptsize$\mathcal{O}((p+K)q)$}}

\Statex \quad \quad (ii.b) Generate $\bgamma^{(l)} = \bmu^\ast +\bL^{\ast-\top}\bu$ \Comment{{\scriptsize$\mathcal{O}((p+K)^2q)$}}
\Else
\For{i in 1:q}
\State Calculate the Cholesky decomposition $\bL_i^\ast$ of $\bV_i^{\ast-1} = \bX_i^{\ast\top} \bX_i^{\ast} = \bL_i^\ast \bL_i^{\ast\top}$ for $\bX_i^{\ast}$ using data
\Statex \quad \quad \quad \quad at the $n_i$ locations $\calS_i$ where the $i$-th outcome is observed.  

\State Compute $\bmu^\ast =\bV_i^{\ast} \bX_i^{\ast\top}\bY^{\ast}_i$, and set $Y_{Xm} = \bY^{\ast}_i - \bX^{\ast}_i\bmu^\ast$ \Comment{{\scriptsize$\mathcal{O}((n_i+p+K)(p+K))$}}
\State Sample $(\bgamma^{(l)}_i, \sigma^{2(l)}_i)$ from \eqref{eq: SLMC_NIG_misaligned}, where $\bgamma^{(l)}_i$ is the $i$-th column of $\bgamma^{(l)}$
\Statex \quad \quad  \quad (i) Sample $\sigma_i^{2(l)}$ from $\mathrm{IG}\bigl(a^\ast, b_i^\ast\bigr)$ with $a^\ast = a + \frac{n_i}{2}$, $b_i^\ast = b_i + \frac{1}{2} Y_{Xmi}^\top Y_{Xmi}$. \Comment{{\scriptsize$\mathcal{O}(n_i+p+K)$}}
\Statex \quad \quad \quad (ii) Sample $\bgamma^{(l)}_i$ from $\mathrm{N}(\bmu^\ast, \sigma_i^{2(l)} \bV_i^{\ast})$  \Comment{{\scriptsize$\mathcal{O}((p+K)^2)$}}

\EndFor
\EndIf
\State \textbf{Return} $(\bbeta,\bLambda,\bSigma)$
\end{algorithmic}
\end{submodule}

\section{Values of parameters in simulation examples}\label{sm: values_examples}

\subsection{Values of parameters to generate simulations in simulation example 1}
\begin{align*}
\bbeta &= \begin{bmatrix} 1.0 & -1.0 & 1.0 & -0.5 & 2.0 & -1.5 & 0.5 & 0.3 & -2.0 & 1.5\\
     -3.0 & 2.0 & 2.0 & -1.0 & -4.0 & 3.0 & 4.0 & -2.5 & 5.0 & -3.0 \end{bmatrix} \\
\bLambda &= \begin{bmatrix}  0.81 & 0.49 & -0.49 & -0.15 & -0.8 & 0.38 & -0.94 & 0.86 & 0.16 & -0.76 \\ -0.11 & 0.02 & -0.33 & 0.74 & -0.75 & -0.73 & -0.3 & 0.92 & -0.38 & -0.59 \end{bmatrix}\\
\mbox{Diagonal}(\bSigma) & = \begin{bmatrix} 0.5 & 1 & 0.4 & 2 & 0.3 & 2.5 & 3.5 & 0.45 & 1.5 & 0.5 \end{bmatrix}
\end{align*}
$$
\phi_1 = 6.0\;, \; \phi_2 = 9.0
$$

\section{Simulation II results}\label{sm: sim2_results}
The second study focused on evaluating the performance of the proposed algorithms under prefixed hyperparameters. For convergence and mixing assessment, three algorithms were considered: (i) the original blocked Gibbs sampler, (ii) the blocked Gibbs sampler with post-processing, and (iii) the proposed ProjMC$^2$ method. Table~\ref{tab:ess_comparison_fix} summarizes the effective sample sizes (ESS) for all model parameters across these approaches, and Figure~\ref{fig:trace_compare_fix} presents trace plots for weakly identifiable parameters, including the intercepts ($\bbeta_0$) and the loadings (elements of $\bLambda$). As shown in Table~\ref{tab:ess_comparison_fix}, sampling efficiency for weakly identifiable parameters, namely the loading matrix $\bLambda$ and the factor matrix $\bF$, improved under prefixed hyperparameters. In contrast, the identifiable parameters---the regression coefficients ($\bbeta_1$) and the noise covariance matrix $\bSigma$---were consistently well estimated by all methods, each yielding minimum ESS values exceeding 5,000.

\begin{table}[ht]
\centering
\begin{tabular}{lccccccc}
\toprule
 & 
\multicolumn{2}{c}{\textbf{Gibbs}} & 
\multicolumn{2}{c}{\textbf{Gibbs + Post}} & 
\multicolumn{2}{c}{\textbf{ProjMC\textsuperscript{2}}} \\
\cmidrule(lr){2-3} \cmidrule(lr){4-5} \cmidrule(lr){6-7} \textbf{ESS}
&(min/mean/med)&  <100 
&(min/mean/med)& <100 
&(min/mean/med)& <100 \\
\midrule
$\bbeta_0$ & 37/41/42 & 100\% & 7332/11827/12398 & 0\% & 8675/12383/12884	& 0\% \\
$\bbeta_1$ & 6650/11261/12001 & 0\% & 6650/11261/12001 & 0\% & 8075/11990/12600 & 0\% \\
$\bLambda$ & 36/46/41	 & 95\%  & 36/43/38  & 100\% & 191/2591/356 & 0\% \\
$\bF$ & 41/67/68 & 100\% & 44/2199/242	 & 23\% & 516/9975/11858 & 0\% \\
$\bSigma$ & 5559/10540/12226	 & 0\%  & 5559/10540/12226 & 0\% & 7600/11563/12756 & 0\% \\
\bottomrule
\end{tabular}
\caption{Comparison of effective sample size (ESS)—reported as minimum, mean, and median, and the proportion of variables with low ESS values (ESS < 100)— across three methods: the original blocked Gibbs sampler (Gibbs), Gibbs with post-processing (Gibbs + Post), and the proposed Projected MCMC (ProjMC$^2$). Results are shown for the intercepts $\bbeta_0$, regression coefficients $\bbeta_1$, loading matrix $\bLambda$, matrix of latent factors $\bF$, and the noise covariance matrix $\bSigma$, based on MCMC chains of 20{,}000 iterations with the first 5{,}000 iterations discarded as warm-up. }
\label{tab:ess_comparison_fix}
\end{table}

For the intercepts ($\bbeta_0$), substantial improvements in convergence and mixing are achieved through recentering latent spatial factors, whether via post-processing or ProjMC$^2$. Post-processing modestly improves the mixing rate of certain spatial factor components ($\bF$), elevating the minimum and median ESS from 41 and 68 to 44 and 242, respectively, across $2 \times 2,000$ parameters. Analogous to study one, for the loading matrix ($\bLambda$) and latent spatial factors ($\bF$), stable and efficient MCMC chains were only obtained through the ProjMC$^2$ method, and trace plots in Figure~\ref{fig:trace_compare_fix} demonstrate rapid convergence of all MCMC chains produced by ProjMC$^2$. 

\begin{figure}[!ht]
  \centering
  \subfloat[Intercepts from Gibbs sampler \label{subfig:trace_gibbs_beta}]{%
    \includegraphics[width=0.33\textwidth]{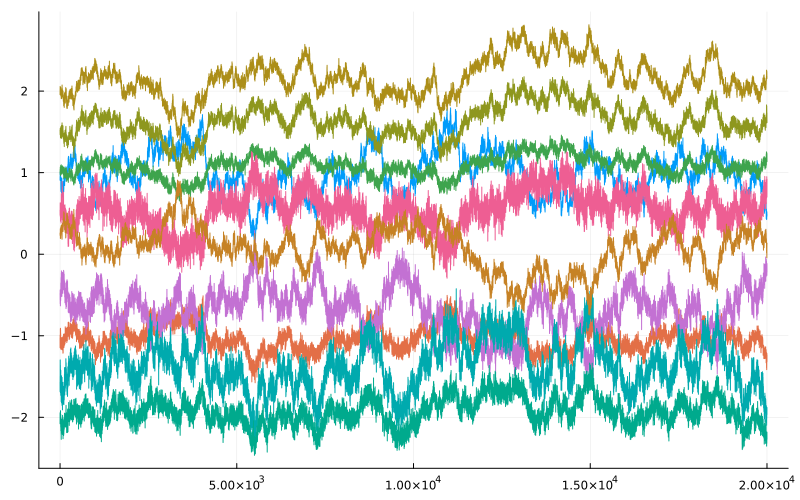}
  }
  \subfloat[Intercepts from Gibbs + Post
  \label{subfig:trace_post_beta}]{%
    \includegraphics[width=0.33\textwidth]{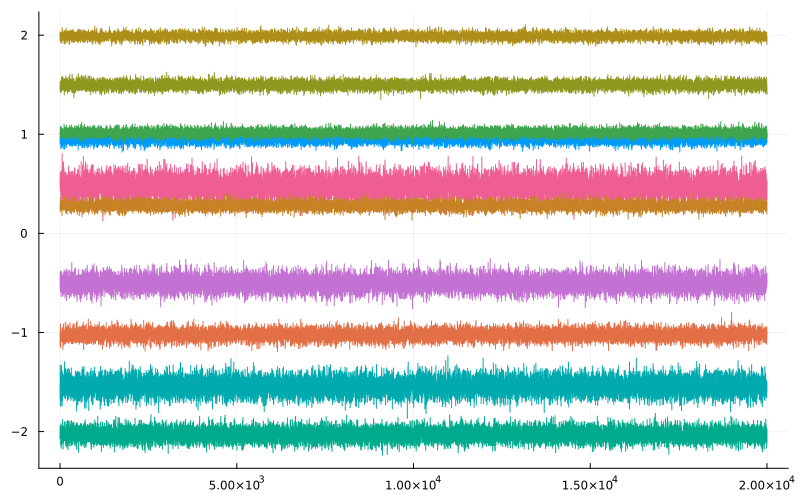}
  }
  \subfloat[Intercepts from ProjMC\textsuperscript{2} \label{subfig:trace_projmc2_beta}]{%
    \includegraphics[width=0.33\textwidth]{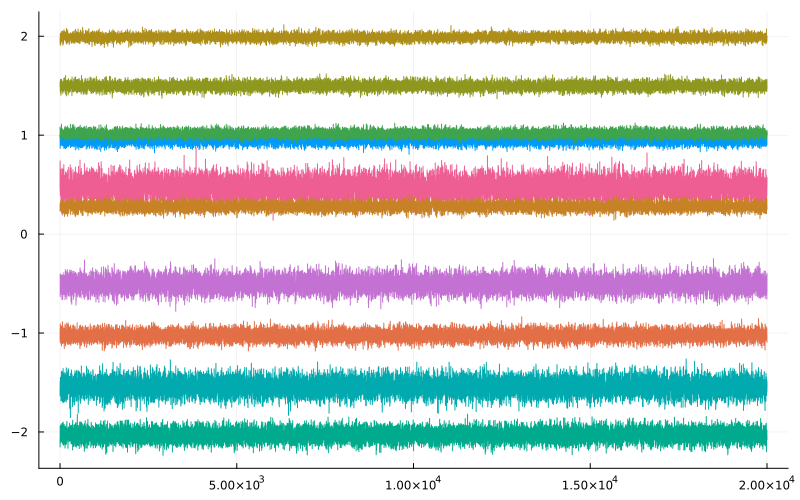}
  }\\[1ex]
  \subfloat[Loadings from Gibbs sampler \label{subfig:trace_gibbs_lambda}]{%
    \includegraphics[width=0.33\textwidth]{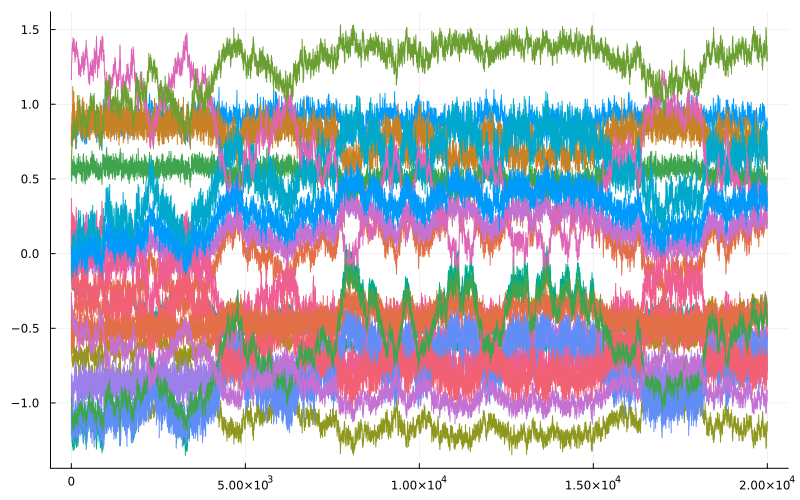}
  }
  \subfloat[Loadings from Gibbs + Post
  \label{subfig:trace_post_lambda}]{%
    \includegraphics[width=0.33\textwidth]{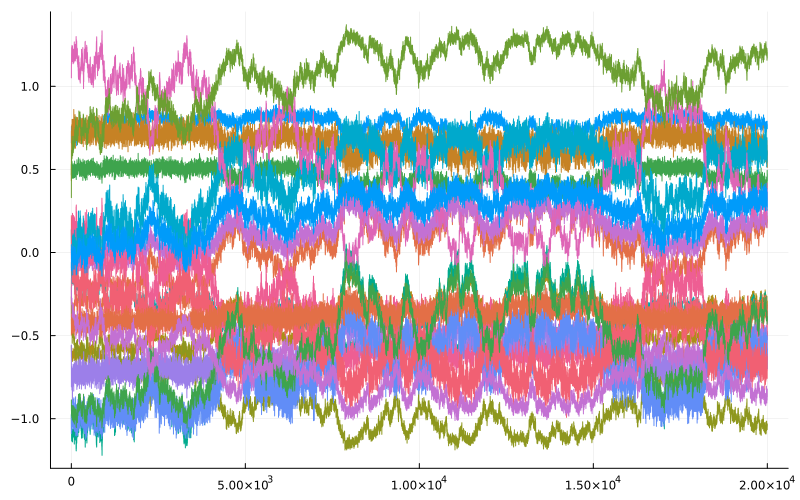}
  }
  \subfloat[Loadings from ProjMC\textsuperscript{2} \label{subfig:trace_projmc2_lambda}]{%
    \includegraphics[width=0.33\textwidth]{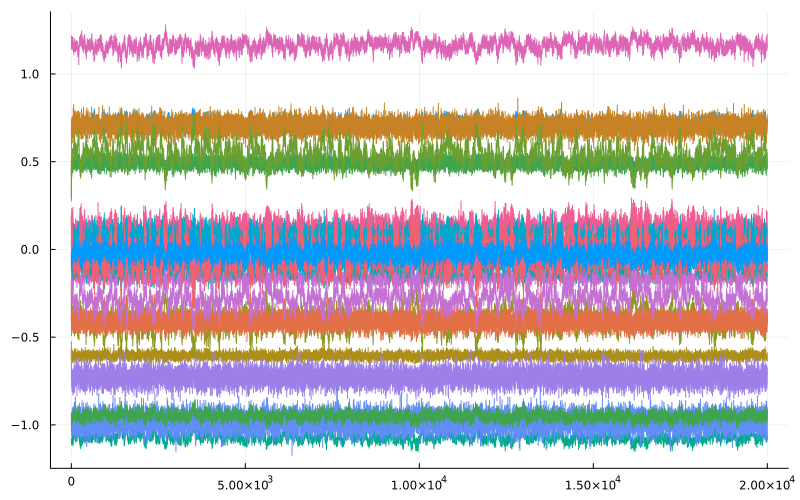}
  }
  \caption{Trace plots of MCMC chains for weakly identifiable parameters: the intercepts (top row) and loading matrix $\bLambda$ (bottom row). Columns correspond to results from the blocked Gibbs sampler (left), Gibbs sampler with post-processing (middle), and ProjMC\textsuperscript{2} (right).}
  \label{fig:trace_compare_fix}
\end{figure}

\paragraph{Inference Accuracy}
Guided by the previous findings, this evaluation focused on the post-processed blocked Gibbs sampler (hereafter, Gibbs+Post) and the proposed ProjMC$^2$ algorithm. Both the true factors and all posterior samples were projected onto the sphere of radius $\sqrt{n-1}$. For each factor, the point estimate was obtained using the Fr\'echet mean (or mean direction) of its posterior samples \cite{mardia2009directional}. Figure~\ref{fig:f_compare_fixphi} presents a visual comparison between the true latent spatial factors and the corresponding point estimates from each method. It is important to note that the decay parameters were deliberately chosen to be smaller than their true values, which predisposed the recovered factors toward smoother spatial patterns. A sensitivity analysis examining the impact of less restrictive priors is provided in Section~\ref{subsec: sens}. Visual inspection of Figure~\ref{fig:f_compare_fixphi} indicates that both Gibbs+Post and ProjMC$^2$ successfully captured the dominant spatial patterns in $\fb_1$ and $\fb_2$. Consistent with the first study, factors estimated via ProjMC$^2$ tended to accentuate patterns specific to each factor, whereas Gibbs+Post occasionally produced estimates in which subtle patterns appeared similar across factors.

\begin{figure}[!ht]
  \centering
  \subfloat[\centering True $f_1$ \label{subfig:MCMC_true_f1_phifix}]{%
    \includegraphics[width=0.25\textwidth]{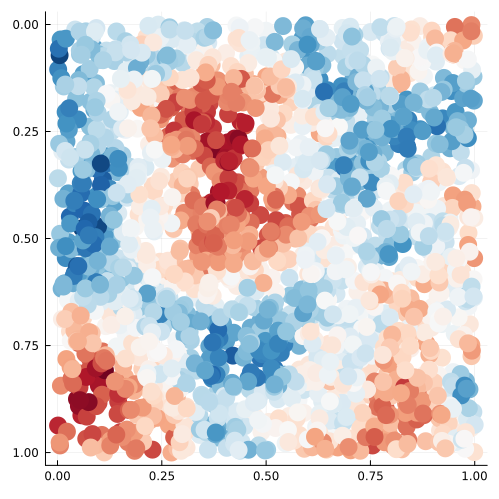}
  }
  \subfloat[\centering Estimated $f_1$ (Gibbs + Post) \label{subfig:MCMC_est_f1_phifix}]{%
    \includegraphics[width=0.25\textwidth]{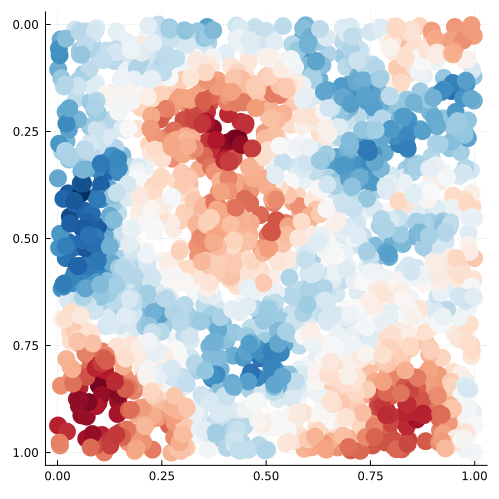}
  }
  \subfloat[\centering Estimated $f_1$ (ProjMC\textsuperscript{2})\label{subfig:PMCMC_est_f1_fixphi}]{%
    \includegraphics[width=0.25\textwidth]{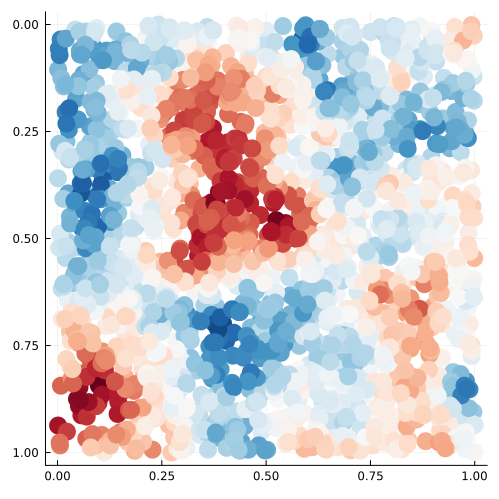}
  }\\[1ex]
  \subfloat[\centering True $f_2$ \label{subfig:MCMC_true_f2_fixphi}]{%
    \includegraphics[width=0.25\textwidth]{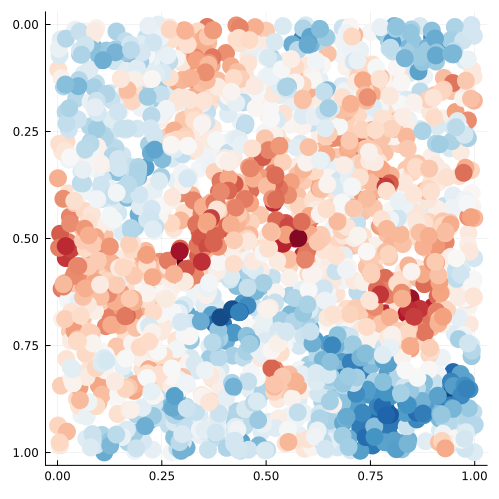}
  }
  \subfloat[\centering Estimated $f_2$ (Gibbs + Post)\label{subfig:MCMC_est_f2_fixphi}]{%
    \includegraphics[width=0.25\textwidth]{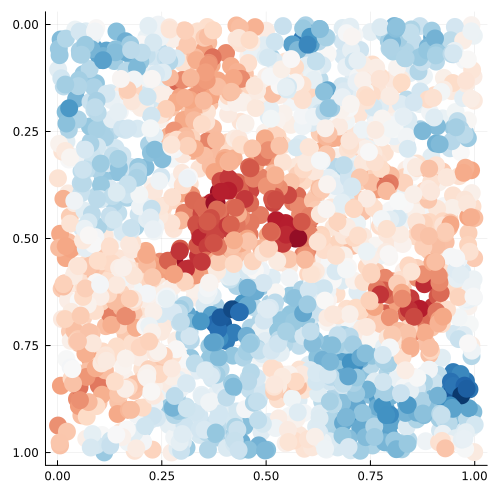}
  }
  \subfloat[\centering Estimated $f_2$ (ProjMC\textsuperscript{2})\label{subfig:PMCMC_est_f2_fixphi}]{%
    \includegraphics[width=0.25\textwidth]{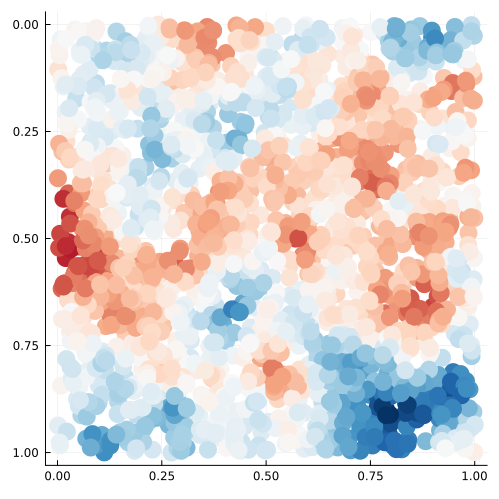}
  }
  \caption{Scatter plots of the true and estimated posterior means for the two latent spatial factors, $\fb_1$ (top row) and $\fb_2$ (bottom row), from the second simulation study. Dot locations indicate spatial positions, and colors represent latent factor values. Columns correspond to the true factors, Gibbs sampler with post-processing, and the proposed ProjMC\textsuperscript{2} method. All factor estimates are centered at zero and rescaled to lie on the sphere of radius $\sqrt{n-1}$, and results for the same factor share a common color scale for visual comparison.}
  \label{fig:f_compare_fixphi}
\end{figure}

Quantitative assessment of inference accuracy is reported in Table~\ref{tab:latent_diag_compare_fixphi}. The two algorithms performed comparably, with Gibbs+Post yielding slightly more accurate point estimates, as indicated by marginally lower Euclidean distances. It should be noted, however, that the posterior distributions approximated by the two methods differ inherently. Consequently, the observed accuracy and variance metrics reflect a complex interplay among the flexibility of the posterior representation, the algorithm's convergence behavior, and its mixing efficiency during posterior sampling.

\begin{table}[ht]
\centering
\begin{tabular}{lcccc}
\toprule
\textbf{Latent Factor} & 
\multicolumn{2}{c}{\textbf{Gibbs + Post}} & 
\multicolumn{2}{c}{\textbf{ProjMC\textsuperscript{2}}} \\
\cmidrule(lr){2-3} \cmidrule(lr){4-5}
& Euclidean Dist. & Sphere Var. & Euclidean Dist. & Sphere Var. \\
\midrule
$f_1$ & 17.37 & 254.8 & 18.70 & 97.5 \\
$f_2$ & 22.18 & 252.0 & 24.66 & 276.1 \\
\bottomrule
\end{tabular}
\caption{Comparison of posterior summaries for the two latent spatial factors ($\fb_1$ and $\fb_2$) across two methods: Gibbs sampler with post-processing (Gibbs + Post) and the proposed Projected MCMC (ProjMC\textsuperscript{2}). Each method is evaluated using two diagnostics: Euclidean distance between the true and estimated factor, and the spherical variance of the posterior samples.}
\label{tab:latent_diag_compare_fixphi}
\end{table}

Finally, the posterior inference for the identifiable model parameters are presented alongside the true parameter values in Table~\ref{tab: sim_infer_sum_fixphi}. Similar to study one, table~\ref{tab: sim_infer_sum_fixphi} reveals that the posterior inference for the regression coefficients $\bbeta$ is almost indistinguishable between the two methods. For the noise variance parameters (diagonal elements of $\bSigma$), the estimates derived from both algorithms are similarly accurate and closely aligned with the true values. Nonetheless, a consistent pattern emerges where ProjMC\textsuperscript{2} produces marginally higher posterior mean estimates for these variances compared to the Gibbs+Post sampler. This slight discrepancy may be attributable to the more constrained sampling space imposed in ProjMC\textsuperscript{2}; limiting the variation captured by the latent factors might necessitate attributing a slightly larger proportion of the residual variance to the noise term $\bSigma$. In summary, both algorithms demonstrate robust and highly comparable performance in estimating the identifiable parameters of the model.

\begin{table}[ht]
\centering
\setlength{\tabcolsep}{3pt}  
\begin{tabular}{cccccccccccc}
\toprule
 &  & 
\multicolumn{2}{c}{\textbf{Gibbs + Post}} & 
\multicolumn{2}{c}{\textbf{ProjMC\textsuperscript{2}}} &
 &  & 
\multicolumn{2}{c}{\textbf{Gibbs + Post}} & 
\multicolumn{2}{c}{\textbf{ProjMC\textsuperscript{2}}}\\
\cmidrule(lr){3-4} \cmidrule(lr){5-6} \cmidrule(lr){9-10} \cmidrule(lr){11-12}
& & mean & 95\%CI &  mean & 95\%CI & & & mean & 95\%CI &  mean & 95\%CI\\
\midrule
$\bbeta_{[1, 1]}$ & 1.0 & 0.95 & (0.88, 1.02) & 0.95 & (0.88, 1.01) &$\bbeta_{[1, 6]}$ & -1.5 & -1.54 & (-1.69, -1.4) & -1.54 & (-1.68, -1.4) \\
$\bbeta_{[1, 2]}$ & -1.0 & -1.02 & (-1.11, -0.93) & -1.02 & (-1.11, -0.93) &$\bbeta_{[1, 7]}$ & 0.5 & 0.47 & (0.3, 0.64) & 0.47 & (0.31, 0.64) \\
$\bbeta_{[1, 3]}$ & 1.0 & 1.02 & (0.96, 1.08) & 1.02 & (0.96, 1.07) &$\bbeta_{[1, 8]}$ & 0.3 & 0.28 & (0.21, 0.35) & 0.28 & (0.21, 0.34) \\
$\bbeta_{[1, 4]}$ & -0.5 & -0.5 & (-0.63, -0.37) & -0.51 & (-0.64, -0.37) &$\bbeta_{[1, 9]}$ & -2.0 & -2.03 & (-2.14, -1.92) & -2.03 & (-2.14, -1.92) \\
$\bbeta_{[1, 5]}$ & 2.0 & 1.99 & (1.93, 2.05) & 1.99 & (1.94, 2.05) &$\bbeta_{[1, 10]}$ & 1.5 & 1.5 & (1.43, 1.56) & 1.5 & (1.43, 1.56) \\
\hline
$\bbeta_{[2, 1]}$ & -3.0 & -2.9 & (-3.02, -2.78) & -2.9 & (-3.01, -2.78) &$\bbeta_{[2, 6]}$ & 3.0 & 3.05 & (2.8, 3.31) & 3.05 & (2.8, 3.3) \\
$\bbeta_{[2, 2]}$ & 2.0 & 2.02 & (1.87, 2.17) & 2.02 & (1.87, 2.17) &$\bbeta_{[2, 7]}$ & 4.0 & 4.01 & (3.71, 4.3) & 4.0 & (3.71, 4.29) \\
$\bbeta_{[2, 3]}$ & 2.0 & 2.01 & (1.91, 2.11) & 2.01 & (1.91, 2.11) &$\bbeta_{[2, 8]}$& -2.5 & -2.49 & (-2.61, -2.37) & -2.49 & (-2.61, -2.37) \\
$\bbeta_{[2, 4]}$ & -1.0 & -1.02 & (-1.25, -0.8) & -1.02 & (-1.25, -0.79) &$\bbeta_{[2, 9]}$ & 5.0 & 5.07 & (4.88, 5.26) & 5.07 & (4.88, 5.26) \\
$\bbeta_{[2, 5]}$ & -4.0 & -3.98 & (-4.08, -3.88) & -3.98 & (-4.08, -3.88) &$\bbeta_{[2, 10]}$ & -3.0 & -2.98 & (-3.1, -2.87) & -2.99 & (-3.1, -2.87) \\
\hline
$\bSigma_{[1, 1]}$ & 0.5 & 0.51 & (0.48, 0.55) & 0.52 & (0.49, 0.56) &$\bSigma_{[6, 6]}$ & 2.5 & 2.61 & (2.45, 2.78) & 2.62 & (2.46, 2.8) \\
$\bSigma_{[2, 2]}$ & 1.0 & 1.01 & (0.95, 1.08) & 1.02 & (0.95, 1.08) &$\bSigma_{[7, 7]}$ & 3.5 & 3.57 & (3.35, 3.8) & 3.57 & (3.36, 3.8) \\
$\bSigma_{[3, 3]}$ & 0.4 & 0.41 & (0.39, 0.44) & 0.41 & (0.39, 0.44) &$\bSigma_{[8, 8]}$ & 0.45 & 0.46 & (0.43, 0.5) & 0.49 & (0.45, 0.52) \\
$\bSigma_{[4, 4]}$ & 2.0 & 2.12 & (1.99, 2.26) & 2.13 & (2.0, 2.27) &$\bSigma_{[9, 9]}$ & 1.5 & 1.57 & (1.47, 1.67) & 1.58 & (1.48, 1.68) \\
$\bSigma_{[5, 5]}$ & 0.3 & 0.31 & (0.29, 0.34) & 0.33 & (0.31, 0.36) &$\bSigma_{[10, 10]}$ & 0.5 & 0.5 & (0.47, 0.54) & 0.51 & (0.48, 0.55) \\
\bottomrule
\end{tabular}
\caption{Posterior inference for identifiable model parameters. Comparison of posterior means and 95\% credible intervals for regression coefficients ($\bbeta$) and noise variances (diagonal elements of $\bSigma$) obtained using the Gibbs sampler with post-processing (Gibbs+Post) and the proposed ProjMC\textsuperscript{2} algorithm, referenced against the true parameter values.}
\label{tab: sim_infer_sum_fixphi}
\end{table}

\label{lastpage}

\end{document}